%% file: main.tex
\documentclass[11pt,a4paper,oneside]{article}

\usepackage{geometry}
\usepackage[utf8]{inputenc}
\usepackage{setspace}

 \usepackage{caption}
\usepackage{graphicx}
\usepackage{epstopdf}
\usepackage{amsmath}
\usepackage[official]{eurosym}
\usepackage{amsthm}
\usepackage{amssymb}
\usepackage{tikz}
\usetikzlibrary{automata,arrows,positioning,calc}
\usepackage{amsthm}
\usepackage{array}
\usepackage{multirow}
\usepackage{bbm}
\usepackage{rotating}
\usepackage{threeparttable}
\usepackage[round]{natbib}   
\theoremstyle{plain}

\newtheorem{lemma}{Lemma}
\theoremstyle{definition}

\usepackage{verbatim}
\usepackage[verbose]{placeins}
\usepackage{tikz}
\usepackage{amsmath}
\usepackage{istgame}
\usepackage{url}
\usepackage{float}
\usepackage{booktabs}
\usepackage{subcaption}
\theoremstyle{plain}
\usepackage{xcolor}
\usepackage[
bookmarksopen,
bookmarksdepth=2,
colorlinks=false,
citecolor    = blue,
urlcolor=blue]{hyperref}

\title{Shrouded sin taxes\thanks{I am grateful for helpful comments and useful suggestions from Youssef Benzarti, Andreas Born, Fabrizio Colella, Markus Dertwinkel-Kalt, Markus Eyting, Gyozo Gyöngyösi, Simon Heß, Aljoscha Janssen, Christine Laudenbach, Alexander Ludwig, Vincenzo Pezone, James Reade, Emmanuel Saez, Patrick Schmidt, Lorenzo Schönleber, Johannes Spinnewijn, Tomasz Sulko, Dmitry Taubinsky, Alfons Weichenrieder and seminar participants at the Goethe University Frankfurt, University of Reading, Leibniz Institute for Financial Research SAFE, and the WHU Young Scholar Tax Conference. Financial support from the Leibniz Institute of Financial research SAFE is gratefully acknowledged. }}
\author{Johannes Kasinger\thanks{Tilburg School of Economics \& Management, Tilburg University \& Leibniz Institute for Financial Research SAFE, j.kasinger@tilburguniversity.edu}}
\date{November 2023}

\begin{document}

\maketitle
\centerline{Latest version available \href{https://1drv.ms/b/s!Ap4Fc3AbgunUxl7P8ryllgVC1MXi?e=yyYsa8}{here}.}

\vfill

\begin{abstract}
\noindent  
Strategic shrouding of taxes by profit-maximizing firms can impair the effectiveness of corrective taxes. 
This paper explores tax shrouding and its consequences after the introduction of a digital sin tax designed to discourage harmful overconsumption of online sports betting in Germany. 
In response to the tax reform, most firms strategically shroud the tax, i.e., exclude tax surcharges from posted prices.
Using an extensive novel panel data set on online betting odds, I causally estimate the effect of the tax on consumer betting prices.
Consumers bear, on average, 76\%  of the tax burden.
There is considerable and long-lasting heterogeneity in effects conditional on shrouding practices. Firms that shroud taxes can pass 90\% of the tax onto consumers, while the pass-through rate is 16\% for firms that directly post tax-inclusive prices.
To understand the results' underlying mechanisms and policy implications, I propose an optimal corrective taxation model where oligopolistic firms compete on base prices and can shroud additive taxes.
Tax shrouding is only attainable in equilibrium if (some) consumers underreact to shrouded attributes. 
According to the theoretical predictions, the empirically identified heterogeneity suggests that strategic tax shrouding significantly attenuates the positive corrective welfare effects of the tax.
The results prompt regulating shrouding practices in the context of corrective taxation.\\

\bigskip
\end{abstract}



\vfill

\vfill

\thispagestyle{empty}

\newpage
\section{Introduction}\label{sec:intro}
\setcounter{page}{1}

\input{Chapters/Introduction}

\section{Background}

\subsection{Institutional background}

\input{Chapters/Institutional_Background_main}

\subsection{The price of a bet and tax shrouding}\label{sec:bet_primer}

\input{Chapters/Background_Betting}

\section{Data}\label{sec:sport_data}

\input{Chapters/Data}

\section{Methodology}

\input{Chapters/Methodology}

\section{Results}\label{sec:sport_results}

\input{Chapters/Results}

\input{Chapters/Theory}

\section{Conclusion}

\input{Chapters/Discussion}


\newpage

\pagebreak

\clearpage

\bibliographystyle{plainnat}
\bibliography{diss_bib}

\newpage

\appendix

\setcounter{page}{1}
\clearpage

\input{Chapters/Background_Regulation}

\input{Chapters/AppendixA}

\newpage
\section{Appendix figures}

\input{Figures_tex/Appendix_figures}

\clearpage
\newpage

\section{Appendix tables}

\input{Tables/AppendixTables}

\pagebreak

\input{Chapters/Standard_tax_incidence}

\section{Lemmas and Proofs}

\input{Chapters/proofs}

\end{document}

%% file: Chapters/Introduction.tex
Consumers often underreact to shrouded (i.e., non-salient or hidden) attributes of goods \citep{dellavigna2009psychology}. There are numerous examples where firms strategically shroud relevant price attributes to exploit these underreactions.\footnote{Examples of such practices include add-on pricing \citep[e.g.,][]{ellison2005model, gabaix2006}, hidden shipping and handling costs \citep[e.g.,][]{hossain2006plus, brown2010}, strategic price complexity \citep{carlin2009strategic} and complicated or confusing product descriptions \citep[e.g.,][]{ellison2009search, chioveanu2013price}.}
When firms can choose how to present and collect taxes, they may optimally shroud those taxes, e.g., by actively excluding tax surcharges from posted prices as frequently observed in the aviation sector \citep{bradley2020hidden}. 
Tax shrouding decreases the salience of taxes, which attenuates the behavioral responses to those taxes and thus influences their welfare effects \citep{chetty2009salience, taubinsky2018attention}. 
 Various theoretical studies on add-on pricing and strategic price complexity \citep[e.g.][]{gabaix2006, carlin2009strategic} further suggest that responses to shrouded attributes and equilibrium outcomes can substantially differ when salience is not a passive but an active decision by firms, depending on consumers' attention to shrouded price information (and the implied firms' motives).\footnote{For instance, Bayesian consumers who infer that products with strategically shrouded add-on prices are likely to be less attractive offers may or may not lead to unshrouding in equilibrium depending on the market structure \citep{gabaix2006}.} 
 
 Consequently, strategic tax shrouding and the underlying market structure have important welfare implications. These implications are essential for corrective taxation because its effectiveness in discouraging the harmful overconsumption of certain goods directly depends on (perceived) tax-induced price changes \citep{bernheim2018behavioral}.\footnote{From a social welfare perspective, goods can be overconsumed if individuals do not fully account for the negative effects of their consumption on others in society---externalities---or on their (future) selves---internalities \citep{allcott2019should}. Typical examples of such goods include environmentally harmful goods or sin goods, such as cigarettes, alcohol, or gambling. Corrective taxes are a central policy tool to limit their negative effects on social welfare.} 
Nevertheless, strategic tax shrouding and its consequences are understudied. The literature typically assumes that consumers optimize fully with respect to tax-inclusive prices, which rules out the relevance of tax shrouding. The tax salience literature, in contrast, relaxes this assumption but almost exclusively focuses on settings where tax salience is exogenous and independent of firms' strategic decisions. 

This article fills this gap by examining the prevalence, effects, and welfare implications of strategic tax shrouding by firms in the context of a corrective sin tax on (online) sports betting.
I document that tax shrouding is a widespread response by firms to the introduction of the 2012 German sports betting tax. 
The tax imposes a 5\% duty on betting turnovers generated by German customers and has to be remitted by all betting agencies that serve German customers, irrespective of their jurisdictional location.\footnote{Note that the average expected net return of a (random) 1\euro{} bet from a bookmaker's perspective is around 7\%, implying a considerable levy on firms' profits.}
By exploiting the induced quasi-experimental variation and a novel extensive panel data set on online betting prices, I estimate the pass-through rates of the tax onto consumers. The paper employs a difference-in-differences (DID) framework, comparing the changes in average consumer betting prices in the German market with changes outside the German market. 
Firms can pass around 76\% of the tax onto consumers, i.e., tax-inclusive consumer prices increased on average by about 3.8 percentage points after the tax reform. Estimated pass-through rates are heterogeneous depending on firms' shrouding practices, which is inconsistent with predictions of standard taxation models \citep{kotlikoff1987tax, weyl2013pass}. 
Firms that actively exclude taxes from posted prices can pass 90\% of the tax onto consumers, while the pass-through rate is 16\% for firms that directly post tax-inclusive prices.\footnote{Note that the pass-through rates refer to the effective consumer prices, i.e., posted prices plus shrouded tax surcharges. Accordingly, a pass-through rate of 90\% for shrouding agencies implies a decrease in posted prices of 10\%, given that the average applied shrouded tax surcharge is equal to the size of the tax.} 

To better understand the underlying mechanisms and policy implications, I propose a stylized sin tax model based on \citet{odonoghue2006optimal} and \citet{varian1980model} that allows for strategic tax shrouding and heterogeneously attentive consumers. According to the model, the empirical results imply that several consumers underreact to shrouded taxes.
The positive corrective effects of the sin tax are undermined by the strategic shrouding behavior of profit-maximizing firms. 
In addition, the employed homogeneous linear sin tax can only effectively implement the first-best outcome when attention to shrouded attributes is homogeneous across consumers. 
To ensure the effectiveness of the sin tax, policymakers should require firms to post tax-inclusive prices. While the study focuses on a sin tax, the policy implications also extend to the optimal design of externality correction policies, more generally.

The German sports betting tax provides a suitable quasi-experimental setting to study the effectiveness and strategic shrouding of sin taxes in a digital context. Importantly, the vast majority of betting agencies that target the German market followed the tax rules despite a missing jurisdiction in Germany. The tax effectively generated considerable tax revenues, which also has important implications for the functioning of digital service taxes.\footnote{Using data from a commercial provider that also serves the betting agencies and official bodies, I find that around 80\% and 90\% of online sports betting revenue are taxed. Until 2020, the German sports betting tax generated revenues of more than 2.6bn Euro according to administrative tax data, making it a significant source of tax revenue.} The primary rationale behind German sports betting regulation and taxation is to prevent betting addiction and problem gambling (German interstate gaming treaty §1). The tax rules do not restrict how and whether to present or collect the tax. In contrast to most other countries that deduct betting taxes from gross betting revenues, i.e., the total wagered amount minus all winnings by bettors \citep{haucap2021glucksspielstaatsvertrag}, Germany bases the sports betting tax on betting turnover. This peculiarity implies that the German tax is independent of a bet's price, analogous to a classical per unit sin tax (see section \ref{sec:bet_primer}). Moreover, there was no other effective taxation of gambling services before and after the reform, and the regulatory landscape in Germany for online sports betting remained essentially unchanged between 2009 and 2018 besides the tax reform. 

The data set combines web-scraped information on betting odds from 68 international betting agencies for more than 80,000 sports events between 2009 and 2018 with information from various other sources. The data set allows to infer implicit betting prices and identify heterogeneity in tax effects conditional on strategic tax shrouding by firms. 
I focus on the pass-through of the tax, which is a central measure to evaluate the policy’s corrective and welfare effects \citep{weyl2013pass}. The estimated average effective pass-through rate is stable over time and shows a slightly increasing trend. 
The betting prices before the tax reform develop very similar within and outside the German market. Treatment timing is homogeneous and plausibly exogenous to betting prices. Accordingly, the identifying DID assumptions hold and effects are not biased because of heterogeneous treatment timing \citep{goodman2021difference, sun2021estimating, borusyak2021revisiting}, allowing for a causal interpretation of the average effects.
The empirical results are robust to different specifications and more restrictive categorizations of the control and treatment groups. 


In response to the tax, the majority, but not all, betting agencies employ a shrouding policy that drives a wedge between posted and effective tax-inclusive betting prices.
Shrouding agencies advertise tax-exclusive prices but deduct taxes after a specific bet is selected---similar to hidden shipping surcharges. Non-shrouding agencies post tax-inclusive prices directly. The agencies introduced shrouding policies on different dates. All agencies that established a shrouding policy kept that policy until the end of the observation period.\footnote{Two out of ten agencies targeting the German market did not shroud taxes over the entire observation period. Six agencies started with tax shrouding within six months after the reform.}  I estimate heterogeneous pass-through rates by two approaches using a DID design: i) subsample analyses where I only consider a subset of German agencies depending on the established shrouding policy; ii) a framework where I interact ``treatment'' with a dummy variable that indicates whether a shrouding policy was in place for a given event and bookmaker.
While shrouding policies are endogenous decisions by firms, average prices and trends before the tax reform were similar for shrouding and non-shrouding agencies, suggesting that non-German agencies still provide a valid counterfactual for the subsample analyses. The estimated difference in pass-through rates using approach ii) is around 80\%. This difference is slightly higher than for the subsample analyses as average effects are not diluted by ``not yet shrouded'' events.\footnote{This result is confirmed by the assessment of one of the main competitors in their 2012 annual report, \citet{bwin2013annual}:  ``Following the introduction of a 5\% turnover tax on sports betting in Germany on 1 July 2012, along with the rest of the industry, we began to pass on the majority of the tax to customers by requiring that winning customers pay a 5\% withholding on their pay-out''.} The heterogeneity persists until the end of the observation period, even if the gap decreases over time.
Furthermore, the difference persistence of this difference implies that firms cannot profitably ``debias'' competitors' consumers \citep{gabaix2006}. 


The theoretical model I am proposing has two extensions to the standard optimal sin tax model. First, the sin good is not supplied under perfect competition but under imperfect Bertrand-Nash price competition in the spirit of \citet{varian1980model}. 
Firms set salient base prices and can shroud additive taxes that apply to the sin good. Second, consumers can perfectly observe the base prices but are heterogeneous regarding their attention toward the shrouded tax.\footnote{To keep the model tractable, I assume a static model where consumers are either fully attentive or completely inattentive to shrouded taxes and homogeneous otherwise.} 
Suppose all consumers are fully attentive to the tax surcharge. In that case, equilibrium prices are symmetric and equal to firms' marginal costs, implying a full pass-through of taxes onto consumers. The implications of classical corrective taxation models \citep{pigou1920economics, odonoghue2006optimal} carry over to the extended model: (i) the sin good is overconsumed compared to the welfare-optimal solution because consumers do not fully internalize costs on their (future) selves when making their decision due to self-control problems or other behavioral mistakes; ii) a corrective per-unit sin tax equal to the average mistake of the marginal consumer restores the welfare-optimal solution. iii) in the case of a homogeneous underreaction to taxes \citep{chetty2009salience, dellavigna2009psychology}, the optimal tax rate increases and equals the average marginal mistake divided by tax salience or attention parameter $\theta$ \citep{farhi2020optimal}.\footnote{Note that as consumers are still homogeneous, Bertrand competition will push firms' profits to zero. However, the corrective effect of the tax decreases.}

If (some) consumers are inattentive to shrouded taxes, (some) firms will optimally exploit this inattention and shroud taxes. Consumers overconsume the sin good in equilibrium, and sin taxes are less effective (or even completely ineffective) in restoring the welfare-optimal solution. If attention to shrouded taxes is heterogeneous, a linear corrective tax-transfer scheme that applies equivalently to all consumers is ineffective in restoring the first-best welfare-optimal solution. The second-best sin tax rate depends on the distribution of consumers' attention to shrouded surcharges. The government faces a tradeoff between ``overcorrecting'' attentive consumers and ``undercorrecting'' the consumption of inattentive consumers. However, banning tax shrouding practices is one simple and effective policy that restores the effectiveness of sin taxes in implementing the welfare-optimal solution independent of the fraction of attentive consumers. 

In the case where self-control problems are negatively correlated with attention (e.g. due to impulsive gambling behavior), the adverse welfare effects of tax shrouding are more pronounced as the individuals that would benefit from the sin tax the most react to it the least. Similarly, a correlation between attention and income can affect the regressivity of the tax and, thus, its welfare implication depending on the correlation's sign \citep[on the welfare implications of the regressivity of sin taxes see][]{gruber2004tax, goldin2013smoke, allcott2019regressive}.\footnote{The theoretical results imply that strategic shrouding behavior by firms undermines the effectiveness of (positive) corrective taxation. However, in the case of subsidies, say, on environmentally-friendly products, the motives of the government and firms align. It is optimal for firms to boost the subsidy salience to make their offers more attractive---amplifying the effectiveness of subsidy-based externality correction policies. Consequently, corrective subsidies may be preferable to taxes in settings where strategic shrouding by firms is significant for consumers' behavioral responses.} 


Following the theoretical implications, the observed pricing strategies and pass-through rates in the empirical part of the paper are only sustainable if (some) consumers underreact to shrouded tax surcharges. Consequently, the empirical results indicate that the average experienced increase in consumer prices is substantially higher than the perceived increase in prices. 
This implies that strategic tax shrouding considerably decreases the positive corrective welfare effect of the German sports betting tax. 
While these results are indicative of consumption adaptions by bettors, the betting revenue data is not granular enough to make a conclusive statement about the effect on actual consumption. However, overall betting revenues in Germany strongly increased and developed very similarly to other European countries before and after the tax reform, which suggests a limited corrective effect of the betting tax. Requiring agencies to directly post tax-inclusive odds can help restore the corrective effect. \citet{bradley2020hidden} provide evidence of the effectiveness of such a policy in the context of US air carriers. They find, in line with my empirical results and theoretical predictions, that introducing full-fare advertising rules, which required airlines and travel agents in the US to advertise tax-inclusive prices on their website, led to a decrease in consumers' tax incidence and demand.

My study contributes to different strands of economic literature. First, numerous behavioral public finance studies provide compelling evidence on the importance of behavioral frictions for markets and optimal economic policy \citep[e.g.][]{chetty2009salience, Finkelstein2009, sallee2011surprising, goldin2013smoke, allcott2015, Feldmann2015, taubinsky2018attention, benzarti2020goes, farhi2020optimal, bradley2020hidden, kroft2020salience}. 
In their seminal paper on tax salience, \citet{chetty2009salience} note that the tax incidence on consumers is lower when ``producers must actively ``shroud'' a tax levied on them in order to reduce its salience''. Yet, theoretical and empirical evidence on this channel is scarce. My study closes this gap in the context of a digital sin tax, bridging the behavioral public finance literature to a rich literature in industrial organization that examines the effect of shrouded attributes and strategic price obfuscation in different market settings \citep[e.g.][]{ellison2005model, gabaix2006, hossain2006plus,  carlin2009strategic, ellison2009search, brown2010, ellison2012search, piccione2012price, kosfeld2017add, heidhues2017}. 

Second, there is a growing literature in public finance, industrial organization and marketing that examines the pass-through and the effects of taxes on various traditional (physical) sin goods, such as sugar-sweetened beverages \citep{allcott2019regressive, allcott2019should, seiler2021impact, keller2023express}, alcohol \citep{kenkel2005alcohol, hindriks2019heterogeneity}, cigarettes \citep{barnett1995oligopoly, gruber2004tax, harding2012heterogeneous}, and marijuana \citep{hollenbeck2021taxation}. Some of these studies show full pass-through of these taxes, or even over-shifting---one possible outcome in models with imperfect competition \citep{anderson2001tax, weyl2013pass}.\footnote{There is also a renewed interest among industrial economists in studying tax pass-through to learn more about consumer preferences, market power, and competition in different markets \citep[e.g.][]{miravete2018market, pless2019pass, miravete2020one, hollenbeck2021taxation}.} 
My findings align with \citet{decicca2013pays}, who find that consumer price search affects the pass-through of cigarette excise taxes. I contribute to this literature by showing that firms can affect ``experienced'' and ``perceived'' pass-through rates and the corrective effect of sin taxes by strategic tax shrouding.
Moreover, to the best of my knowledge, this is the first study that examines the pass-through of sin taxes on online sports betting prices. Sin taxes on sports betting can be expected to become increasingly policy-relevant in the near future, given the recent legalization of sports betting in the US \citep{Supreme2017} and the persistent unprecedentedly growth in global sports betting revenues \citep{montone2021optimal}.

Third, my study builds on and contributes to the literature on taxation and imperfect competition \citep{stern1987effects, delipalla1992comparison, anderson2001tax, weyl2013pass}. In a recent paper, \citet{kroft2020salience} extend the standard tax salience model to imperfect competition. They find that tax salience and market structure interact so that higher tax salience can increase tax incidence under certain conditions. My study adds to their comprehensive analysis and related studies \citep{bradley2020hidden} by relaxing the assumption of exogenous tax salience in a price competition model inspired by \citet{varian1980model} with an application to corrective taxation.

The remainder of the paper is organized as follows. Section 2 gives more background on the tax reform, online sports betting markets, and how shrouding affects the difference between posted and effective prices. Sections 3 and 4 discuss the data and empirical strategy. Section 5 presents the empirical results. Section 6 lays out the theoretical model and discusses its policy implications. Section 7 concludes.

%% file: Chapters/Institutional_Background_main.tex
This section gives an overview of the relevant institutional details of the German sports betting tax reform and regulations of the European and German online sports betting market. It further presents some descriptive statistics on the development of online betting revenues in Europe to put the study into perspective. Appendix \ref{sec:app_inst_details} describes the underlying regulatory framework in detail.

\subsubsection*{The German sports betting tax and online sports betting regulation}\label{sec:sports_tax}

In July 2012, Germany extended the geographical scope of its sports betting tax code, effectively introducing a new 5\% tax on online sports betting turnovers generated by German bettors.\footnote{According to the Racing Bets and Lottery Tax Act (RBLTA, §17), the term turnover includes all expenses incurred by the bettor to participate in the wager.}  Betting agencies must remit the tax on aggregated monthly turnovers, irrespective of their jurisdictional location and whether the provider operates lawfully in Germany. There is no regulation on how or whether to display or collect the tax from consumers. According to §1 of the German interstate Gaming treaty (``Glücksspielstaatsvertrag''), the main goal behind regulating and taxing online sports betting is to prevent problem gambling and betting addiction.

Before the reform, online sports betting generated virtually no tax revenue because the preceding \emph{Sports Bet Tax} only applied to bets that were "organized" in Germany. According to German tax law, the organizer of a bet is defined as the entity that fixes the odds and not the entity that brokers a bet. Consequently, betting agencies, even if established or with some physical presence in Germany, avoided taxation and the national sports betting ban by simply operating from liberal foreign EU entities, mostly from Gibraltar and Malta \citep{englisch2013taxation}.  The tax reform in 2012 extended the geographical scope of the tax code to bettors, making the Sports Bet Tax applicable if the bettor \emph{or} the organizer of a bet is domiciled in Germany.\footnote{Domiciled in Germany means that the bettor has a broadly defined residence or his habitual abode in Germany. If a German resident has placed the wager on an occasional stay outside of Germany, the tax does not apply. \citep{englisch2013taxation}.} Accordingly, all online betting agencies that provide their services to German customers became liable for the tax \citep{englisch2013taxation}. So, while the statutory tax rate actually decreased from 16.7\% to 5\%, the effective tax rate for online sports betting increased from zero to 5\%. In this study, we thus refer to a newly introduced sports betting tax. 

Before and after the reform, betting services are exempted from the German VAT if they are subject to the Sports Bet Tax \citep{englisch2013taxation}. The Sports Bet Tax does further take precedence over federal and state duties. Moreover, international betting agencies are usually not liable to business taxation as they operate from low-tax foreign entities, implying that the turnover tax was and is the only effective taxation of online sports betting before and after the reform. On the consumer side, there are no additional taxes in Germany, as gambling winnings, including profits from sports betting, are excluded from the personal income tax.

Except for the online sports betting tax, the regulatory landscape for online sports betting in Germany was in abeyance between 2008 and 2018 (see Appendix \ref{sec:app_inst_details} for details).\footnote{ A statement in the 2009 annual report of the betting provider \textit{sporting bet} summarizes the industry view (and strategy) on the German gambling regulation quite fittingly: "Furthermore, enforcement action against operators where they actively target German residents (including through local marketing) has been curbed due to the lack of clarity in the legal position. In our view, therefore, legislation that was intended to almost comprehensively block online gambling has had only limited effect and the general inability of the German government to block online gambling websites, coupled with the questionable legality of the legislation, has led to a continued supply of online gambling services, and an absence of extra-territorial enforcement against the activity."}  The German regulatory framework for online gambling and betting was characterized by a division of tax and regulatory competencies between the federal and the state level. This division of competencies led to opaque and ineffective legislation, which was at variance with European law \citep{eugh2010} and even German constitutional law. In this regulatory environment, betting agencies, domiciled in liberal (primarily European) jurisdictions outside Germany where their services were admissible, could serve the German demand for sports betting in a semi-legal (due to the free movement of services in the EU) and more or less unrestricted grey market tolerated by German authorities \citep{rebeggiani2017neue}. The regulatory framework effectively changed only in 2021. 

The German case provides an ideal (quasi-experimental) setup to study the incidence of a sports betting sin tax. First, the tax was effective in raising considerable considerable tax revenues and the vast majority of betting agencies that target the German market followed the tax rules despite a missing jurisdiction in Germany (see Section \ref{sec:revenues_sport}). Second, Germany taxes the betting turnover, which makes the tax analogous to a classical excise sin tax. In contrast, the vast majority of countries levy taxes on gross betting revenues, i.e., the total wagered amount minus all winnings by bettors, \citep{haucap2021glucksspielstaatsvertrag}. Note that a 5\% tax on turnover constitutes and considerable levy, given that the average expected gross betting revenue per \euro{} wagered on a single game is around 0.07 Euro cents (see section \ref{sec:sport_data}). Third, there was no other effective taxation of gambling services besides the turnover tax before and after the reform. So, the estimation of the tax incidence is not diluted by other taxes, such as corporate taxation or VATs. 

\subsubsection*{The rise of sports betting in Europe and Germany}

Thanks to the internet and deregulation, sports betting---and gambling in general---is now available to a broad audience, and the sector has grown unprecedentedly over the past years. The \citet{ec2012press} describes online gambling as "one of the fastest growing service activities in the EU". European online gross betting revenues, i.e., the total wagered amount minus all winnings by bettors, grew from around \euro{}2.4bn in 2008 to over \euro{}12.1bn in 2019 (H2 Gambling Capital 2021, see more details in section \ref{sec:revenues_sport}).\footnote{Revenues include onshore (i.e., regulated/taxed activity) and offshore (i.e., the revenue of operators targeting the market without holding a local license) revenues. Revenues on the completely unregulated black market are not included. The total wagered amount includes bonuses.}

As can be seen in Figure \ref{fig:bettingrevenues_countries}, German online betting revenues exhibit a similar trend to other European countries, growing from \euro{}206m in 2008 to \euro{}1,280m in 2019. European sports betting revenues account for around 36.3\% of global sports betting revenue in 2019, second to Asia with a share of 47.2\%. If only online sports betting services are considered, Europe had the largest share with over 50\% in revenues. In contrast, sports betting revenues in North America only account for 6.7\% of global sports betting revenues \citep{ibia2021}. However, after the legalization of sports betting by the \citeauthor{Supreme2017} in 2018, the sports betting revenues in the US seem to start following a similar trend to those in Europe \citep{forbes2022betting}.

		\begin{figure}[!ht]
		\caption{Total online gross betting revenues in different countries (in \euro{}M)}
		\centering	
	
		\includegraphics[width=\textwidth]{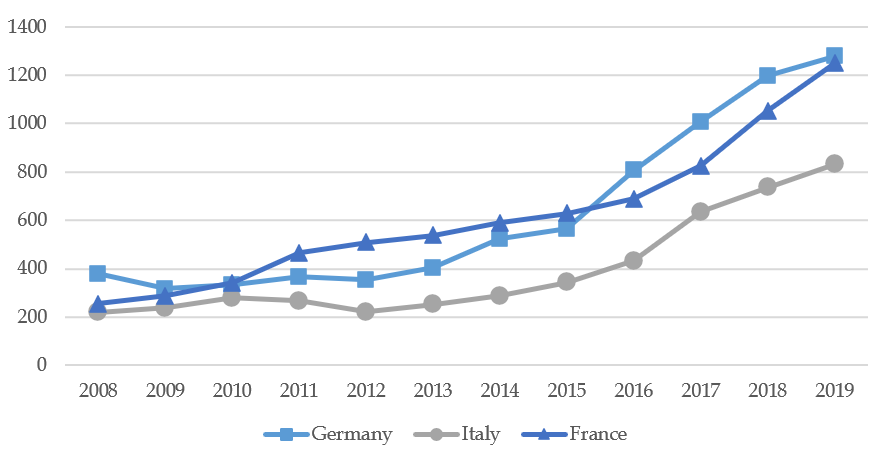}
		\begin{minipage}{15.8cm}  \footnotesize  \emph{Notes:} 
		This figure illustrates the annual total online gross betting revenues between 2008 and 2020 in Germany, Italy, and France (in \euro{}M). Gross betting revenues are equal to the total wagered amount (including bonuses) minus all winnings by bettors. Revenues include onshore (i.e., regulated/taxed activity) and offshore (i.e. revenue of operators targeting the market without holding a local license) revenues. Source: H2 Gambling Capital – July 2021.  
		\end{minipage}
		  \label{fig:bettingrevenues_countries}
	\end{figure}

%% file: Chapters/Background_Betting.tex
Bets constitute simple contingent claims on different outcomes of a game with a well-defined asset and state price (the inverse of the decimal odds), which gives betting markets a direct application to financial markets. For this reason, betting markets are well-suited to serve as a natural asset pricing laboratory to test market efficiency and rationality \citep[e.g.][]{figlewski1979subjective, camerer1989does, brown1993fundamentals, woodland1994market, levitt2004gambling, snowberg2010explaining, moskowitz2021asset}. This paper does not consider these issues but focuses on the pass-through of betting taxes on the event's betting margin charged by the bookmakers. 

In the context of this paper, bets can be best understood as recreational consumption goods  that come with a price equal to the expected net return of a (random) bet on a specific event from a bookmaker's perspective. Goods are very homogeneous across providers. 
A price of zero implies an actuarial fair bet with an expected return of zero. 
Throughout this paper, we consider fixed odds betting markets, where prices---in contrast to parimutuel markets---are fixed at the time of ``purchase''.
For each event $i$, there are $n$ mutually-exclusive outcomes ($s=1,2,...n$). In this paper, the outcomes are defined by the game's result: i) home team wins; ii) away team wins; or iii) draw.\footnote{In Basketball, Hockey, and American Football, usually, games never end with a draw as teams play overtime(s) until one team wins the game. For this reason, we concentrate on ``two-way'' markets for these sports that exclude a draw as a potential outcome. In the infrequent event of a draw (only possible in American Football), bettors' wagers are paid back.} Bookmaker $b$ promises bettors to pay $r_{isb}$ for every \$1 wagered if (in event $i$) outcome $s$ is realized and zero otherwise. $r_{isb}$ are the effective odds, including all potential tax surcharges (see section \ref{sec:bet_primer}).\footnote{I use the decimal odds notation in this paper that directly refers to this definition. The other common notations are fractional (or English) odds and moneyline (or American) odds. The odds can be easily converted and contain the same information.}  The odd $r_{isb}$ for an arbitrary outcome $s$ can be expressed as:
\begin{equation}
\label{eqn:singlemargin}
    r_{isb}=\dfrac{1}{\theta_{ib} \pi_{is}},
\end{equation}
where $\theta_{ib}$ is bookmaker $b$'s markup for event $i$. $\pi_{is}$ represents the (objective) likelihood of state $s$, which is unobserved. However, as we observe the odds for all outcomes of an event we can derive $\theta_{ib}$ and the implied probabilities of each outcome, which is defined as $\Bar{\pi}_{is}=\frac{1}{ \theta_{ib} r_{isb}}$. To do that, we rearrange Equation \ref{eqn:singlemargin} and sum over all possible states:\footnote{Note that in this step I implicitly assume that the bookmaker chooses markups independently of the type of state or its likelihood.}
\begin{equation*}
   \theta_{i, b} \underbrace{\sum_{s=1}^n\pi_{i, s}}_{=1} = \sum_{s=1}^n \dfrac{1}{r_{i,s,b}}
\end{equation*}
$\theta_{i, b}$ can be understood as the amount a bettor has to pay to receive a safe payout at the end of the game. Note that a safe payout of 1\$ can be secured by betting $\frac{1}{r_{i,s,b}}$ \$ on outcome $s$ of each possible outcome of an event $i$. In the case of unbiased odds, i.e., the odds give an unbiased prediction of the true probabilities, the expected and certain payouts equate. I define the consumer price of a one-dollar bet on event $i$ as:
\begin{equation}
\label{eqn:price}
   p_{ib}= 1- \dfrac{1}{\theta_{ib}},
\end{equation}
which gives the expected net return of a 1\$ bet from the bookmaker's perspective. A lower price makes a bet ``cheaper'' to the bettor and less profitable to the betting agency. This notion has the advantage that it is directly applicable to standard microeconomic theory. A tax on betting turnover is independent of the price and can be directly treated as a standard per unit excise tax. For instance, a 5\% tax on turnover decreases the expected net revenue of the betting agency per \$1 wagered (i.e., the price of a bet) by \$0.05.

Note that betting agencies sometimes use different marketing tools to attract consumers that may indirectly affect the overall attractiveness of different providers. Examples include cash bonuses when a bettor deposits money on their account for the first time, loyalty bonuses that are usually ``paid'' out as free bets, and ``odds boost'' that improve the effective odds if bettors combine a large number of bets. Most of these bonuses do not affect the betting prices of single events. Historical data on these programs are not available. Using data from \textit{top100bookmaker} that provide a snapshot of bonus programs for several agencies in the year 2021, bonus programs seem not to vary considerably between betting agencies, so I abstract from differences in bonuses in this paper.

\subsubsection*{Posted and effective betting prices}\label{sec:tax_shrouding}

On a website of a betting agency, the bettors can choose from numerous events in different leagues and for varying sports. Figure \ref{fig:odds_example} presents a typical example from a website of one of the main competitors. The list on the left-hand side represents different leagues the bettor can choose. In the main panel, one can find different events---Premier League games in our example---with respective odds for a \emph{home win}, \emph{draw}, or \emph{away win}. There are other betting types for the same game, such as the number of scored goals or spread bets appearing once a bettor clicks on a particular game. 

		\begin{figure}[!ht]
		\caption{Example of odds presentation on a typical website}
		\centering	
	
		\includegraphics[width=\textwidth]{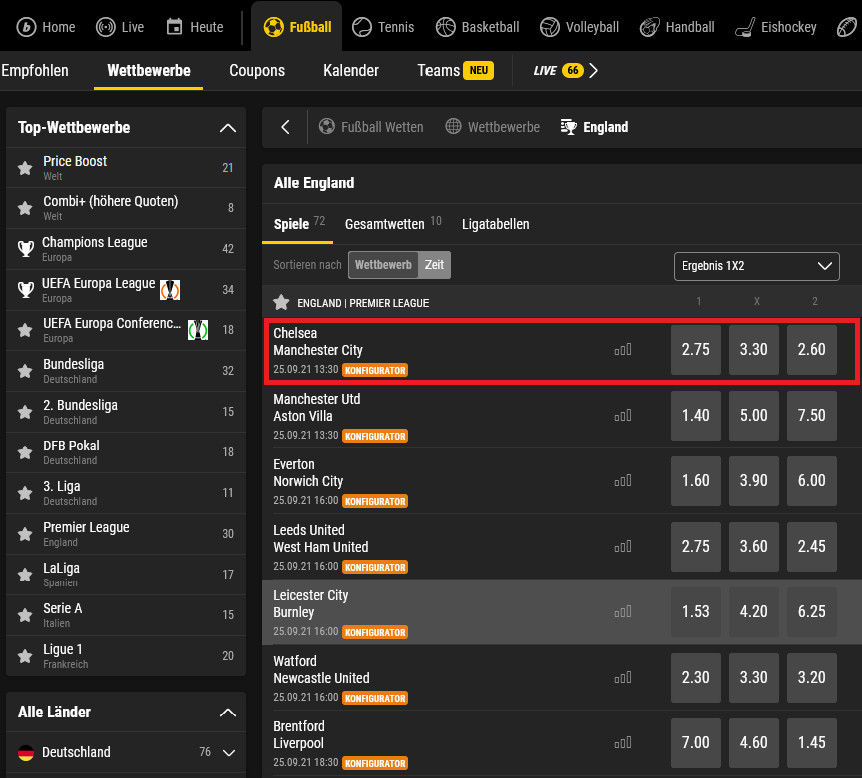}
		\begin{minipage}{15.8cm}  \footnotesize  \emph{Notes:} 
This figure presents an example of how tax-exclusive odds are displayed on a typical website of a betting agency (browser version) before the bettor chooses one of the outcomes. The list on the left represents the different Leagues a bettor can select. Outcome ``1'' represents a win of the home team, ``X'' a draw, and ``2'' a win of the away team. Source: Bwin.com  - September 2021. 
		\end{minipage}
		  \label{fig:odds_example}
	\end{figure}

Before bettors choose one particular bet, odds for different outcomes are always presented exclusive of a potential additional tax surcharge. All agencies refer to the surcharge as a ``tax'', similar to a value-added tax, which gives the impression that agencies must pass the tax onto consumers. In principle, there is, however, no direct link as the tax must be remitted by suppliers, irrespective of whether agencies are applying a surcharge. The tax authority does not directly receive the tax surcharge. In principle, the agencies can apply other surcharges common in other markets, such as in the aviation industry or the e-commerce sector more generally \citep{dertwinkel2020buy, bradley2020hidden}. However, there is no example of a betting agency applying additional surcharges. 

Let denote the posted, advertised surcharge-exclusive odds as $\Tilde{r}_{isb}$. $\Tilde{r}_{isb}$ may differ from the effective surcharge-inclusive odds $r_{i,s,b}$, depending on endogenous shrouding practices set by the betting agencies. Accordingly, the surcharge-inclusive betting prices, $p_{ib}$,  may also differ from the implied posted tax-exclusive betting prices $\Tilde{p}_{ib}$, which are equal to:
\begin{equation*}
   \Tilde{p}_{ib}=1- \dfrac{1}{\Tilde{\theta}_{ib}}
\end{equation*}
where $\Tilde{\theta}_{ib}=\sum_{s=1}^n \frac{1}{\Tilde{r}_{isb}}$ is the implied surcharge-exclusive markup for bookmaker $b$ and event $i$. Prior to the tax reform, there was no difference between posted and effective betting prices, implying that $p_{ib}=\Tilde{p}_{ib}$. Let's denote $\tau$ as the shrouded tax surcharge, i.e., the difference between the tax-inclusive and posted betting prices:
\begin{equation*}
\tau=p_{ib}-\Tilde{p_{ib}}=\dfrac{1}{\Tilde{\theta}_{ib}}-\dfrac{1}{\theta_{ib}}
\end{equation*} 

As evident from agencies' annual reports, shrouding policies are an active decision by betting agencies to increase the effective betting prices after the introduction of the sports betting tax in Germany. For instance, the 2012 annual report of \emph{Bwin}, one of the leading international betting agencies, makes the following statement: the ``gross win margin [...] increased sharply in the second half [...] also because from August 2012 we began to withhold 5\% of players' winning bets to cover the gaming tax due'' \citep{bwin2013annual}. 

The shrouding policies are irrelevant to the tax effect on consumer prices according to the standard theory of taxation \citep{kotlikoff1987tax}. In these models, the pass-through rates depend on tax-inclusive price elasticities of demand and supply and, in the case of imperfect competition, additionally on market power \citep{weyl2013pass}. See section \ref{sec:app_standard} in the appendix for details.

In principle, three different policies betting agencies set that affect the difference between posted surcharge-exclusive and effective surcharge-inclusive betting odds: 
\begin{itemize}
    \item[i)] \textbf{No further deductions}: The posted odds are equivalent to the effective odds: $r_{i,s,b}=\Tilde{r}_{isb}$. There is no difference between advertised and effective betting prices: $\tau=0$.
    \item[ii)] \textbf{Deducting tax surcharge from winnings}: In the case of a successful bet, the bettors receive only $(1-t)$ of the implied surcharge-exclusive payout. As a result, the effective tax-inclusive odds are equal to: $r_{i,s,b}=\Tilde{r}_{isb}(1-t)$. This implies: $\tau=\frac{t}{\Tilde{\theta}_{ib}}$ (See Appendix \ref{sec:sport_app_tau} on the derivation of $\tau$ for the different shrouding policies).
    \item[iii)] \textbf{Deducting tax surcharge from wager}:  The tax surcharge applies directly to the effective wager, such that the effective wager multiplied by $(1+t)$ equals the amount the bettor pays. This implies that the effective tax-inclusive odds are equal to: $r_{i,s,b}=\frac{\Tilde{r}_{isb}}{(1+t)}$. This policy is slightly more favorable to the bettor than policy ii), given that surcharge-exclusive odds equate (and $t>0$). In this case: $\tau=\frac{t}{(1+t)\Tilde{\theta}_{ib}}$.
\end{itemize}

In principle, betting agencies could also set policies in which the deductions are larger or smaller than $t$. However, there is no example of a betting agency that does that. To get a better idea about the different shrouding policies, turn again to the example above (Figure \ref{fig:odds_example}). 

Suppose a bettor wants to bet on a draw between Chelsea and Manchester City (marked by the red box). The posted tax-exclusive odds are equal to 3.30, implying that without further deductions, the bettor would receive \euro{}3.30 for each \euro{}1 wagered in the case the game ends with a draw. If the bettor clicks on one of the outcomes, draw in our example, there is a new panel popping up---the betting slip (``Wettschein'')---that has an input window for the amount a bettor wants to wager. After the amount is entered, the potential net payout, i.e., tax-inclusive odds multiplied by the wager, is calculated and presented to the bettor. This is the first time during the interaction with the website when tax-inclusive odds are directly posted to the bettor.

Figure \ref{fig:policies} illustrates examples of such a betting slip of three betting agencies using different shrouding policies for the same bet (draw in the game Chelsea vs. Manchester City). Panel ii) refers to the betting agency in our example above applying shrouding policy ii) (``deducting tax from winnings''). A bet of \euro{}100 would not result in a potential net payout of \euro{}330 as implied by the tax-exclusive odds of 3.30 but only in \euro{}313.50, which is equal to $330*(1-0.05)$.\footnote{The yellow button ``Einzahlen'' button translates to Deposit, as the bettor has no money in her account. If the bettor has sufficient funds, clicking the yellow button will place the bet.} In contrast, the betting agency in panel i) does not impose further deductions on the posted odd, implying that the potential net payout is equal to the wager times the posted odds, \euro{}330 in our example. To advertise that there are no further tax surcharges (``Gebühr'') on this website, the bookmaker additionally displays the surcharges a bettor would pay with a bookmaker that deducts 5\% of the winnings as in the previous example (\euro{}16.50), see the crossed-out red number. Panel iii) illustrates an example of an agency that employs the third policy (``deducting a surcharge from the effective wager'').\footnote{Please note that the bookmaker in this example uses $t=0.053$, the new tax rate after the reform in July 2021. Before the reform, the agency used $t=0.05$ as well.} The tax-exclusive odd for this agency is only equal to 3.15. Compared to the agency in Panel ii), the agency is more transparent about how the net payout (``Möglicher Gewinn'') is calculated, displaying the steps of the underlying calculation, including effective wager (``Effektiver Einsatz'') and the ``tax'' (``Steuer'') directly. The presentation of the betting slip for other agencies by shrouding policies is very similar. 

    \begin{figure}[!ht]
     	\caption{Example of different shrouding policies}
     	 \label{fig:policies}

        \begin{minipage}[b]{0.31\linewidth}
            \centering
            \caption*{i) No further deductions}
            \includegraphics[width=\textwidth]{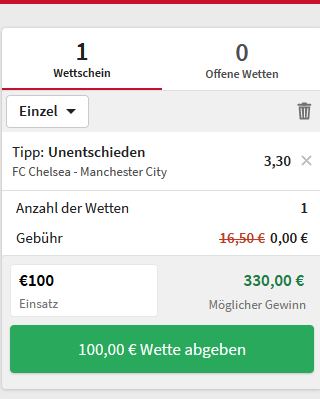}
            \label{fig:a}
        \end{minipage}
        \hspace{0.2cm}
        \begin{minipage}[b]{0.31\linewidth}
            \centering
            \caption*{ii) Deduction of winnings}
            \includegraphics[width=\textwidth]{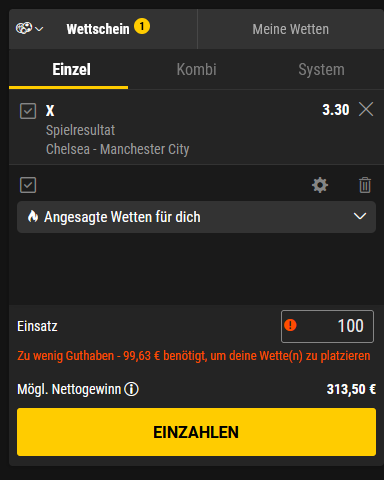}
            \label{fig:b}
        \end{minipage}
            \hspace{0.2cm}
        \begin{minipage}[b]{0.31\linewidth}
            \centering
            \caption*{iii) Deduction of wager}
            \includegraphics[width=\textwidth]{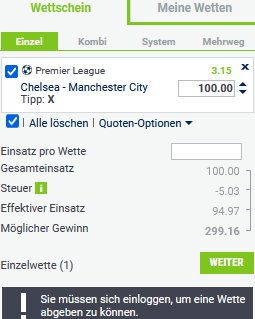}
            \label{fig:c}
        \end{minipage}
        \begin{minipage}{15.8cm}  \footnotesize  \emph{Notes:} 
This figure presents examples of different betting slips (browser version) after the bettor has chosen one of the potential outcomes: a draw in the game between Chelsea FC and Manchester City in our example. The gross odds on the top are presented, excluding the tax, for all three panels. All of the betting slips show the potential net payouts on the bottom. The size of the net payouts depends on different shrouding policies: i) No (additional) deduction of taxes, i.e., effective tax-inclusive odds are equal to posted tax-exclusive odds; ii) deduction of 5\% of the winning, i.e., effective tax-inclusive odds are equal to 95\% of the posted odds; iii) deduction of 5\% of the wager, i.e., the effective tax-inclusive odds are equal to $1/1.05$ of the posted odds. Note that example iii) already applies the new tax rate of 5.3\%, which was introduced in Juli 2021. Source: i) Tipico.com; ii) bwin.com; iii) Bet365.com  - September 2021.   
		\end{minipage}
    \end{figure}

%% file: Chapters/Data.tex
The main data set covers fixed sports betting odds from 68 online betting agencies between 2008 and 2018 for more than 80,000 events in 16 different leagues. The events in the data set span leagues in six different countries and five different sports. The number of unique events per league and corresponding fractions are displayed in Table \ref{sum_comp}. For each event, I observe the pre-match closing odds by betting agencies on each potential outcome of the game, which allows me to calculate the price of a bet as defined above. If not stated otherwise, betting prices $p_{ib}$ always refer to the effective surcharge-inclusive betting prices. For each event, the data set further records the result, date, league, country of competition, and the sport of the event. The 68 agencies include agencies that share the same brand but operate from different domains. For instance, \textit{Bwin} operates from different domains in some countries, e.g. in Italy from \emph{.it}; and in Spain from \emph{.es}. I include each domain as a single agency because, typically, the prices differ across different domains for the same brand. In total, there are 55 unique bookmaker brands in the sample.\footnote{ A list of all betting agencies, the number of included events and observations per year, and the average price by betting agencies can be found in the online appendix.}

\subsection{Data sources}

This section gives an overview of the data sources used to construct the data set. More detailed information are provided in appendix.

\subsubsection*{Betting prices} 

The raw odds data is obtained from \textit{oddsportal.com}---a commercial website that provides odd comparison tools and an odds archive for many events, betting types, and periods to their users since 2008.\footnote{Data for the betting agency Tipico, one of the leading agencies in Germany, is scraped directly from the results archive from their website. The Tipico archive was available until its relaunch in 2020 under the following URL: https://www.tipico.com/de/ergebnisse/. Data were available from the start of 2012 onwards. All odds before 2012 refer to odds scraped from the website of their partner label Rivalo. Rivalo provides betting services for Swiss customers only, using identical odds as Tipico before the tax reform, and still uses the old website with an active scrapeable results archive that dates back to 2008. In line with this claim (motivated by personal observations), the odds in 2012 before the introduction of the tax available on the Tipico and the Rivalo webpage are identical. There is further no jump in betting prices between 2011 and 2012. All raw odds are surcharge-exclusive.} \textit{Oddsportal.com} collects the odds data directly from online betting agencies' websites---either via their Application Programming Interface (API) or xml (Extensible Markup Language) odds feeds (as confirmed by email exchange). Thus, one limitation of the data set is that agencies report only one set of odds for each event. However, betting agencies could, in principle, use different odds for different regions or countries conditional on regional IP addresses. Historical archived odds may thus not necessarily refer to the actual odds displayed to bettors in different countries. However, several points assures me that this is not a major concern for this study.
First, to the best of my knowledge, there is no evidence in the agencies' annual reports or on other websites that agencies have engaged in such behavior between 2009 and 2018. Second, while it is not possible to check which odds were displayed to German customers in the past, I have manually compared the "live" odds displayed on oddsportal.com with odds on the agencies' websites for a subsample of events in 2021. The comparisons were made from a German IP address. I could not find any discrepancies between the odds on oddsportal.com and the agencies' websites.\footnote{Exemplary screenshots can be provided upon request.}  Third, if the odds on oddsportal.com were considerably off from the true odds in different countries, their services would provide no value to their consumers. Fourth, there are several "local" agencies that operate from country-specific domains that are only accessible by domestic bettors.  

\input{Tables/List_leagues}

\subsubsection*{Agency characteristics and shrouding policies}

To determine the type and timing of the betting agency's "shrouding" policy, I match the odds data with hand-collected data from different sources. The sources include annual reports of the listed agencies, websites from betting agencies, web-scraped information from different commercial betting websites and forums, including \textit{Top100bookmakers.com}, as well as information from bodies that represent the interest of betting agencies, such as the  "Deutsche Sportwettenverband" (DSWV, German sports betting association) and the European Gaming and Betting Association (EGBA).\footnote{I use snapshots from \textit{Top100bookmakers.com} at two dates: September 2021, which I collected, and data that was already used in \citet{montone2021optimal}. I am immensely thankful to Maurizio Montone, who provided me with the data. Other websites include: \textit{www.wettforum.info}, \textit{www.wettbuero.de}, \textit{www.wettsteuer.com} and others.} I further use information from official regulatory bodies, including the Malta Gaming Authority (MGA). The matched data set further allows me to infer whether a betting agency targets the German market or not.

\subsubsection*{Online betting and tax revenues}

Data on online betting revenues was provided by H2 Gambling Capital (website: \url{https://h2gc.com/}). H2 Gambling Capital is a betting and gaming consultancy that is one of the leading market data providers in the gambling industry. Their data is regularly used by the main listed sports betting agencies in their annual reports \citep[e.g.][]{bwin2013annual} and by the International Betting Integrity Association -- IBIA \citep{ibia2021}. The data is primarily taken from information by official gambling regulatory bodies in different jurisdictions, governments, and tax authorities. It only includes revenues of operators licensed in at least one jurisdiction. Entirely unlicensed activities (black market) are not included. Collecting and disaggregating the data, H2 further uses information from national monopoly operators (where available for land-based and onshore interactive activity) and information derived from major operators and suppliers active in each market, including publicly available reports. To determine whether an agency targets a particular market, H2 undertakes a biannual audit of all operators, examining whether agencies proactively target domestic players, accept domestic players, or (not) prohibit domestic players from betting.\footnote{
According to email correspondence, H2 further considers any actual market information by companies and feedback from their users (they currently have well over 1,000 users with access to different local operators and suppliers) when allocating activity across markets. Furthermore, they also consider any restrictions in place, any onshore supply, and several generic measures ranging from GDP to broadband and smartphone penetration to check the size of each market. The author thanks H2 for providing the data on a pro-bono basis.} Tax revenues refer to administrative tax data of the German Federal Ministry of Finance.


\subsection{Descriptive statistics}

In total, there are around 3.3m observations in the data set. Each observation refers to a unique event and betting agency combination. The increasing number of observations over time (Table \ref{tab:sum_price}) can be explained by an increasing number of betting agencies included in the data set.\footnote{Missing observation for some agencies in specific years or specific events can have three explanations: i) the agencies were not operating in a given year; ii) the agencies were providing betting services for a specific event; or iii) oddsportal.com was not collecting the odds for a betting agency for a specific year or event. According to an email by oddsportal.com, the reasons are purely technical and primarily related to betting agencies' APIs.} Differences in observations across leagues are primarily due to a varying number of games per season. Betting agencies offer a wide range of leagues that usually cover the major leagues and sports considered in this paper, even at the start of their services.

\input{Tables/Summary_price}

Table \ref{tab:sum_price} illustrates the mean and standard deviation of (tax-inclusive) betting prices. The average betting price for all events is equal to 0.0706, which implies, as described above, that a betting agency makes around 0.0706 \euro{} per wagered \euro{}.
The mean of betting prices for soccer games is equal to 0.0734, and the standard deviation slightly decreases from 0.0317 to 0.0315. If we only consider agencies for which odds data is available for all years, the average betting prices slightly increase to 0.0727 and 0.0077, respectively. In contrast, the standard deviation decreases to around 0.027.

In general, the betting prices vary across different leagues. Prices are smaller for "higher" leagues, where the number of bettors and competition among agencies is usually higher. A larger standard deviation accompanies this difference. Furthermore, prices tend to be smaller for sports with two possible outcomes (Basketball, American Football, Hockey).
These findings align with results by \citet{montone2021optimal}, who rationalizes the findings with risk-averse bookmakers. The prices across different countries, controlling for the division, are very similar. Only English leagues tend to exhibit a comparably smaller price, which may relate to the larger market size and competition in the UK. Prices show a decreasing trend after 2013. Before 2013, betting prices are relatively constant over time.\footnote{Please note that betting prices from 2012 onward can be affected by the sports betting tax.}

Figure \ref{fig:hist_prices} in the appendix illustrates the distribution of betting prices across all observations for all sports and soccer games. The distribution is relatively symmetric around the mean of betting prices. There is, however, some excess mass around a betting markup of around 0.1, a standard assumption often made in the literature that studies relative betting prices \citep[e.g.][]{levitt2004gambling}. A log transformation considerably skews the data to the right, as prices are smaller than 1, which is, besides the direct interpretability, one of the main reasons why I concentrate on absolute betting prices in my empirical analyses.

%% file: Tables/List_leagues.tex
\begin{table}[!htbp]\centering
\caption{List of included competitions and number of unique matches}
\label{sum_comp}
\scalebox{0.95}{
\begin{threeparttable}

\begin{tabular}{ l|c|c|c|c|c } 

 League & Sport & Country  & \# Outcomes & \# Matches & Percent\\  \hline
 Bundesliga & Soccer & Germany & 3-way & 3,060 & 3.78  \\ 
 2. Bundesliga & Soccer & Germany & 3-way & 3,109 & 3.84 \\ 
 3. Liga & Soccer & Germany & 3-way & 3,802 & 4.70 \\ 
 Premier League & Soccer & England & 3-way &  3,801 &  4.70  \\ 
 Championship & Soccer & England & 3-way & 5,558 & 6.87   \\ 
 Primera Division & Soccer & Spain & 3-way & 3,807 & 4.71 \\ 
 Segunda Division & Soccer & Spain & 3-way & 4,691 & 5.80 \\ 
 Serie A & Soccer & Italy & 3-way & 3,819 &  4.72   \\ 
 Serie B & Soccer & Italy & 3-way & 4,668 & 5.77  \\
 Ligue 1 & Soccer & France & 3-way & 3,789 & 4.68  \\ 
Ligue 2 & Soccer & France & 3-way & 3,818 & 4.72  \\ 
Bundesliga & Handball & Germany & 3-way & 3,097 & 3.83 \\
Bundesliga & Basketball & Germany & 2-way & 3,305 & 4.08  \\
 NBA & Basketball & USA & 2-way & 13,702 & 16.94 \\ 
  NFL & Am. Football & USA & 2-way & 3,323 & 4.11  \\ 
   NHL & Hockey & USA & 2-way & 13,560 & 16.76  \\ 
   \hline
   Total & & & & 80,909 &100.00 \\
   \hline
 
\end{tabular}
\begin{tablenotes}[flushleft]
\item \emph{Notes:} The table lists the leagues included in the data set together with the respective sport, the country of the league, and the number of potential outcomes of the considered bet. The last two columns report the number of unique matches and respective frequencies for each league.
\end{tablenotes}
\end{threeparttable}
}
\end{table}

%% file: Tables/Summary_price.tex
\begin{table}[!htbp]\centering
\caption{Summary statistics - betting prices by league and year}
\label{tab:sum_price}
\scalebox{0.95}{
\begin{threeparttable}

    \begin{tabular}{l|ccc}
    \multicolumn{1}{c|}{League - Sport} & Mean  & Std.Dev. & Observations \\ \hline
     \textbf{All agencies:}     &       &       &  \\
         [0.5em]
    All events & 0.0706 & 0.0317 & 3,289,135 \\ 
    Soccer events     &  0.0734     &   0.0315    & 2,067,137  \\ \hline
    \textbf{Agencies with data for all years:}     &       &       &  \\
         [0.5em]
    All events & 0.0727& 0.0273 &  2,259,449 \\ 
    Soccer events     &  0.0770     & 0.0271      & 1,407,215 \\ \hline
    \textbf{All agencies by League:} &       &       &  \\
    [0.5em]
    Bundesliga - Soccer & 0.0621 & 0.0280 & 147,894 \\
    2 Bundesliga - Soccer & 0.0792 & 0.0302 & 147,501 \\
    3 Liga - Soccer & 0.0887 & 0.0289 & 161,139 \\
    Premier League - Soccer & 0.0576 & 0.0284 & 184,909 \\
    Championship - Soccer & 0.0740 & 0.0325 & 264,213 \\
    Primera Division - Soc.. & 0.0616 & 0.0284 & 183,514 \\
    Segunda Division - Soc.. & 0.0847 & 0.0311 & 220,208 \\
    Serie A - Soccer & 0.0627 & 0.0286 & 184,610 \\
    Serie B - Soccer & 0.0848 & 0.0307 & 214,155 \\
    Ligue 1 - Soccer & 0.0659 & 0.0276 & 179,113 \\
    Ligue 2 - Soccer & 0.0819 & 0.0285 & 179,881 \\
    Bundesliga - Handball & 0.0932 & 0.0321 & 115,624 \\
    Bundesliga - Basketball & 0.0752 & 0.0282 & 110,088 \\
    NBA - Basketball & 0.0615 & 0.0294 & 474,717 \\
    NFL - American Football & 0.0585 & 0.0280 & 112,485 \\
    NHL - Hockey & 0.0630 & 0.0314 & 409,084 \\ \hline
    \textbf{All agencies by Year:} &       &       &  \\
         [0.5em]
    2009  & 0.0745 & 0.0275 & 199,885 \\
    2010  & 0.0725 & 0.0264 & 229,491 \\
    2011  & 0.0747 & 0.0301 & 255,470 \\
    2012  & 0.0752 & 0.0315 & 288,038 \\
    2013  & 0.0770 & 0.0327 & 325,451 \\
    2014  & 0.0748 & 0.0329 & 353,585 \\
    2015  & 0.0708 & 0.0320 & 377,322 \\
    2016  & 0.0678 & 0.0323 & 391,238 \\
    2017  & 0.0649 & 0.0322 & 437,352 \\
    2018  & 0.0622 & 0.0312 & 431,303 \\ \hline
\textbf{German agencies by shrouding policy:} &       &       &  \\
         [0.5em]    
No "shrouding" (i) & 0.0737 & 0.0199 & 24,351 \\
Deduction from winnings (ii)  & 0.0726 &  0.0225 & 143,842 \\
Deduction from wager (iii) & 0.0745 & 0.0207 & 47,068 \\

    \hline \hline
    \end{tabular}%
\begin{tablenotes}[flushleft]
\item \emph{Notes:} The first four rows of the table display the mean and standard deviation of betting prices for all events and soccer events only, considering all agencies or only agencies for which odds data is available in all years. The two middle panels show the summary statistics for all agencies disaggregated by leagues and years. The last rows illustrate the mean and standard deviation of pre-reform betting prices of treated agencies grouped according to the respective implemented post-reform "shrouding" policy. The last column reports the respective number of observations. One observation refers to a price for a specific event and bookmaker.
\end{tablenotes}
\end{threeparttable}
}
\end{table}

%% file: Chapters/Methodology.tex
\subsection{Empirical strategy}\label{sec:empirical_strategy}

\subsubsection*{Average tax effect on consumer prices}

To estimate the effect of the betting tax on consumer prices, I use a DID framework. I compare changes in consumer betting prices (as defined above) around the 2012 tax reform between German and Non-German agencies. DID frameworks are commonly used to causally estimate the pass-through of taxes \citep[e.g.][]{doyle20082, harding2012heterogeneous, hindriks2019heterogeneity, harju2018firm, harju2022heterogeneous}. In the main specification, I estimate the following equation:

\begin{equation}
\label{eq:tax_avg_effect}
p_{i, m(t, c)}=\beta_0 + \beta_1 T_{i,t} + \alpha_i + \lambda_t + \psi_c +\epsilon_{i, m},
\end{equation}

where $p_{i,m(t,c)}$ is the betting price of agency $i$ for event $m$, taking place in week $t$ and league $c$. $T_{i,t}$ is a dummy variable that equals one if agency $i$ is in the treatment group, i.e., is active in the German market, and the event takes place after the tax was in place.  $\alpha_i$, $\lambda_t$, $\psi_c$ are agency, week and league fixed effects. I use robust standard errors clustered on the agency level.

The main identifying assumption is that the betting prices of treated and control agencies would have followed a parallel time trend in the absence of the tax reform, controlling for time, agency, and league fixed effects. If the parallel trend assumption holds, $\beta_1$ identifies the average causal effect of the betting tax on consumer betting prices with an equal weight on each observation. Note that the tax is applied to all treated agencies at the same time. Consequently, the issue of biased multi-period DID estimators (commonly referred to as "two-way fixed-effects" models) in settings with heterogenous treatment effects and timing does not apply.\footnote{ A detailed discussion of these issues can be found here: \citet{de2020two, goodman2021difference, sun2021estimating}.} In addition, the treatment can be expected to be exogenous to betting prices since the tax reform is mainly determined by unrelated political processes and rulings by courts (see section \ref{sec:sports_tax}). Yet, changing regulation after the German tax reform in other European countries may still confound our estimates. Note, however, that most European countries, including the UK, Italy, and France, comprehensively regulated online sports betting already before 2012, implying that effects of these regulation should already show up in the pre-trends.\footnote{\emph{Thomson Reuters} and the \emph{Global Legal Group} provide detailed overviews of online gambling law in different European countries: \url{https://uk.practicallaw.thomsonreuters.com/} and \url{https://iclg.com/practice-areas/gambling-laws-and-regulations}.}

As there was effectively no taxation on online betting services in Germany before the reform and no additional taxes after the reform, the change in consumer prices can be directly translated to the pass-through rate of the sports betting tax on consumer prices. The estimated pass-through of the tax is equal to $\frac{d p}{d t}=\frac{\beta_1}{0.05}$.\footnote{This notation implicitly assumes that the marginal (or local) pass-through rate is equal to the non-marginal (or global) pass-through rate, which is only true for the linear case. The global pass-through rate is equal to the average local pass-through rate over the respective tax change interval.} Besides the Sports Bet Tax, producer prices $q$ are assumed to follow the same trend as consumer prices before and after the tax reform, such that the incidence of producers is equal to $\frac{dq}{dt}=1-\frac{dp}{dt}$.

While it is not possible to test the parallel trends assumption directly, the assumption is more plausible if betting prices in the control and treatment groups showed a common trend prior to July 2012. Figure \ref{fig:avgprices_all} show average quarterly and weekly betting prices over time for all agencies in the sample. The circles signify average weekly betting prices in each group pooled and aligned by the respective quarter. The solid line illustrates the average quarterly betting prices. Prior to the introduction of the tax reform, which is illustrated by the vertical red line, the average betting prices moved very similarly over time in both groups. We observe common pre-trends if we consider all sports (Panel a), only soccer events are considered (Panel c) and if exclude cross-leagues, i.e., Non-German Leagues in the Treatment group and German games in the Control group, the trends prior to the tax introduction remain very similar (Panels b and d). The same holds true if we only consider agencies for which odds data is available for all periods (Figure \ref{fig:avgprices_compl}). Betting prices are higher on average before July 2012 for agencies with a foreign domain than in the treatment group (see Panels a and b of Figure \ref{fig:avgprices_all_robust}). At the same time, the pre-trends between these two groups are parallel, implying that the parallel trend assumption also holds in this case.

\input{Figures_tex/Avg_prices_all.tex}

The observed seasonality in betting prices, considering all sports, stems from the fact that there are relatively more events with two potential outcomes in the winter, Basketball, Hockey, and American Football, that tend to have lower prices. Given that more than 80\% of all wagers are made on Soccer matches in the European market \citep{tipico2018social}, and soccer games are not affected by seasonality, I will mainly focus on soccer matches in the remainder of this paper. 
To provide a more structured test of the parallel trend assumption and examine how the pass-through to consumer prices evolve over time, I estimate the following multi-period event study design:

\begin{equation}
\label{eq:tax_dyn_effect}
p_{i, m(t, c)}=\alpha_i + \lambda_t  +  \psi_c + \sum_{k=-13}^{25} \beta_k D_{i,t}^k +\epsilon_{i,m},
\end{equation}

where $D_{i,t}^k$ is an indicator for event $m$ being $k$ quarters relative to the introduction of the tax when the agencies $i$ is in the treatment group. $D_{i,t}^k$ equals zero for all agencies in the control group. For better results readability, $t$ refers to a quarter instead of a week here. The results are robust if I use week-fixed effects instead. Insignificant coefficients for leading quarters, $k<0$, indicate that the parallel trend assumption is plausible.

\subsubsection*{Heterogeneity in tax incidence}

I use two approaches to estimate heterogeneity in pass-through rates conditional on different "shrouding" policies. First, I run specifications \ref{eq:tax_avg_effect} and \ref{eq:tax_dyn_effect} for different agencies' subsamples conditional on the applied shrouding policy. This categorization is consistent over time, as I observe no switching back to "No shrouding" once one of the shrouding policies is implemented or between different "shrouding" policies. 

Second, I use an interaction term that interacts treatment ($T_{i,m}$) with a $no Shroud_{i,m}$ indicator that equals one if the event took place after the tax reform and agency $i$ had a NO shrouding policy in place for event $m$. 
In detail, I estimate the following equations for the average tax effects on betting prices:

\begin{equation}
\label{eq:tax_avg_het}
p_{i, m(t, c)}= \beta_1 T_{i,t} + \beta_2 T_{i,t} \times no Shroud_{i,m} + \alpha_i + \lambda_t + \psi_c +\epsilon_{i, m},
\end{equation}

The coefficient on $T_{i,t}$ can be interpreted as the effect of the tax on betting prices when a shrouding policy was in place for a specific event and agency. In contrast, the coefficient of the interaction term estimates the difference between the tax effect when no shrouding policy was in place and the tax effect when a "shrouding" policy was in place. Accordingly, the t-statistics of the coefficient on the interaction term can directly serve as a test of a differential tax effect conditional on whether a shrouding policy was in place or not. The interpretation of the estimated effects between the subsample analyses and the regressions with the interaction term differ because "shrouding" policies, used for the subsample analyses, are defined on the agency level, while the $no Shroud_{i,m}$ dummy is defined on the agency-time dimension. This implies that coefficients of the subsample analyses may be diluted because it includes treated observations in "shrouding samples", in which agencies that switched to shrouding later have not put the shrouding policy in place yet. In contrast, the estimation with an interaction term compares prices for treated events not subject to an active shrouding policy with events for which an active shrouding was in place.

Similarly, I also interact the leads and lags of periods relative to the tax ($D_{i,t}^k$) with the $no Shroud_{i,m}$ indicator to capture the dynamics in the differential tax effects over time: 

\begin{equation}
\label{eq:tax_dyn_het}
p_{i, m(t, c)}=\alpha_i + \lambda_t  +  \psi_c + \sum_{k=-13}^{25} \beta_k D_{i,t}^k +\sum_{k=-13}^{25} \beta_k D_{i,t}^k \times no Shroud_{i,m} +\epsilon_{i,m},
\end{equation}

The heterogeneity analyses require somewhat stronger assumptions. The parallel trend assumption remains crucial to estimating the causal effects of the tax on consumer prices in the different subsamples. The difference is that the parallel trend assumption now applies to the respective treatment group subsamples (and not to all agencies in the treatment group) and the control group. For the heterogeneity analyses, we need to assume that the shrouding policies are exogenous to changes in price setting behavior and firms' characteristics that affect price setting behavior, conditional on time and agency fixed effects. Agencies seem to not differ significantly in price setting behavior before the reform (see Table \ref{tab:sum_price}), and pre-trends follow the same trends as the control group for the shrouding and non-shrouding subsamples (see Figure \ref{fig:coef_soccer_het}, Panel a). Yet, shrouding practices are endogenous by definition, and differences in pass-through rates should thus not be interpreted as being caused by the shrouding practices alone. Other characteristics correlated with the shrouding policies may also drive part of the observed heterogeneity.

\subsection{Defining the control and treatment group}

One of the critical issues for a clean identification of the tax incidence of the German sports betting tax is to define a consistent and reliable treatment and control group. The treatment group should consist of agencies that provide their services to German customers, i.e., who are active in the German market. The control group should ideally consist of agencies whose services are not available to German customers, i.e., that do not operate in the German online betting market. This categorization is challenging due to the opacity and the cross-border dimension of the online betting market.

I limit the potential set of treated bookmakers to agencies that are or were members of the DSWV, the German sports betting association. The DSWV, founded in 2014, represents the interest of leading German and European sports betting providers and describes itself on its website as "the public contact of their members to German politics, sports, and media". The members of the DSWV represented around 90 percent of the entire German sports betting market in August 2018 \citep{dswv2018}. According to the information from Top100bookmakers.com, used in \citet{montone2021optimal}, all of these agencies had a German language version as of 2021. I further use the hand-collect data on this sub-set of agencies (see above) to track their activities in the German betting market over time.  
Agencies only active in Germany for a limited period between 2009 and 2018 are excluded from the sample.\footnote{This approach excludes three betting agencies for the estimation: Expekt, Pinnacle, and Betclic. The empirical results are robust if we include these agencies in our sample.} This approach leaves me with a subsample of ten agencies in the Treatment group, henceforth also referred to as "German" agencies, to ease readability.\footnote{The set of German agencies: Betfair Sportsbook, Betsson, Betway, Interwetten, Sportingbet, Tipico, Unibet, Bet365, Bet-at-home, and Bwin. Betfair Sportsbook does not refer to the well-known betting exchange but to a subsidiary that offered fixed odds betting services to German customers after the exchange left the German market. Odds data for Betfair sportsbook are available from the third quarter of 2015. I included Betfair sportsbook in the treatment group as the agency was active on the German market over all periods where data is available. Data for all other treated agencies is available for the entire observation period.} While I am confident that this approach limits the treatment group to agencies that were active in the German market between 2009 and 2018, defining a stringent control group is more challenging, as this requires ruling out that these agencies were never serving German customers at any point. 

In the main specification, I include all other remaining agencies in the data set in the control group. However, I cannot entirely rule out that some German bettors bet with one of the foreign betting agencies that do not primarily target the German market. In this case,  the estimated pass-through rates are potentially biased.

I take two approaches to resolve these concerns. First, I restrict the control group to agencies that are linked to a foreign country-specific domain. Agencies that operate from country-specific domains are usually only available to customers in these countries. Furthermore, even if foreign customers could, in principle, access the webpage, e.g., through a VPN client, they should have no possibility to register with the respective agency. If I consider all agencies,  there are 65 agencies in the sample. Ten are treated, and 55 are in the control group. If I exclude agencies for which data is only available for a subset of years, the number of agencies decreases to 30, with nine agencies in the treatment and 21 agencies in the control group. If I only consider agencies in the control group with a foreign domain, the number of agencies in the control group decreases to 22 and 6, respectively. According to the domain identifier, 12 out of 22 agencies that use a country-specific domain are located in Eastern Europe (Slovakia, Czech Republic, Poland, and Russia). Four agencies use an Italian, three a French, and three a Spanish domain.

Second, suppose that bettors prefer to bet on local games. In this case, potential concerns can be alleviated by excluding "cross" leagues, i.e., German Leagues for Non-German agencies and Non-German games for German agencies. Without "cross" league events, the control group only includes betting markets on foreign events that are less likely to be populated with German bettors--- and thus less affected by potential confounders. Similarly, the treatment group then only includes domestic markets that are primarily populated with German bettors and thus more realistically represent (changes in) German consumer prices. Following the same argument, I also run robustness checks only considering lower-level leagues that are assumingly even less prone to cross-country bettors.

\subsubsection*{Heterogeneity in "shrouding" policy implementation }

As already discussed above, betting agencies mainly engage in three different "shrouding" policies that affect the difference between the advertised surcharge-exclusive and the effective surcharge-inclusive odds: i) no further deductions; ii) deducting tax surcharge from winnings; iii) deducting a surcharge from the wager. Six out of the ten treated agencies choose policy ii), while two choose policy iii), and two do not apply additional tax surcharges for the entire observation period. Most agencies that choose policy ii) or iii) implement "shrouding" policies within six months after the tax reform. There is one "outlier" that switches from no "shrouding" to policy iii) only in March 2016. The type and timing of "shrouding" policies for all treated agencies can be found in the Appendix in Table \ref{sum_policy}. No agency switches back to policy i) or between policies after one of the policies was established. The average pre-reform betting prices conditional on post-reform policies do not differ considerably and range from 0.0726 to 0.0745 (Table \ref{tab:sum_price}).

%% file: Figures_tex/Avg_prices_all.tex
\begin{figure}

    \centering
      \caption{Avg. tax-inclusive consumer betting prices of all German and Non-German betting agencies over time}
    
    \begin{subfigure}[t]{0.48\textwidth}
        \centering
        \caption{ All sports }
        \includegraphics[width=\linewidth]{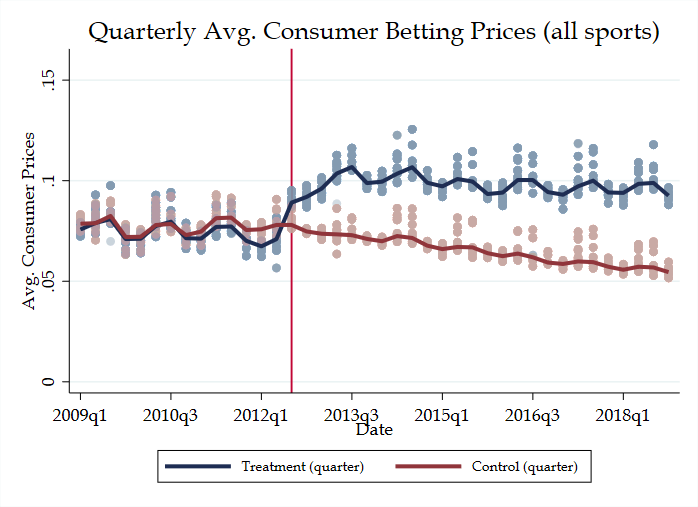} 
         \label{fig:timing1}
    \end{subfigure}
    \hfill
        \begin{subfigure}[t]{0.48\textwidth}
        \centering
        \caption{Soccer games}
        \includegraphics[width=\linewidth]{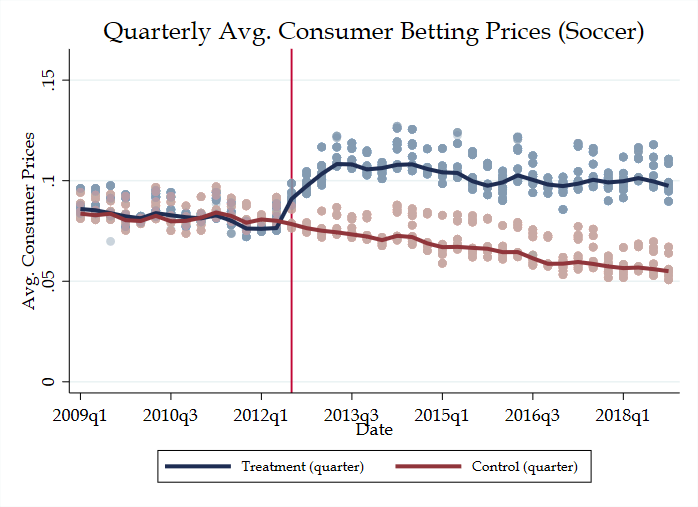} 
        \label{fig:timing1}
    \end{subfigure}

    \vspace{0.5cm}
    \begin{subfigure}[t]{0.48\textwidth}
        \centering
         \caption{ All sports - excluding cross leagues }
        \includegraphics[width=\linewidth]{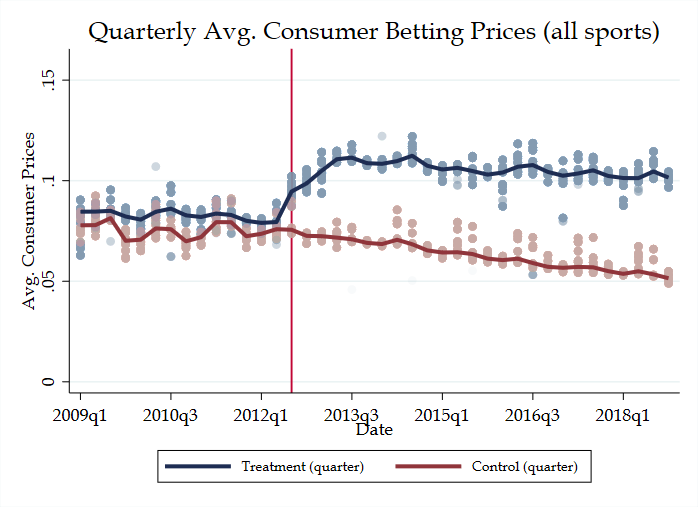} 
        \label{fig:timing2}
    \end{subfigure}
    \hfill
    \begin{subfigure}[t]{0.48\textwidth}
        \centering
        \caption{Soccer games - excluding cross leagues}
        \includegraphics[width=\linewidth]{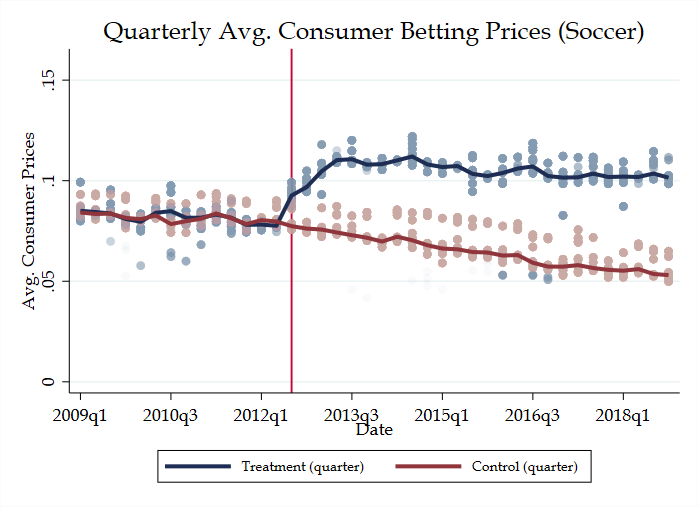} 
         \label{fig:timing2}
    \end{subfigure}
     		\begin{minipage}{15.8cm}  \footnotesize  \emph{Notes:} 
This figure shows the average quarterly and weekly betting prices over time for agencies that target the German market (Treatment) and for agencies that do not target the German market (Control). All agencies in the Control and Treatment group are considered, N=65 (Treatment=10/Control=55). The circles signify average weekly betting prices in each group, pooled and aligned by the respective quarter. The solid line illustrates the average quarterly betting prices. The odds for the graphs are trimmed at the quarterly 1- and 99-percentiles, considering all events and all agencies. Panel (a) and (c) consider all sports included in the data set, while Panel (b) and (d) only consider soccer events. Panel (c) and (d) exclude German leagues for the control group and Non-German leagues for the treatment group. The vertical red line illustrates the introduction of the tax reform.
		\end{minipage} 
		\label{fig:avgprices_all}
  
\end{figure} 

%% file: Chapters/Results.tex
Before I show the estimated tax effects according to Section \ref{sec:empirical_strategy}, I examine whether agencies follow the tax rules, which is crucial for the validity of the estimation approach and has important implications for digital taxation in general. 

\subsection{Sports betting tax revenues} \label{sec:revenues_sport}

As noted above, most betting agencies operate from offshore locations and usually have no jurisdiction in Germany. Consequently, it is not straightforward to assume that agencies also pay the German sports betting tax---an issue that is common for many highly digitalized business models. In this section, I use administrative tax data and aggregated sports betting revenue data provided by a commercial provider to examine whether agencies follow the tax rules. 

Figure \ref{fig:tax_revenues} shows the aggregated yearly sports betting tax revenues, according to administrative tax data of the German Federal Ministry of Finance (2021). Tax revenues more than doubled from around \euro{}200m in 2013 to over \euro{}450m in 2019. In total, the German sports betting tax generated revenues of more than \euro{}2.6bn until 2020.

		\begin{figure}[!ht]
		\caption{German sports betting tax revenue 2012-2020 in \euro{}M}
		\centering	
	
		\includegraphics[width=\textwidth]{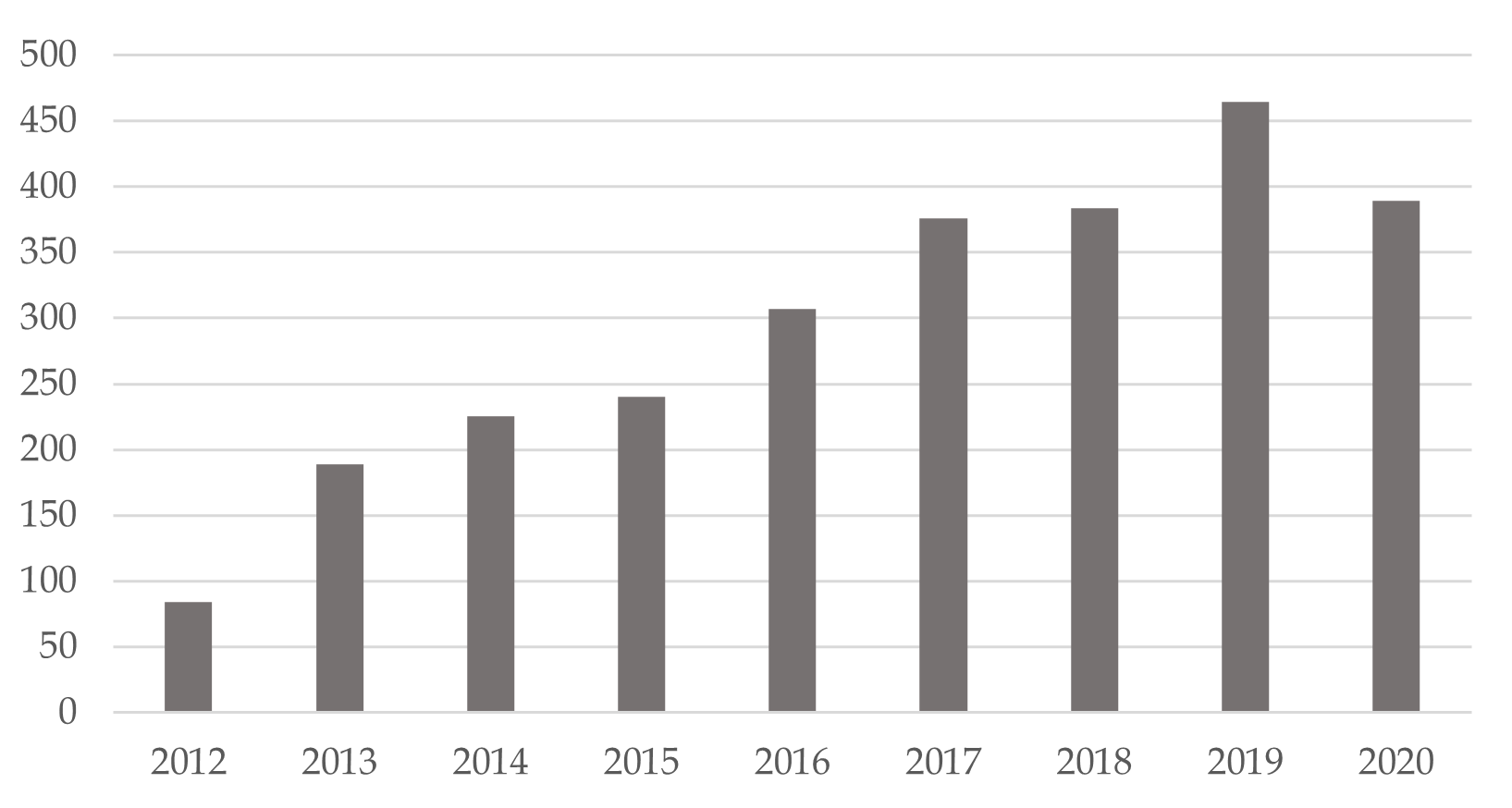}
		\begin{minipage}{15.8cm}  \footnotesize  \emph{Notes:} 
This figure illustrates the aggregated annual tax revenue (in \euro{}M) generated by the German sports betting tax between 2012 and 2020. The tax revenue in 2012 only covers a period of 6 months, as the tax was introduced on 1 July 2012. Source: German Federal Ministry of Finance, own representation.  
		\end{minipage}
		  \label{fig:tax_revenues}
	\end{figure}

Figure \ref{fig:bettingrevenues_ger} illustrates the annual total gross online betting revenues between 2008 and 2020 in Germany, disaggregated by onshore and offshore betting revenues.\footnote{For 2012, the onshore revenues refer to annual revenues by agencies that paid sports betting taxes after the tax reform in 2012.}  Onshore revenues refer to all online sports betting activities for which taxes were paid in Germany. Offshore revenues include all sports betting revenues by agencies that target the German market but do not pay the sports betting tax in Germany. Note that gross betting revenues are defined as the total wagered amount minus all winnings by bettors. Consequently, gross betting revenues' size directly depends on the price adaptions by agencies, ceteris paribus. In contrast, the tax base only indirectly depends on betting prices through behavioral responses by bettors. 

	\begin{figure}[!ht]
		\caption{Onshore vs. offshore gross online betting revenues in Germany in \euro{}M}
			\label{fig:bettingrevenues_ger}

		\centering	
	
		\includegraphics[width=\textwidth]{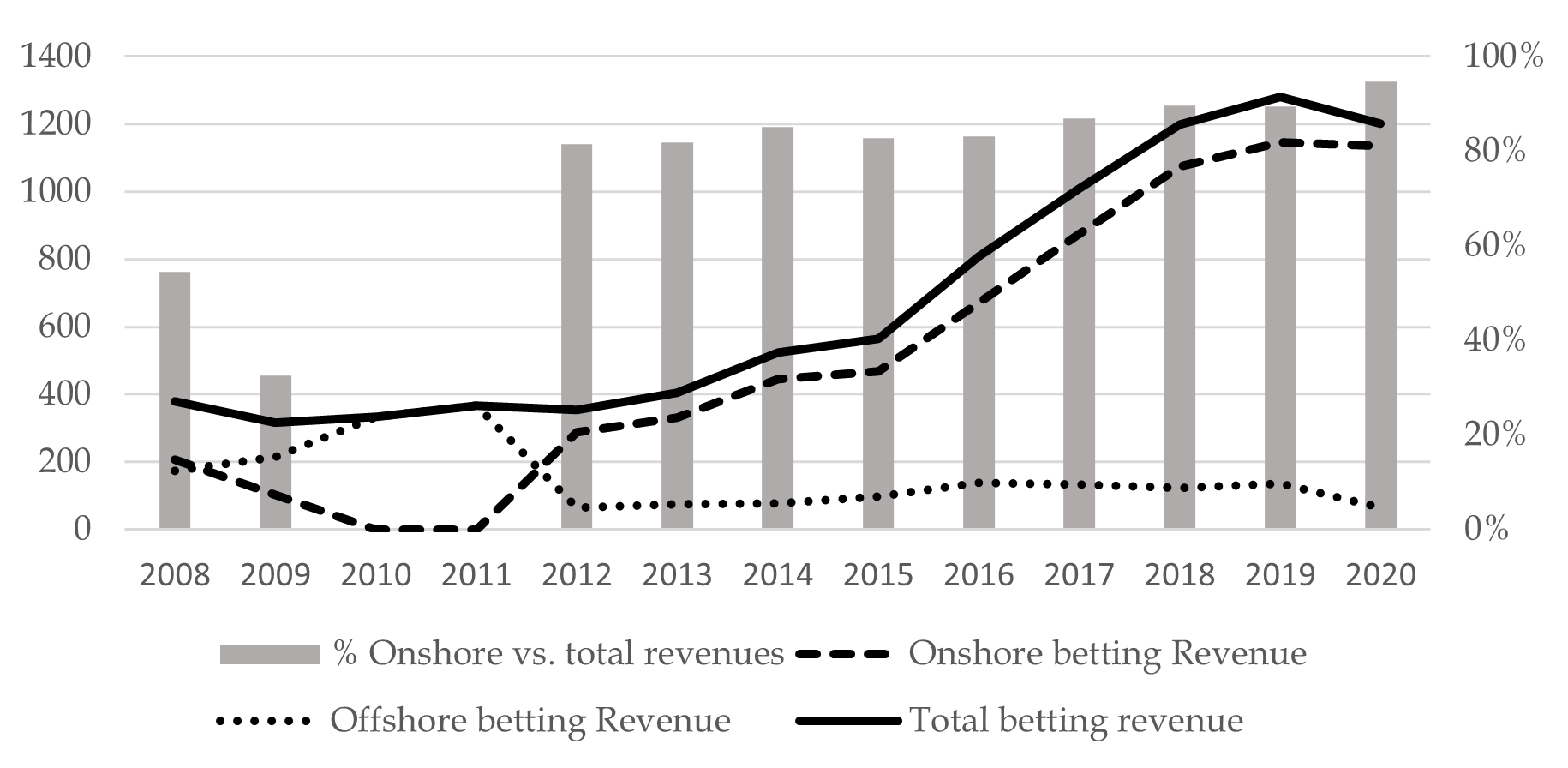}
		\begin{minipage}{15.8cm}  \footnotesize  \emph{Notes:} 
This figure illustrates the annual total gross online betting revenues between 2008 and 2020 in Germany (in \euro{}M), disaggregated into onshore and offshore betting revenues. The bars illustrate the share of onshore betting revenues. Onshore revenues refer to all sports betting activities for which taxes were paid in Germany. For 2012, the onshore revenues refer to all revenues by agencies that paid sports betting taxes in 2012. Offshore revenues include all sports betting revenues by agencies that target the German market but do not pay the sports betting tax in Germany. Gross betting revenues are equal to the total wagered amount (including bonuses) minus all winnings by bettors, disaggregated by onshore and offshore agencies. Data was provided by H2 Gambling Capital (July 2021) and is primarily based on official data from official gambling regulatory bodies, governments, and tax authorities. The data only includes revenues of operators licensed in at least one jurisdiction, and entirely unlicensed activities are not included.
		\end{minipage}
	\end{figure}

As illustrated by the bars in Figure \ref{fig:bettingrevenues_ger}, the share of taxed betting revenues is constantly over 80\% after the tax reform in 2012. The share shows a slightly increasing trend. One can conclude that the vast majority of betting agencies that target the German market also paid the sports betting tax, despite a missing jurisdiction in Germany. Annual reports of listed betting agencies regularly discuss the implications of the German sports betting tax for their business \citep[e.g.][]{bwin2013annual}, confirming this conclusion. 


While this finding seems rather trivial, it is noteworthy in the context of digital taxation. The German sports betting tax shares several features with digital service taxes that aim to tax domestic revenues of international digital businesses.\footnote{Online betting services are a good example of a highly digitalized business model. All three factors of highly digitalized business models, as defined by the \citet{oecd2014beps} apply to the online betting market: i) Cross-jurisdictional scale without mass, ii)  Reliance on intangible assets, including intellectual property (IP). iii)  Data, user participation, and their synergies with IP.} 
Several countries introduced such digital service taxes in recent years to protect their domestic tax base from extensive profit shifting by digital firms that typically do not require a physical presence in a specific country to serve domestic customers \citep{oecd2014beps}.\footnote{For an overview of the current state of digital services taxes in different countries, see \citet{kpmg2022tax}.}

Operating from an offshore location implies that monitoring and prosecution by public authorities are very difficult or even impossible. Thus, one objection to introducing digital taxes, at least in some sectors, is that digital taxes are unlikely to raise significant tax revenues as most firms that are not registered with the tax authorities will also not pay domestic taxes. In the setting at hand, firms do not only operate from jurisdiction to avoid taxes but also to circumvent national bans, making the argument particularly applicable in our setting. In fact, many observers of the debate questioned in the beginning whether 5\% on stakes would be low enough that private betting agencies would cooperate with the German tax authorities \citep{englisch2013taxation}.

My findings suggest that taxing online sports betting and digital services can generally work, in the sense, that they can generate considerable tax revenues even if the taxed entity has no jurisdiction within the respective country. Agencies seem not to evade the German sports betting tax on a large scale.\footnote{Tax evasion has important implications for the relationship between statutory incidence and the economic incidence of the tax \citep[e.g.][]{kopczuk2016does}.} 



\subsection{Overall effect on consumer prices}

The results for the average tax effects can be found in Table \ref{diffindiff_soccer} and \ref{diffindiff_allsports}. The estimated average pass-through ($\frac{dp}{dt}$) in our main specification, which considers all agencies and leagues for soccer games and includes week, league, and agency fixed effects (Column 2, Table \ref{diffindiff_soccer}), is equal to 76\% ($\beta_1$=0.038 with a standard error (SE) of 0.004). The average pass-through increases to 82\% if we exclude "cross" leagues. The estimated pass-through rates remain robust to only considering agencies with odds data available for all years or using betting prices for all sports events (Table \ref{diffindiff_allsports}). Furthermore, the average pass-through rate remains unchanged if we consider a TWFE model with only time and agency fixed effects (Columns 1 and 4, Table \ref{diffindiff_soccer}) or if we include league-by-agency fixed effects instead (Columns 3 and 6, Table \ref{diffindiff_soccer}).

\input{Tables/DiffinDiff_soccer}

Results for the dynamic tax effects of the tax reform, illustrated by the connected black squares, can be found in Figure \ref{fig:coef_soccer}. The grey shaded area illustrates the respective 95\% confidence intervals. Coefficients for all periods before the tax reform are statistically insignificant. The average quarterly pass-through rates steadily increase in the first four quarters after the tax reform to around 0.035 (70\%). This development is associated with the staggered introduction of different shrouding policies that seem to help agencies pass a larger share of the tax onto consumers. After one year, the pass-through rate stays at around 80\% for several quarters with a slightly increasing trend. At the end of the observation period, the estimated $\beta_k$-coefficients reach around 0.045. In general, the pass-through rate of the tax seems very persistent over time. Similar to the average effects, the coefficients slightly increase if we exclude cross-leagues, showing the same general pattern.

\input{Figures_tex/TWFE_all}

Figure \ref{fig:coef robust} in the Appendix shows the estimated dynamic tax effects for samples that only considers agencies with a complete set of odds data (Panels c-f), and events for all sports (Panel a and b). In general, the pattern of the tax effect dynamics are similar to the main specification. The estimated coefficients tend to be slightly smaller if we consider all sports. Pass-through rates seem to be slightly higher for soccer games if we only consider "complete" agencies, reaching nearly a full-pass through ($\beta=0.05$) in 2018. In general, not only the coefficients, but also the confidence intervals are quite stable across different specifications and samples. To check whether the results are driven by outliers, I also estimate equation (5) based on trimmed data (1-99\% percentiles). The results are robust and can be found in Figure \ref{fig:coef_tr} in the Appendix.

The estimated average and dynamic pass-through rates are higher for specifications with the more restrictive control group that only includes agencies that operate from a foreign nation-specific domain. The pass-through rates are close to a full-pass-through for these specifications (see Table \ref{diffindiff_robust} in the Appendix). However, these slightly larger pass-through rates can be mainly attributed to higher estimates in later periods (Panel c and d of Figure \ref{fig:avgprices_all_robust}) that may be partly driven by specific developments in these local (control) markets. According to the domain identifier, 12 out of 22 agencies that use a country-specific domain are located in Eastern Europe (slovakia, Czech Republic, Poland and Russia), four agencies use an Italian, three a French and three a Spanish  domain.

The estimated average and dynamic pass-through rates are higher for specifications with the more restrictive control group that only includes agencies that operate from a foreign nation-specific domain. The pass-through rates are close to a full pass-through for these specifications (see Table \ref{diffindiff_robust} in the Appendix). However, these slightly larger pass-through rates can be mainly attributed to higher estimates in later periods (Panel c and d of Figure \ref{fig:avgprices_all_robust}) that may be partly driven by specific developments in these local (control) markets.

\subsection{Heterogeneity in tax effects}

Table \ref{diffindiff_het_pol} shows the results of the subsample analyses. There is large heterogeneity in pass-through rates depending on the agencies' shrouding policy (see Section \ref{sec:tax_shrouding}). The estimated average pass-through in the subsample that only considers treated agencies that do not shroud taxes for the entire observation period is about 16\% (Column 3, Table \ref{diffindiff_het_pol}). In contrast, agencies that shroud taxes can pass most of the tax onto consumers. Excluding all "transparent" agencies from the sample, I estimate a pass-through of 82\% (Column 2, Table \ref{diffindiff_het}. These differences get even more pronounced when excluding "cross" leagues or extending the data set to non-soccer events. For instance, the $\beta_1$ coefficients for the "No shrouding"-subsample is only slightly positive and insignificant (see Panel B, Table \ref{diffindiff_het_pol}) when all events are used in the estimations.

Interestingly, there is also heterogeneity in estimated pass-through rates along the two shrouding policies (columns 2 and 3, Table \ref{diffindiff_het_pol}). Agencies employing the shrouding policy that deducts the tax from the winnings increase their betting prices most significantly in response to the tax, with an estimated tax effect of 0.043. In contrast, the estimated tax effect for agencies that deduct the surcharge directly from the wager only equals 0.036. This difference is considerably larger than one would expect, given the policies' different $\tau$ (see Section \ref{sec:sport_app_tau}). The difference between the shrouding policies could be explained by a lower tax salience when the surcharge is framed as being only due in the case of winning. While it does not make a big difference for the effective betting prices, the framing may exploit a present bias where tax surcharges are weighted less if they are deducted only in the case of winning. 

\input{Tables/DiffinDiff_het_pol}

The results from the subsample analyses are confirmed by estimating equation \ref{eq:tax_avg_het} (Columns 3 and 6, Table \ref{diffindiff_het}). The coefficient on $T_{i,m}$ can be interpreted as the effect of the tax on betting prices when a shrouding policy is in place for a specific event $m$ and bookmaker $i$. In contrast, the coefficient of the interaction term shows the difference between the coefficient on $T_{i,m}$ and the tax effect without shrouding. In line with previous findings that showed a strong positive association between the effect on consumer betting prices and active shrouding policies, the estimated tax effect is around 0.046 when a shrouding policy was in place for a specific event-bookmaker observation. If we exclude cross-leagues, the estimated coefficients increase to 0.05, implying a full pass-through of the tax onto consumers ($\frac{dp}{dt}=1$). The estimated coefficients on the interaction term are equal to -0.041 and -0.043, respectively, which implies a highly statistically significant difference in pass-through rates of over 80\%. The estimated effects are robust to only considering agencies, for which odds data is available for all years (Table \ref{diffindiff_het_compl} in the Appendix) or considering all sports (Table \ref{diffindiff_het_allsport} in the Appendix).

Panel (a) of Figure \ref{fig:coef_soccer_het} illustrates the estimation results of the event study design, considering only treated agencies that never shroud tax surcharges (blue line) and shrouding agencies (red line). The control group is the same for both subsample regressions. Importantly, in both subsamples, the pre-trends show a similar pattern as the control group, suggesting that the parallel trend assumption holds for both subsamples. Only after the tax reform estimated treatment effects diverge between the two subsamples. While other explanations cannot be entirely ruled out, this evidence highly suggests that shrouding policies cause differences in tax effects. 

Panel (b) shows the estimation results for specifications with leads and lags of treatment interacted with the $no Shroud_{i,m}$ dummy (equation \ref{eq:tax_dyn_het}). The coefficients are calculated such that the red connected dots show the tax effects over time for a specific event and bookmaker with an active shrouding policy. The blue connected dots show the dynamic tax effects for events where agencies did not employ shrouding policy was employed. In contrast to the subsample analysis, the estimated quarterly tax effects jump to around 0.045 directly after the tax reform and stay pretty constant until the end of the observation period. This immediate increase can be attributed to the fact that not all shrouding agencies immediately employ a shrouding policy after the reform. The tax effects for ``no shrouding''-events are insignificant in the first quarterly periods after the tax reform and steadily increase until the end of the observation period to around 0.02 in 2018, similar to the subsample analysis. The results are very similar if we exclude cross leagues or consider all events (Figure \ref{fig:coef_het_robust} in the Appendix). 

\input{Tables/DiffinDiff_het_new}

\input{Figures_tex/TWFE_het}

%% file: Tables/DiffinDiff_soccer.tex
\begin{table}[ !htbp]\centering
\def\sym#1{\ifmmode^{#1}\else\(^{#1}\)\fi}
\caption{Avg. effect of tax on consumer betting prices - soccer matches}
\label{diffindiff_soccer}
\scalebox{0.9}{
\begin{threeparttable}
\begin{tabular}{l*{7}{c}}
\hline\hline

                 &\multicolumn{3}{c}{All Leagues}     & &\multicolumn{3}{c}{Excl. "cross" leagues}     \\    \cmidrule{2-4} \cmidrule{6-8}
                    &\multicolumn{1}{c}{(1)}      &\multicolumn{1}{c}{(2)}      &\multicolumn{1}{c}{(3)}   &   &\multicolumn{1}{c}{(4)}      &\multicolumn{1}{c}{(5)}      &\multicolumn{1}{c}{(6)}      \\

\midrule

\textbf{Panel A: All agencies} &&&&&&& \\
[1em]
Tax effect on prices             &               0.038\sym{***}&               0.038\sym{***}&               0.038\sym{***}&   &            0.041\sym{***}&               0.041\sym{***}&               0.041\sym{***}\\
                    &             (0.004)         &             (0.004)         &             (0.004)       &  &             (0.005)         &             (0.005)         &             (0.005)         \\

\hline
Observations        &             1,936,322         &             1,936,322         &             1,936,322       &  &             1,290,843         &             1,290,843         &             1,290,843         \\
\(R^{2}\)           &               0.642         &               0.759         &               0.811      &   &               0.643         &               0.754         &               0.809         \\
\midrule
\textbf{Panel B: Comp. agencies} &&&&&&& \\
[1em]
Tax effect on prices           &               0.038\sym{***}&               0.038\sym{***}&               0.038\sym{***}&  &             0.041\sym{***}&               0.041\sym{***}&               0.041\sym{***}\\
                    &             (0.005)         &             (0.005)         &             (0.005)      &   &             (0.005)         &             (0.005)         &             (0.005)         \\

\hline
Observations        &             1,276,400         &             1,276,400         &             1,276,400         &    &          778,207         &              778,207         &              778,207         \\
\(R^{2}\)           &               0.521         &               0.699         &               0.747         &   &            0.467         &               0.660         &               0.711         \\

\midrule

Constant            &                 Yes         &                 Yes         &                 Yes       &  &                 Yes         &                 Yes         &                 Yes         \\

Time FE             &                 Yes         &                 Yes         &                 Yes    &     &                 Yes         &                 Yes         &                 Yes         \\

Agency FE           &                 Yes         &                 Yes         &                  No       &  &                 Yes         &                 Yes         &                  No         \\

League FE           &                  No         &                 Yes         &                  No     &    &                  No         &                 Yes         &                  No         \\

League-agency FE    &                  No         &                  No         &                 Yes      &   &                  No         &                  No         &                 Yes         \\
\hline \hline

\end{tabular}
\begin{tablenotes}[flushleft]
\item \emph{Notes:} This table reports the estimated average sports betting tax effects on consumer betting prices of soccer events according to Eq. \ref{eq:tax_avg_effect}, comparing changes in prices before and after the tax reform between German and Non-German agencies. Each agency-event combination (unique observation) is equally weighted. Columns 1-3 consider all leagues, while Columns 4-6 exclude events in German leagues for the control group and Non-German leagues for the treatment group. All estimations include time-fixed effects, and the columns differ in the included agency and league fixed effects. Estimations in Panel A consider observations from all agencies, and Panel B considers agencies with observations from all years only. Robust standard errors clustered at the agency level are reported in the brackets. * denotes significance at the 10-\%, ** at the 5-\%, and *** at the 1-\% level.
\end{tablenotes}
\end{threeparttable}
}
\end{table}

%% file: Figures_tex/TWFE_all.tex
\begin{figure}[!htbp]

    \centering
      \caption{Pre-trends and tax effects on consumer betting prices over time - all agencies}
    
    \begin{subfigure}[t]{0.8\textwidth}
        \centering
        \caption{ Soccer games}
        \includegraphics[width=\linewidth]{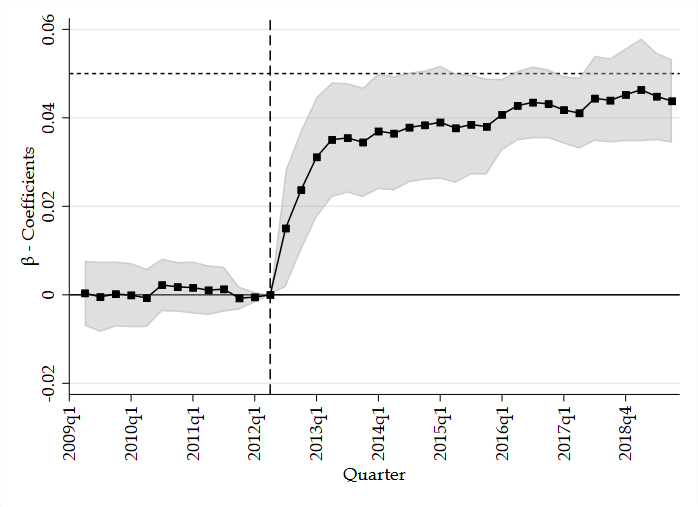} 
        \label{fig:timing1}
    \end{subfigure}
        \vspace{0.5cm}
    \begin{subfigure}[t]{0.8\textwidth}
        \centering
        \caption{Soccer games - excluding cross leagues}
        \includegraphics[width=\linewidth]{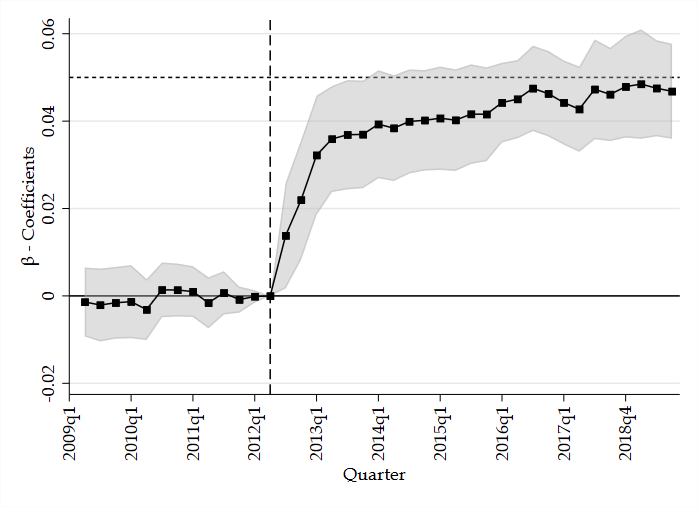} 
         \label{fig:timing2}
    \end{subfigure}
     		\begin{minipage}{15.8cm}  \footnotesize  \emph{Notes:} 
This figure shows the estimated effects (according to Eq. \ref{eq:tax_dyn_effect}) and respective 95-\% confidence intervals of the German sports-betting tax on consumers' betting prices for soccer events over time. All agencies are considered, N=65 (Treatment=10/Control=55). In contrast to Panel (a), Panel (b) excludes German leagues for the control group and Non-German leagues for the treatment group. The vertical dashed line illustrates the introduction of the tax reform. The horizontal dotted line signifies a full pass-through of the tax on consumers.
		\end{minipage} 
		\label{fig:coef_soccer}
  
\end{figure} 

%% file: Tables/DiffinDiff_het_pol.tex
\begin{table}[ !ht]\centering
\def\sym#1{\ifmmode^{#1}\else\(^{#1}\)\fi}
\caption{Avg. effect of tax on consumer betting prices - Subsamples different policies}
    \label{diffindiff_het_pol}
        \scalebox{0.9}{
\begin{threeparttable}
\begin{tabular}{l*{7}{c}}
\hline\hline

                 &\multicolumn{3}{c}{All Leagues}     & &\multicolumn{3}{c}{Excl. "cross" leagues}     \\    \cmidrule{2-4} \cmidrule{6-8}
                 
 "Shrouding" policy                   &\multicolumn{1}{c}{(ii)}      &\multicolumn{1}{c}{(iii)}      &\multicolumn{1}{c}{(i)}   &   &\multicolumn{1}{c}{(ii)}      &\multicolumn{1}{c}{(iii)}      &\multicolumn{1}{c}{(i)}      \\
 [0.3em]
                     &\multicolumn{1}{c}{(1)}      &\multicolumn{1}{c}{(2)}      &\multicolumn{1}{c}{(3)}   &   &\multicolumn{1}{c}{(4)}      &\multicolumn{1}{c}{(5)}      &\multicolumn{1}{c}{(6)}      \\

\midrule
\textbf{Panel A: Soccer} &&&&&&& \\
[1em]
Tax effect on prices              &               0.043\sym{***}&               0.036\sym{***}&               0.008\sym{***}& &              0.046\sym{***}&               0.043\sym{***}&               0.008\sym{**} \\
                    &             (0.003)         &             (0.007)         &             (0.002)         &  &           (0.003)         &             (0.005)         &             (0.002)         \\
[1em]
Constant            &               0.085\sym{***}&               0.085\sym{***}&               0.085\sym{***}&       &        0.091\sym{***}&               0.091\sym{***}&               0.091\sym{***}\\
                    &             (0.003)         &             (0.003)         &             (0.003)         &   &          (0.003)         &             (0.003)         &             (0.003)         \\

\hline
Observations        &             1,793,974         &             1,619,606         &             1,587,090         &   &          1,257,582         &             1,217,911         &             1,211,548         \\
\(R^{2}\)           &               0.764         &               0.743         &               0.728         &   &            0.754         &               0.743         &               0.737         \\

\midrule
\textbf{Panel B: All sports} &&&&&&& \\
[1em]
Tax effect on prices             &               0.044\sym{***}&               0.034\sym{***}&               0.002     &    &               0.041\sym{***}&               0.037\sym{***}&               0.001         \\
                    &             (0.002)         &             (0.007)         &             (0.002)         &   &          (0.002)         &             (0.007)         &             (0.002)         \\
[1em]
Constant            &               0.079\sym{***}&               0.084\sym{***}&               0.084\sym{***}&   &            0.088\sym{***}&               0.088\sym{***}&               0.088\sym{***}\\
                    &             (0.003)         &             (0.003)         &             (0.003)      &   &             (0.003)         &             (0.003)         &             (0.003)         \\

\midrule
Observations        &             2,815,902         &             2,512,623         &             2,466,796         &    &         1,974,102         &             1,912,848         &             1,901,791         \\
\(R^{2}\)           &               0.745         &               0.716         &               0.703         &    &           0.736         &               0.721         &               0.713         \\
\midrule
Time FE             &                 Yes         &                 Yes         &                 Yes         &         &        Yes         &                 Yes         &                 Yes         \\
Agency FE           &                 Yes         &                 Yes         &                 Yes         &      &           Yes         &                 Yes         &                 Yes         \\
League FE           &                 Yes         &                 Yes         &                 Yes         &   &              Yes         &                 Yes         &                 Yes         \\
\hline\hline

\end{tabular}
\begin{tablenotes}[flushleft]
\item \emph{Notes:} This table reports the estimated average sports betting tax effects on consumer betting prices according to Eq. \ref{eq:tax_avg_effect}, comparing changes in prices before and after the tax reform between German and Non-German agencies. Columns 3 and 6 only consider German agencies that do not shroud taxes. Likewise, columns 1 and 4 (2 and 5) only consider agencies that deduct a 5\% tax surcharge from the advertised winnings (from the advertised wager). Each agency-event combination (unique observation) is equally weighted. Columns 1-3 consider all leagues, while Columns 4-6 exclude events in German leagues for the control group and Non-German leagues for the treatment group. All estimations include time-, agency- and league-fixed effects. Estimations in Panel A consider observations from soccer matches, and Panel B considers observations from all sports. Robust standard errors clustered at the agency level are reported in the brackets. * denotes significance at the 10-\%, ** at the 5-\%, and *** at the 1-\% level.
\end{tablenotes}
\end{threeparttable}
}
\end{table}

%% file: Tables/DiffinDiff_het_new.tex
\begin{table}[!htbp]\centering
\def\sym#1{\ifmmode^{#1}\else\(^{#1}\)\fi}
\caption{Avg. effect of tax on consumer betting prices - soccer}
    \label{diffindiff_het}
    \scalebox{0.9}{
\begin{threeparttable}
\begin{tabular}{l*{7}{c}}
\hline\hline

                 &\multicolumn{3}{c}{All Leagues}     & &\multicolumn{3}{c}{Excl. "cross" leagues}     \\    \cmidrule{2-4} \cmidrule{6-8}
                                  &\multicolumn{2}{c}{Subsamples}     & Interact. & &\multicolumn{2}{c}{Subsamples} &  Interact.  \\    \cmidrule{2-3} \cmidrule{6-7}
                    &\multicolumn{1}{c}{No shrouding}      &\multicolumn{1}{c}{Shrouding}      &   &   &\multicolumn{1}{c}{No shrouding}      &\multicolumn{1}{c}{Shrouding}      &   \\
                    &\multicolumn{1}{c}{(1)}      &\multicolumn{1}{c}{(2)}      &\multicolumn{1}{c}{(3)}   &   &\multicolumn{1}{c}{(4)}      &\multicolumn{1}{c}{(5)}      &\multicolumn{1}{c}{(6)}      \\
\midrule

Tax effect on prices             &               0.008\sym{***}&               0.041\sym{***}&                             &   &            0.008\sym{**} &               0.045\sym{***}&                             \\
                 &             (0.002)         &             (0.003)         &                             &   &          (0.002)         &             (0.003)         &                             \\
[1em]
$T_{i,m}$           &                             &                             &               0.046\sym{***}&   &                          &                             &               0.050\sym{***}\\
                    &                             &                             &             (0.003)         &  &                           &                             &             (0.003)         \\
[1em]
$T_{i,m}$ x $no Shroud_{i,m}$            &                             &                             &              -0.041\sym{***}&  &                           &                             &              -0.043\sym{***}\\
                    &                             &                             &             (0.002)         &  &                           &                             &             (0.002)         \\
[1em]
Constant            &               0.085\sym{***}&               0.084\sym{***}&               0.084\sym{***}& &              0.091\sym{***}&               0.091\sym{***}&               0.091\sym{***}\\
                    &             (0.003)         &             (0.003)         &             (0.003)         & &            (0.003)         &             (0.003)         &             (0.003)         \\
\midrule
Time FE             &                 Yes         &                 Yes         &                 Yes         &  &               Yes         &                 Yes         &                 Yes         \\
Agency FE           &                 Yes         &                 Yes         &                 Yes         &  &               Yes         &                 Yes         &                 Yes         \\
League FE           &                 Yes         &                 Yes         &                 Yes         &   &              Yes         &                 Yes         &                 Yes         \\
\hline
Observations        &             1,587,090         &             1,881,406         &             1,936,322         &   &          1,211,548         &             1,277,394         &             1,290,843         \\
\(R^{2}\)           &               0.728         &               0.767         &               0.779         &  &             0.737         &               0.757         &               0.761         \\
\hline\hline

\end{tabular}
\begin{tablenotes}[flushleft]
\item \emph{Notes:} Columns 1, 2, 4, and 5 of the table report the estimated average sports betting tax effects on consumer betting prices for soccer events according to Eq. \ref{eq:tax_avg_effect}, comparing changes in prices before and after the tax reform between German and Non-German agencies. Columns 1 and 4 only consider German agencies that do not shroud taxes, while columns 2 and 5 only consider shrouding agencies (either policy ii or iii). Columns 3 and 6 show the estimation results for Eq. \ref{eq:tax_avg_het}. The coefficient on $T_{i,m}$ can be interpreted as the effect of the tax on betting prices when a shrouding policy is in place for a specific event $m$ and bookmaker $i$. In contrast, the coefficient of the interaction term shows the difference between the coefficient on $T_{i,m}$ and the tax effect without shrouding. Each agency-event combination (unique observation) is equally weighted. Columns 1-3 consider all leagues, while Columns 4-6 exclude events in German leagues for the control group and Non-German leagues for the treatment group. All estimations include time-, agency- and league-fixed effects. Robust standard errors clustered at the agency level are reported in the brackets. * denotes significance at the 10-\%, ** at the 5-\%, and *** at the 1-\% level.
\end{tablenotes}
\end{threeparttable}
}
\end{table}

%% file: Figures_tex/TWFE_het.tex
\begin{figure}[H]

    \centering
      \caption{Pre-trends and tax effects on consumer betting prices over time for different "shrouding" policies- all agencies}
    
    \begin{subfigure}[t]{0.7\textwidth}
        \centering
        \caption{Subsample analysis (agency level)}
        \includegraphics[width=\linewidth]{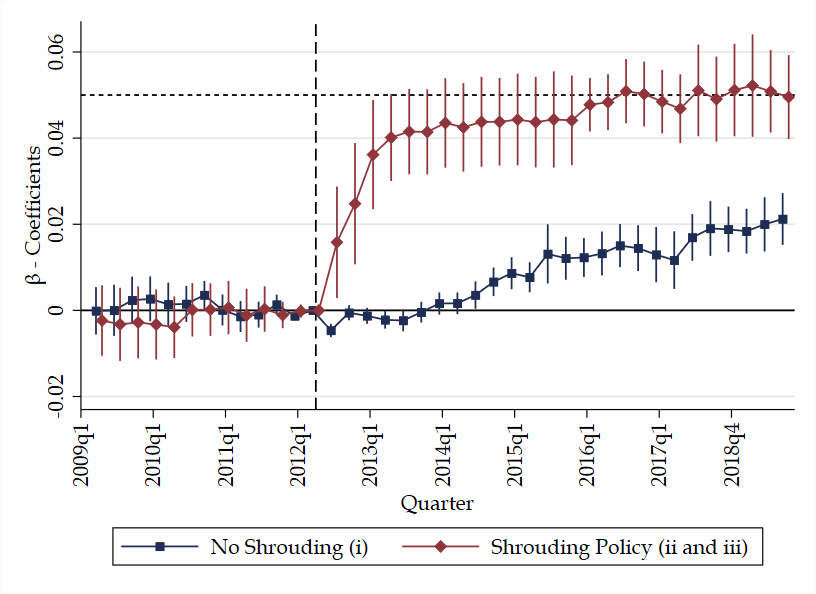} 
        \label{fig:timing1}
    \end{subfigure}
        \vspace{0.5cm}
    \begin{subfigure}[t]{0.7\textwidth}
        \centering
        \caption{Interaction terms (event-agency level)}
        \includegraphics[width=\linewidth]{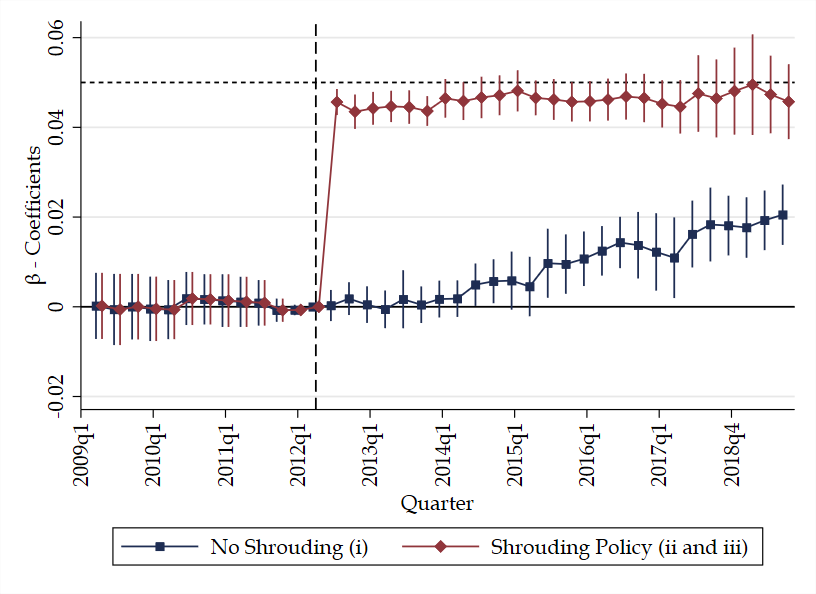} 
        \label{fig:timing1}
    \end{subfigure}
     		\begin{minipage}{13.8cm}  \footnotesize  \emph{Notes:} 
This figure shows the estimated effects and respective 95-\% confidence intervals of the German sports-betting tax on consumers' betting prices over time. Panel (a) shows the estimates for the event study design according to equation \ref{eq:tax_dyn_effect}. The red estimates only consider German agencies that employ a shrouding policy at some point, and the blue estimates only consider German agencies that never employ a shrouding policy. The control group is the same for both regressions. Panel (b) shows the coefficients estimating Eq \ref{eq:tax_dyn_het} that interacts leads and lags of treatment with a dummy that equals one if the event took place after the tax reform and agency $i$ had a NO shrouding policy in place for event $m$. Coefficients are transformed such that both groups can be directly compared. All agencies are considered, N=65 (Treatment=10/Control=55). The vertical dashed line illustrates the introduction of the tax reform. The horizontal dotted line signifies a full pass-through of the tax on consumers.
		\end{minipage} 
		\label{fig:coef_soccer_het}
  
\end{figure} 

%% file: Chapters/Theory.tex
\section{Optimal sin taxation}\label{sec:theory_sin}

The goal of this section is to derive the policy and welfare implications of the empirical results and understand the underlying mechanism behind and effects of strategic tax shrouding. I use the optimal sin taxation framework by \citet{odonoghue2006optimal} and relax their assumptions of perfect competition and fully attentive consumers to accommodate the main features of online sports betting markets. 

\subsection{Optimal sin tax model with homogeneous consumers and perfect competition}

The main rationale behind sin taxes is that certain goods are ``overconsumed'' because (some) consumers do not internalize negative externalities on one's (future) self because of self-control problems, e.g. due to time-inconsistent preferences \citep{laibson1997}, or other frictions when making a decision. In these cases, consumer do not optimize their own true ''experienced'' utility or welfare $u$ with respect to this good's consumption but base their decision on some ``decision'' utility $\Tilde{u}$ that differ from the true utility. Following the arguments for Pigouvian taxes \citep{pigou1920economics}, taxing the consumption of sin goods can thus induce consumers to internalize negative consumption effects and increase social welfare. 

Consider a variant of the model by \citet{odonoghue2006optimal}. There is a unit mass of consumers $i \in  [0,1]$ that consume two goods: a sin good $x$, say sports betting, and a composite good $z$. The composite good is produced with constant return to scale and supplied in a perfectly competitive market. The price of the composite good is normalized to one. Sin good $x$ gives consumer $i$ an (immediate) enjoyment from consumption $v_i(x)$ but come with (future) costs due to potential issues related to problem gambling, $c_i(x)$. Assume following quasi-linear utility function: $u_i=v_i(x_i)-c_i(x_i)+z$. Here, $u_i$ refers to``long-run'' utility of consumption, which eliminates the time dimension. $v_i$ is increasing in $x$ with decreasing marginal benefits. Potential costs $c$ are increasing in $x$ with no constraints on the marginal costs beside that $v_{xx}-c_{xx}<0$, guaranteeing well-behaved behavior. 

Suppose that consumers may not fully internalize the costs from consumption of the sin good and discount the long-run costs of the sin good with $0 \leq \gamma_i \leq 1$, implying that $\Tilde{u}_i=v_i(x_i)-\gamma_i c_i(x_i)+z$. Consumers are endowed with per period income $W$ that is large relative to the sin good consumption. The government can impose a per unit tax $t$ on the consumption of $x$ remitted by firms that finances a lump sum transfer $L$ to the consumers. The effective tax-inclusive consumer price is $p$ whereas the producers' price is $q=p-t$. The sin good is produced at constant marginal costs $m$. Assume for now that all consumers are homogeneous. Accordingly, each consumer solves following maximization problem:

\begin{equation}\label{eq:dec_util}
   \underset{x, z}{max\,} { \Tilde{u}=v(x)-\gamma c(x)+z}
\end{equation}
\centerline{s.t. $p x + z \leq W + L$}\\

If the market for the sin good is perfectly competitive, producer prices will equal marginal costs (including $t$), such that $p^*=m+t$. Accordingly, the tax will be fully passed onto consumers ($\frac{dp^*}{dt}=1$).\footnote{The full pass-through of the tax is the result of the constant marginal costs assumption, which implies a perfectly elastic supply. In the standard partial equilibrium model of taxation under perfect competition, the pass-through rate of the tax will be a function of demand and supply elasticities, where the relatively less elastic side of the market will bear most of the tax burden (see Appendix \ref{sec:app_standard} for details).} 

Without taxes ($t=0$) and no self-control problems ($\gamma=1$) the optimal consumption bundle $(x^*, z^*)$ is determined by the first order conditions: $v_x(x^*)-c_x(x^*)=m$ and $z^*=W-m x^*$.\footnote{Note that the optimal consumption bundle also depends on $m$, I simplify the notation in the interest of better readability.} This case provides the first-best solution that maximizes social welfare. Without taxes but $\gamma<1$, consumers ignores $1-\gamma$ of the costs. They set $\Tilde{x}^*(t)$ according to their decision utility, such that: $v_x(\Tilde{x}^*(t))-\gamma c_x(\Tilde{x}^*(t))=m+t$ and $z^*(t)=W+L-(m+t) \Tilde{x}^*(t)$. This implies that consumers overconsume the sin good if $t=L=0$: $\Tilde{x}^*(0)>x^*$. Similar to a simple Pigouvian tax-transfer scheme designed to correct for externalities, a tax $t^*=(1-\gamma)c_x(x^*)$ on sports betting and a uniform transfer of the tax revenue to consumers would implement the first-best solution, where all consumers choose $(x^*, z^*)$. The proof is straightforward and provided in Appendix \ref{sec:app_proof_perfcomp}, Lemma \ref{lem:app_proof_perfcomp}.

Suppose now that consumers underreact to taxes and misperceive the tax being equal to $\theta t$. $\theta<1$ can be interpreted as an exogeneous tax salience parameter similar to \citet{chetty2009salience, dellavigna2009psychology}. In this case consumers choose $\Tilde{x}_{\theta}^*(t)$ such that: $v_x(\Tilde{x}_{\theta}^*(t))-\gamma c_x(\Tilde{x}_{\theta}^*(t))=m+\theta t$ and $\Tilde{z}_{\theta}^*(t)=W+L-(m+t) \Tilde{x}_{\theta}^*(t)$.\footnote{Note that for $\Tilde{z}_{\theta}^*(t)$ there is no $\theta$ on the right-hand side of the equation, as consumers will spend the remaining ``true'' income on the composite good. $W$ is assumed to be large enough such that the consumers do not breach the budget constraint for any $\theta$ and $\gamma$.} The corrective effect of the tax is attenuated and the optimal tax rate increases. The optimal sin tax that reinstates the social optimal solution is equal to (see Appendix \ref{sec:app_proof_perfcomp}, Lemma \ref{lem:app_proof_perfcomp}): 
\begin{equation}\label{eq:opt_salient}
t_{\theta}^*=\dfrac{(1-\gamma)c_x(x^*)}{\theta}  
\end{equation} 
The sin tax is completely ineffective in correcting behavior if $\theta=0$. As tax revenues are completely transferred back to consumers and profits are zero irrespective of the tax, the sin tax is irrelevant from a social welfare perspective in this extreme case. 

The optimal tax formula also applies if there is an overreaction to the tax $\theta>1$ and the maximization problem of the social planner has an interior solution. In addition, all results of the homogeneous case also hold if additional to---or instead of---internalities, sports betting exhibits externalities to other consumers. The Pigou \emph{dollar principle} still applies: A dollar of internalities and/or externalities must be corrected with $\frac{1}{\theta}$ tax dollars \citep{farhi2020optimal}. The basic results hinge on the assumption that consumers are homogeneous. \citet{odonoghue2006optimal} study optimal sin taxation with perfectly attentive consumers that may vary in their tastes and $\gamma$. I will focus on consumers that are heterogenous in terms of their attention towards shrouded taxes in the next section.

\subsection{Optimal sin tax model with heterogeneous consumers and tax shrouding}

Now suppose that the sin good is not supplied under perfect competition but under imperfect Bertrand-Nash price competition. The supply side of the model follows \citeauthor{varian1980model}'s model of sales (\citeyear{varian1980model}) and shares several features with the model by \citet{armstrong2009inattentive} that study the optimal choice of binary product quality in the presence of inattentive consumers. 

$k=1,2,...,N$ firms produce the sin good at the same constant marginal costs $m$. Firms set a salient posted price $p^s_k$ and a shrouded tax surcharge $\tau_k$, which can be either zero or equal to the unit tax $t$ ($\tau_k \in \{0, t\}$). 
Accordingly, the firm $k$'s consumer price is defined as $p_k=p^s_k+\tau_k$ and producers prices are $p_k-t$. Supply is completely elastic.

On the demand side, suppose that only a fraction $\lambda$ of consumers are attentive. They perfectly observe all effective prices (including shrouded tax surcharges) and the shrouding decision by firms. Buying from a shrouding firm may come with a fixed disutility $s\geq0$, which may refer to attention or search costs, or behavioral costs related to the attempt of being ripped off. 
The remaining fraction of consumers ($1-\lambda$) are inattentive to shrouded taxes. They perfectly observe posted prices but misperceive the shrouded tax surcharge being equal to $\theta t$ if $\tau_k=t$. In the case where all consumers are entirely inattentive to shrouded taxes ($\theta=0$), consumers act on the posted price alone, i.e., they mistakenly believe that there are no tax surcharges. Demand from inattentive consumers is split equally among firms with the lowest posted price.
Firms cannot discriminate between consumer types and set one pricing strategy for all consumers. All demand function will further depend on competitors' prices and shrouding decisions. 

\subsubsection*{Homogeneous (in)attention}

As a benchmark consider following three boundary cases which results in some sort of marginal cost pricing. First, without taxes both types of consumers have perfect information about relevant prices and will simply buy from the firm with the lowest posted (i.e., effective) price. Following a price undercutting argument \`{a} la \citet{bertrand1883theorie}, this implies a unique symmetric equilibrium where all firms set prices equal to marginal costs $m$. Firm's market shares are equal to $\frac{1}{N}$ and profits are equal to zero independent of $\gamma$. With no self-control problems ($\gamma=1$), the optimal consumption bundle $(x^*,z^*)$ is chosen. With self-control problems, the demand for the sin good is higher, i.e., $x$ will be overconsumed from a social welfare perspective.

Second, suppose all consumers are fully attentive to the tax surcharge ($\lambda=1$). In that case, all firms set symmetric equilibrium prices equal to firms' marginal costs $p^*=m+t$. Sin taxes are fully passed through onto consumer prices ($\frac{dp^*}{dt}=1$). All firms will split the market equally and equlibrium profits are zero. If $s>0$, firms do not shroud taxes in equilibrium. 
The implications of the classical corrective taxation model under perfect competition discussed above carry over to the extended model in this case.  

Third, if all consumers underreact to shrouded taxes, firms optimally shroud taxes ($\tau_k=t$) in equilibrium. As consumers are homogeneous, the Bertrand price undercutting argument applies, and firms will set symmetric posted prices equal to $m$, implying zero equilibrium profits. Firms effectively still fully pass through taxes onto consumers (through shrouded taxes); the corrective effect of the tax is, however, attenuated by $\theta$. The optimal sin tax results under a homogeneous underreaction and perfectly competitive markets apply, and the optimal tax rate is given by Eq. \ref{eq:opt_salient}.

\subsubsection*{Heterogeneous attention}

With both types of consumers ($0<\lambda<1$) and $\theta<1$, there is no equilibrium in pure strategies where all firms unshroud taxes. (Lemma \ref{lem:nopure} in Appendix \ref{sec:app_proof_imperfcomp}). This implies that tax shrouding prevails in equilibrium with some inattentive consumers, and the sin tax is less effective, on average, as in the base case. If $s>0$, there is further no symmetric pure strategy equilibrium (Lemma \ref{lem:nopure} in Appendix \ref{sec:app_proof_imperfcomp}). 

However, if $N\geq4$, a ``segmented'' market equilibrium exists (Lemma \ref{lem:segment} in Appendix \ref{sec:app_proof_imperfcomp}). In the first market segment, firms shroud taxes, set $p^{s*}_k=m$, and split the demand of all inattentive consumers. In the second market segment, firms do not shroud taxes, $p^{s*}_k=m+t$, and sell their products to all attentive consumers. The equilibrium profits are zero in both segments, as in the boundary cases with homogeneous attention discussed above. \footnote{Segmented market equilibria can also arise in differentiated consumer search models in the spirit of \citet{wolinsky1986true} where firms set prices and produce either high or low-quality products \citep{gamp2022competition}. Low quality can be interpreted as a product with shrouded surcharges. While the equilibrium predictions are richer, the main policy implications concerning optimal corrective taxation qualitatively still apply. In \citet{gamp2022competition}, firms earn positive profits in both segments. Low-quality firms only serve inattentive consumers but earn higher markups. High-quality firms serve both types of consumers, which implies higher market shares but lower effective markups. The equilibrium outcomes will primarily depend on the share of inattentive consumers, the search costs, and the degree of product heterogeneity.} The tax will have a different corrective effect in each market segment. In this case, a (consumer-)type-specific sin tax can implement the first-best welfare optimal solution. However, type-specific tax-transfer schemes are usually not feasible. In the remainder of this section, I thus investigate optimal sin taxation when the government is restricted to a linear tax-transfer scheme that is homogeneous across consumers. 

Throughout the analysis, I assume equal welfare weights for all people. If $t>0$, $\theta<1$ and $\gamma<1$, inattentive consumers consume more of the sin good than attentive consumers, even if self-control problems are homogeneous. Thus, no linear tax-transfer scheme implements the first-best solution from a social planner's perspective.  If the tax is set equal to $t*$, attentive consumers perfectly internalize the internalities and/or externalities but inattentive consumers still overconsume the sin good: $\Tilde{x}_{\theta}^*(t^*)>\Tilde{x}^*(t^*)=x^*$. On the other hand, when the sin tax is set to $t_{\theta}^*$, the sin good consumption of the inattentive consumers will be optimal but attentive consumers ``underconsume'' the sin good compared to the social optimum: $\Tilde{x}_{\theta}^*(t_{\theta}^*)=x^*>\Tilde{x}^*(t_{\theta}^*)$. These two considerations are optimally traded off by the government when setting the sin tax rate. The optimal tax rate increases with the fraction of inattentive consumers. In the stylized model, requiring firms to post tax-inclusive prices would eliminate differences in attention to tax-induced price changes, implying that $t^*$ would effectively implement the social optimum. 

I focus on the segmented market equilibrium when deriving the implications for the optimal sin tax. In related models,  there is usually also a symmetric mixed strategy equilibrium that involves random pricing and implies positive profits for firms. \citet{armstrong2009inattentive} study a similar setup where firms can ``cheat'' inattentive consumers by selling low-quality products. Low quality can be interpreted as a product with shrouded attributes. While their setup assumes uniform demand for the product and that the low-quality product is entirely worthless, the implications of their model provide some important insights for the optimal corrective policy in a mixed strategy equilibrium. 

In the symmetric mixed strategy equilibrium, firms draw prices from a bounded probability distribution involving a cutoff price. Below the cutoff price, a firm targets inattentive consumers by choosing a low quality and lower price. Above the price cutoff, the firm chooses high quality, higher prices and sells to both types of consumers. In this equilibrium, firms make positive profits, implying that the presence of inattentive consumers harms attentive consumers, in contrast to \citet{gabaix2006}. There is always a positive probability that firms choose low quality in equilibrium as long as $\lambda>1$. An education policy that increases the number of attentive consumers tends to decrease the number of low-quality firms in equilibrium. However, the effects of such a policy are ambiguous as profits and aggregate consumer surplus is non-monotonic in the share of attentive consumers. Increasing the number of firms drives firms' profits to zero but does not prevent ``cheating by firms'' by firms. If low quality refers to a product with shrouded tax surcharges, there is, thus, always a positive probability that firms shroud taxes in equilibrium as long as $\lambda>1$, attenuating the average corrective effect of a sin tax. Again, requiring firms to reveal their quality, i.e., posting tax-inclusive prices, would prevent ``cheating'' in equilibrium and increases total welfare.

%% file: Chapters/Discussion.tex
One aspect often disregarded when designing tax policies is that firms cannot only change prices but also decide whether to hide or shroud taxes levied on them from consumers. 
If firms strategically shroud taxes, effective tax-inclusive prices may rise, but the welfare effect of corrective taxes would be impaired as a part of the increase is not salient. 
In this paper, I show that the strategic shrouding of sin taxes, which drives a wedge between posted and effective tax-inclusive consumer prices, is widespread in online betting markets. 
According to the theoretical model, the observed shrouding strategies and tax effects are only attainable if (some) consumers misperceive and underreact to shrouded tax surcharges. 
Consequently, the corrective effect of the tax is likely to be small because most of the high pass-through rates are attributed to increases in shrouded tax surcharges. At the same time, bettors bear most of the ``true'' tax burden due to higher effective prices. 
The best policy response depends on the distribution of attention toward shrouded taxes, and homogeneous corrective tax-transfer schemes are unsuccessful in restoring the welfare optimum if attention is heterogeneous. Banning tax shrouding practices successfully restores the effectiveness of sin taxes independent of these considerations (see section \ref{sec:theory_sin}).

The theoretical and empirical results can be directly applied to externality correction policies, such as carbon taxes or environmental subsidies. Similar to sin taxes, the effectiveness of these policies depends on perceived rather than experienced changes in incentives. If consumers do not internalize shrouded taxes or do not know about subsidies, they will underreact to those policy incentives. However, a critical difference between taxes and subsidies is that firms have the incentive to increase, instead of decrease, the salience of a subsidy because it makes their offers more attractive. The motives of the government and firms align. Consequently, consumers' limited attention to surcharges and firms' strategic incentives to manipulate policy's salience provides a novel argument for why environmental subsidies can be socially preferable to taxes.

One common objection against sin taxes is their regressivity, i.e., low-income consumers tend to bear most of their tax burden \citep{allcott2019regressive}. Past studies suggest that regressivity is not a major concern for (online) sports betting taxes as sports bettors, in contrast to other forms of gambling, tend to be relatively wealthy, and income is not a predictor for excessive gambling \citep{humphreys2012bets, wicker2012examining, hing2016demographic}. However, a correlation between attention and income can have significant implications depending on the correlation's sign. \citet{goldin2013smoke} find that low-income consumers pay more attention to less salient cigarette taxes. In our case, this attention gap would imply that the negative effects of shrouding tend to fall disproportionately strong on high-income individuals, making the sin taxes less regressive. Similarly, if attention and self-control problems are correlated, this can affect the policy implications. For instance, impulsive bettors are likely more prone to self-control problems and, at the same time, less attentive to shrouded price attributes \citep{hing2018spur}. In this case, tax shrouding is even more harmful as the individuals that would benefit from the corrective effect the most react to it the least. A thorough analysis of these effects goes beyond the scope of this paper and provides a promising avenue for future research.\footnote{The theoretical model assumes that self-control problems and income is homogeneous across consumers and thus independent of attention toward shrouded taxes.\citet{odonoghue2006optimal} provide an analysis of optimal sin taxes when consumers are heterogeneous in their tastes and self-control problems concerning sin good consumption.}    

As predicted by the theoretical model, the tax reform has led to market segmentation between non-shrouding and shrouding firms. I do not find significant differences between non-shrouding and shrouding agencies before the reform. According to the model, the tax has introduced a new dimension of product differentiation of an otherwise homogeneous good. In equilibrium, non-shrouding firms target attentive consumers, while shrouding firms target inattentive consumers. Following predictions from related models on binary quality choice \citep{armstrong2009inattentive, gamp2022competition}, profit-maximizing firms trade off higher market share with higher markups. Anecdotal evidence from the annual reports of listed agencies is in line with this conjecture. While the leading agency that did not shroud taxes was able to considerably increase its market share in Germany to over 50\% since the tax reform \citep{tipico2018social}, the market share, implied by market size and revenues of the prior market leader (that chose to shroud taxes in response to the reform) decreased over time \citep{bwin2013annual, bwin2018}.  

Aggregated betting revenues in Germany strongly increased, analogous to other European countries before and after the tax reform. Together with the observed equilibrium prices, this development suggests a limited corrective effect of the betting tax on online sports betting consumption. However, more granular betting revenue data within and outside the German market is needed to make a conclusive empirical statement about the precise corrective effect on consumption. Unfortunately, online betting markets remain very opaque, and granular revenue data is not available. Future studies that conduct surveys with bettors about their betting behavior and attention to the tax or structured interviews with industry experts may address this limitation of the study.

The study examines the pass-through on the average betting prices across all outcomes of an event. However, the setting also offers a promising natural laboratory to study the effects of transaction taxes on asset prices and market efficiency. Betting markets on a single event can be understood as small and complete financial markets with state-contingent claims (bets) that have well-defined state prices (odds) and termination points \citep{thaler1988anomalies, moskowitz2002returns}. Accordingly, the effect of the tax on relative betting prices in single-event markets can have important implications for financial markets. For instance, one may use the unique setting presented in this study to examine whether the sports betting tax lead to less efficient betting markets.

%% file: Chapters/Background_Regulation.tex
\section{Institutional background - details}\label{sec:app_inst_details}

\subsection{Regulation and of online sports betting in Europe}


A coherent regulation of online sports betting has been a complicated task in nearly all European countries. Competencies regarding taxes and regulation of online gambling activities are divided between the European, federal, and state levels, sometimes even at the municipality level. This diffusion of competencies further widened the existing time lag between the regulation of online services and the development of (Internet) technology \citep{laffey2016patriot}. There is no sector-specific EU regulation, and EU member states are, in principle, autonomous in regulating online betting services as long as they respect the fundamental freedoms established under the Treaty on the Functioning of the European Union (TFEU) \citep{haberling2012internet}. In the early 2000s, the common response of European countries was either the prohibition of commercial online gambling, including sports betting, or introducing very restrictive monopoly models that only allowed one or few, primarily state-owned, operators to provide online betting services or lotteries \citep{haberling2012internet}.\footnote{The UK was an exception that updated its regulatory sports betting framework to the new digital era in 2001 \citep[][]{paton2002policy}. } However, these very restrictive regulatory models proved largely unsuccessful in limiting the provision and substantial growth of commercial cross-border online betting services in the EU \citep{ec2012faq}.

The main reason behind the limited success of the sports betting bans was that most online betting providers operate(d) in a grey market where they could meet national demands mostly independent of national regulations. Most international agencies hold sports betting concessions from and are domiciled in low-tax jurisdictions that are parts of the European Union, such as Malta and Gibraltar. This offshore operation has several advantages from their perspective: i) possibility to claim free movement of services in the EU to circumvent national and sub-federal regulations (Article 56 - Treaty on the Functioning of the EU, TFEU); ii) prevent national authorities to monitor and prosecute agencies; iii) minimizing the tax liability of their services, both regarding direct and indirect taxation; iv) added credibility by holding a European license \citep{zborowska2012regulation}. 

The provision of cross-border online gambling services is an economic activity that falls within the scope of the fundamental freedoms of the TFEU, which implies that the cross-border provision of gambling services authorized in one member state is legal \citep{ec2012press}. Member states can, however, impose restrictions justified by public interest objectives, such as consumer protection or fraud prevention. Importantly, these restrictions must be proportionate, non-discriminatory, applied consistently and transparently, and suitable to achieve the pursued objective \citep{haberling2012internet}. These strict principles triggered several judgments by the Court of Justice of the European Union (CJEU) and actions by the \citet[e.g.][]{ec2014} that pushed several member states to update their national online sports betting regulations.\footnote{A list of relevant judgments can be found here: \url{https://ec.europa.eu/growth/sectors/online-gambling/gambling-case-law_en}.} As a consequence of this pressure and the limited success of prohibition of online gambling and resulting missed tax revenue from private sports betting providers, the vast majority of European countries started to liberalize online sports betting. This process to a practical regulatory framework often took several years, as the (admittedly special) example of Germany shows, where an effective sports betting law was only passed in 2021.

\subsection{Regulation in Germany}

Like other European countries, Germany traditionally banned the private provision of sports betting.\footnote{Some exceptions existed for horse racing, where a selected group of private horse racing clubs and registered bookmakers were allowed to organize events and provide parimutuel betting \citep{englisch2013taxation}.} However, the state monopoly on sports betting and gambling had already been challenged in the course of German reunification by a debate on how to deal with four private sports betting licenses granted under GDR law \citep{rebeggiani2017neue}. 
The new Gaming Treaty in 2021 was the final step of unsuccessful attempts to adjust the German gambling regulatory framework to European law and the digital era.

In 2006, the German constitutional court ruled that the state monopoly in its traditional form was at variance with German constitutional law \citep{bverfg2006}, which triggered the first Glücksspielstaatsvertrag (interstate Gaming Treaty) in 2008 that imposed an even stricter state monopoly \citep{rebeggiani2017neue}. This initial gaming treaty was under immediate pressure from European authorities, legal scholars and the private sector. In September 2010, the European Court of Justice ruled that national monopolies are incompatible with European law if they do not serve the pursued public interest objective, requiring a more coherent gambling regulation in Germany. This ruling led to an intense debate among German states and to the first amendment of the German interstate Gaming Treaty ("Gl\"{u}cksspiel\"{a}nderungsstaatsvertrag") in July 2012. The first amendment of the interstate Gaming Treaty was accompanied by the introduction of the new "Federal Sports Bets Tax" (see section \ref{sec:sports_tax}).

The amendment stipulated that online sports betting remains prohibited but defined a 7-year "experimental phase", in which private operators could apply for a maximum of 20 state-issued concessions.\citep{englisch2013taxation}.\footnote{The northern state of Schleswig-Holstein (SH) initially opted out of the interstate Gaming Treaty and applied a more liberal gambling regulation, particularly concerning online poker and casinos. By the end of December, SH had granted more than 20 betting licenses. These liberalization efforts were only short-lasted, as in January 2013, the new government of SH reversed this decision and rejoined the interstate Gaming Treaty \citep{rebeggiani2017neue}.} However, several court rulings halted the process of granting betting concessions because the process was not in compliance with the law, erroneous, and non-transparent. \citep{rebeggiani2017neue}. 
The ratification of an additional amendment of the Treaty failed due to resistance from Schleswig-Holstein and North Rhine-Westphalia in March 2017. By the end of 2020, there was still no concession granted for legal sports betting in Germany \citep[for more details see:][]{hessen2021eval}. The deadlock in German online gaming regulation between 2006 and 2021 was further fueled by betting agencies that actively challenged the Gaming treaties in court, leading to rulings that eventually delayed an effective gambling regulation in Germany until 2021. 

In sum, gambling regulation in Germany was in abeyance until the new "Glücksspielstaatsvertrag" (interstate Gaming Treaty) was passed in 2021 \citep{haucap2021glucksspielstaatsvertrag}. The German regulatory framework for online gambling and betting was characterized by a division of tax and regulatory competencies between the federal and the state level. This division of competencies led to opaque and ineffective legislation, which was at variance with European law \citep{eugh2010} and even German constitutional law. In this regulatory environment, betting agencies, domiciled in liberal (primarily European) jurisdictions outside Germany where their services were admissible, could serve the German demand for sports betting in a semi-legal and more or less unrestricted grey market tolerated by German authorities\citep{rebeggiani2017neue}.\footnote{ A statement in the 2009 annual report of the betting provider \textit{sporting bet} summarizes the industry view (and strategy) on the German gambling regulation quite fittingly: "Furthermore, enforcement action against operators where they actively target German residents (including through local marketing) has been curbed due to the lack of clarity in the legal position. In our view, therefore, legislation that was intended to almost comprehensively block online gambling has had only limited effect and the general inability of the German government to block online gambling websites, coupled with the questionable legality of the legislation, has led to a continued supply of online gambling services, and an absence of extra-territorial enforcement against the activity."}

%% file: Chapters/AppendixA.tex
\section{The difference between posted and effective betting prices}\label{sec:sport_app_tau}

We defined $\tau$ as the difference between the tax-inclusive and posted betting prices (i.e., the shrouded tax surcharge):
\begin{equation}
\tau=p_{ib}-\Tilde{p_{ib}}=\dfrac{1}{\Tilde{\theta}_{ib}}-\dfrac{1}{\theta_{ib}}
\end{equation} 

For the three different policies that betting agencies set, $\tau$ can be derived as follows:
\begin{itemize}
    \item[i)] \textbf{No further deductions}: $\tau=0$. The derivation of $\tau$ is trivial as: $p_{ib}=\Tilde{p_{ib}}$
    \item[ii)] \textbf{Deducting tax surcharge from winnings}:

    Using that $\theta_{ib} = \sum_{s=1}^n \dfrac{1}{r_{i,s,b}}$ and $r_{i,s,b}=\Tilde{r}_{isb}(1-t)$ we can rewrite equation (4) to:
    \begin{equation*}
        \tau=\dfrac{1}{\Tilde{\theta}_{ib}}-\dfrac{1}{\sum_{s=1}^n\dfrac{1}{(1-t)\Tilde{r}_{isb}}}=\dfrac{1}{\Tilde{\theta}_{ib}}-\dfrac{1-t}{\Tilde{\theta}_{ib}}=\dfrac{t}{\Tilde{\theta}_{ib}}
    \end{equation*}
    \item[iii)] \textbf{Deducting tax surcharge from wager}:  Noting that $r_{i,s,b}=\dfrac{\Tilde{r}_{isb}}{(1+t)}$ and $\theta_{ib} = \sum_{s=1}^n \dfrac{1}{r_{i,s,b}}$ we can rewrite equation (4) to:
    \begin{equation*}
        \tau=\dfrac{1}{\Tilde{\theta}_{ib}}-\dfrac{1}{\sum_{s=1}^n\dfrac{(1+t)}{\Tilde{r}_{isb}}}=\left(1-\dfrac{1}{(1+t)}\right)\dfrac{1}{\Tilde{\theta}_{ib}}=\dfrac{t}{(1+t)\Tilde{\theta}_{ib}}
    \end{equation*}
\end{itemize}

%% file: Figures_tex/Appendix_figures.tex
\input{Figures_tex/hist_prices}

\input{Figures_tex/Avg_prices}

\input{Figures_tex/avg_prices_robust}

\input{Figures_tex/TWFE_robust}

\input{Figures_tex/TWFE_all_tr}

\input{Figures_tex/TWFE_het_rob}

%% file: Figures_tex/hist_prices.tex
\begin{figure}[!ht]

    \centering
      \caption{Distribution of consumer betting prices}
    
    \begin{subfigure}[t]{0.49\textwidth}
        \centering
        \caption{ All sports }
        \includegraphics[width=\linewidth]{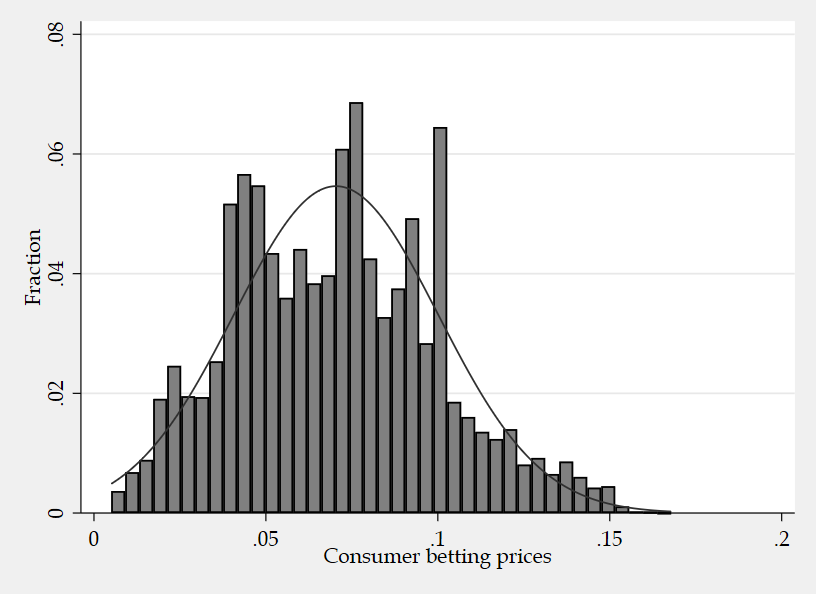} 
         \label{fig:timing1}
    \end{subfigure}
    \hfill
        \begin{subfigure}[t]{0.49\textwidth}
        \centering
        \caption{Soccer games}
        \includegraphics[width=\linewidth]{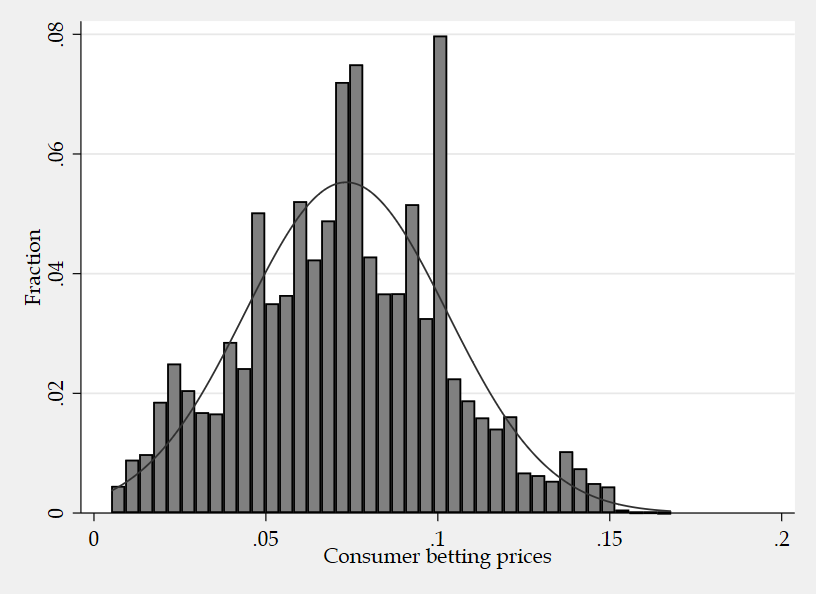} 
        \label{fig:timing1}
    \end{subfigure}

     		\begin{minipage}{13.8cm}  \footnotesize  \emph{Notes:} 
This figures shows the distribution of betting prices and the corresponding normal distribution with the mean and standard deviation of the sample. Panel a) considers all sports, Panel (b) only soccer games. The prices considered for the graphs are trimmed by quarter at the quarterly 1- and 99-percentiles, considering all events and all agencies. The number of bins is fixed to 40 and the first bin starts at 0.005, which includes also the smallest price. Width per bin is equal to 0.0040721.
		\end{minipage} 
		\label{fig:hist_prices}
  
\end{figure} 

%% file: Figures_tex/Avg_prices.tex
\begin{figure}[!ht]

    \centering
      \caption{Avg. tax-inclusive consumer betting prices of German and Non-German betting agencies over time with data for all years}
    
    \begin{subfigure}[t]{0.48\textwidth}
        \centering
        \caption{ All sports }
        \includegraphics[width=\linewidth]{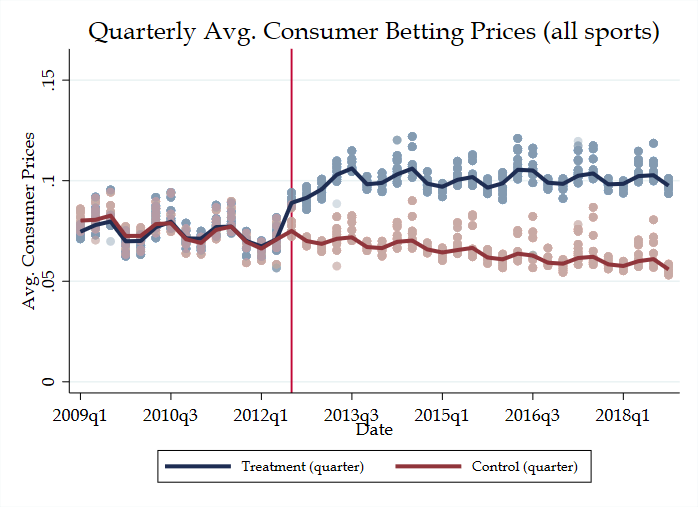} 
         \label{fig:timing1}
    \end{subfigure}
    \hfill
        \begin{subfigure}[t]{0.48\textwidth}
        \centering
        \caption{ Only soccer games}
        \includegraphics[width=\linewidth]{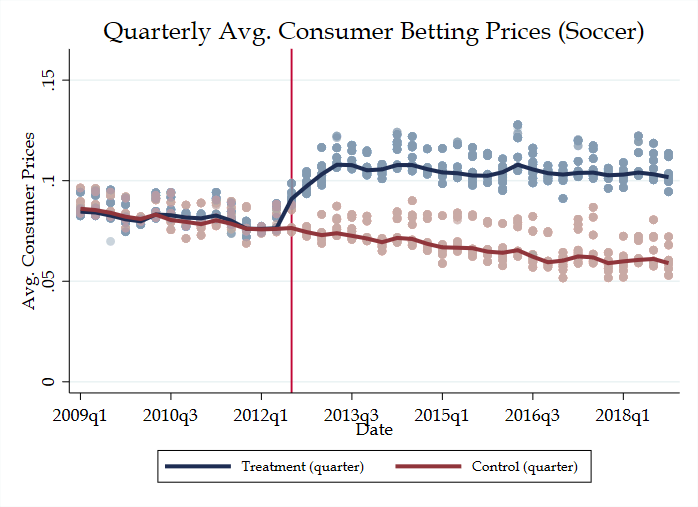} 
        \label{fig:timing1}
    \end{subfigure}

    \vspace{0.5cm}
    \begin{subfigure}[t]{0.48\textwidth}
        \centering
         \caption{ All sports - excluding cross leagues }
        \includegraphics[width=\linewidth]{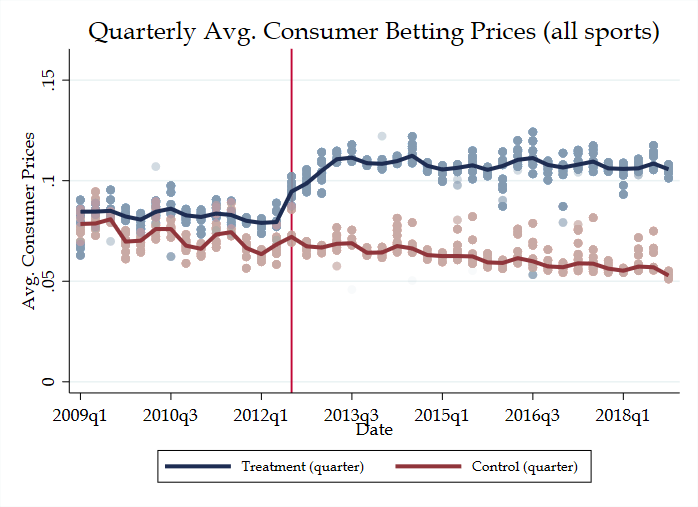} 
        \label{fig:timing2}
    \end{subfigure}
    \hfill
    \begin{subfigure}[t]{0.48\textwidth}
        \centering
        \caption{ Only soccer games - excluding cross leagues}
        \includegraphics[width=\linewidth]{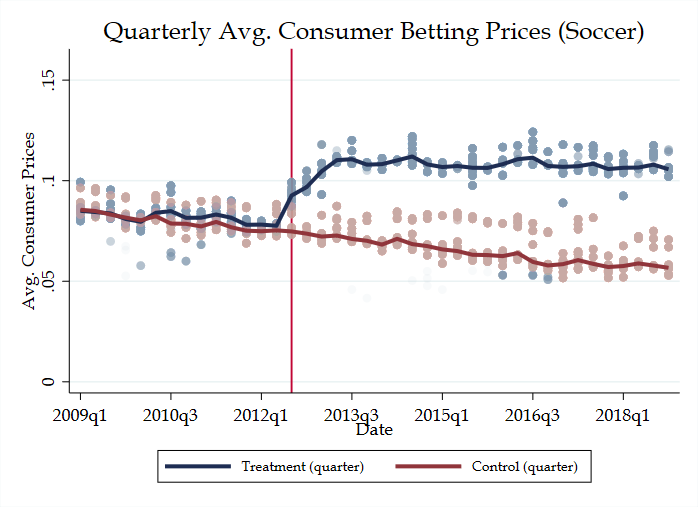} 
         \label{fig:timing2}
    \end{subfigure}
     		\begin{minipage}{15.8cm}  \footnotesize  \emph{Notes:} 
This figures shows the average quarterly and weekly betting prices over time for agencies that target the German market (Treatment) and for agencies that do not target the German market (Control). Only agencies for which odds data in all years is available are considered, N=30 (Treatment=9/Control=21). The circles signify average weekly betting prices in each group pooled and aligned by the respective quarter. The solid line illustrates the average quarterly betting prices. The odds considered for the graphs are trimmed by quarter at the quarterly 1- and 99-percentiles, considering all events and all agencies. Panel (a) and (c) consider all sports included in the data set, while Panel (b) and (d) only consider soccer events. Panel (c) and (d) excludes German leagues for the control group and Non-German leagues for the treatment group.
		\end{minipage} 
		\label{fig:avgprices_compl}
  
\end{figure} 

%% file: Figures_tex/avg_prices_robust.tex
\begin{figure}

    \centering
      \caption{Avg. tax-inclusive consumer betting prices and estimated tax effects over time - treated agencies vs. agencies with a foreign domain}

        \begin{subfigure}[t]{0.48\textwidth}
        \centering
        \caption{Soccer games}
        \includegraphics[width=\linewidth]{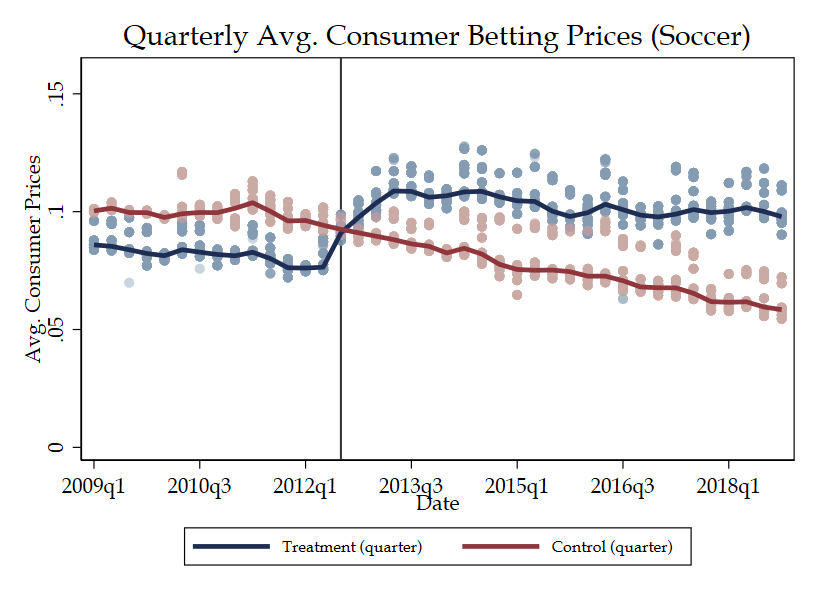} 
        \label{fig:timing1}
    \end{subfigure}
    \hfill
    \begin{subfigure}[t]{0.48\textwidth}
        \centering
        \caption{Soccer games - excluding cross leagues}
        \includegraphics[width=\linewidth]{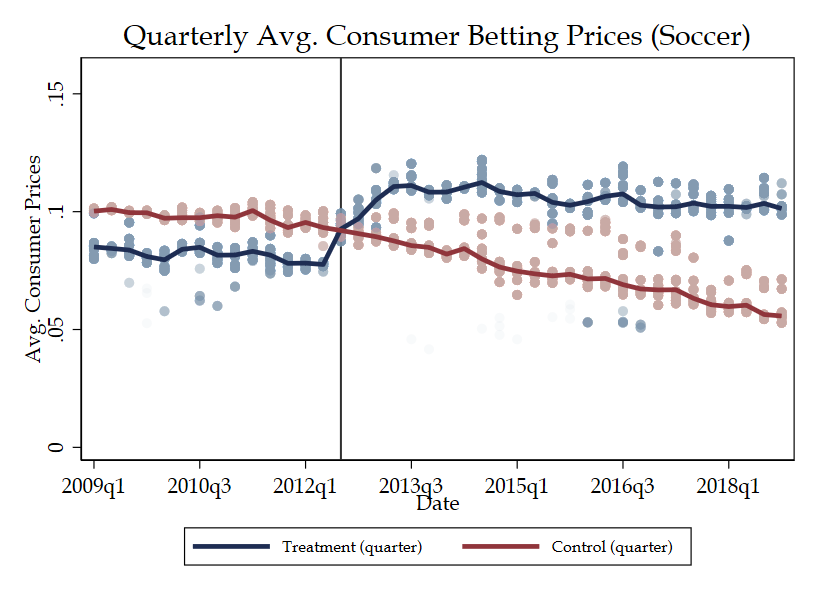} 
         \label{fig:timing2}
    \end{subfigure}
    
        \vspace{0.5cm}
    \begin{subfigure}[t]{0.48\textwidth}
        \centering
         \caption{ Soccer games }
        \includegraphics[width=\linewidth]{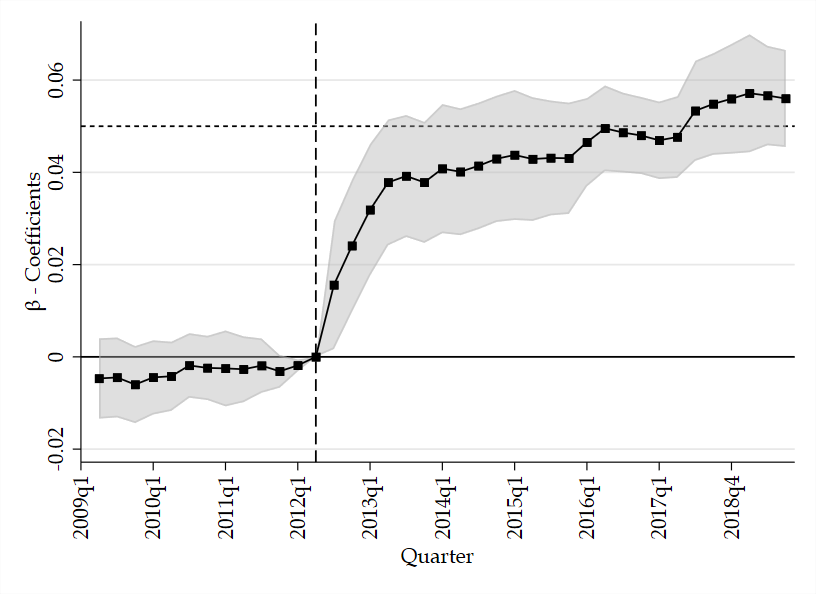} 
        \label{fig:timing2}
    \end{subfigure}
    \hfill
    \begin{subfigure}[t]{0.48\textwidth}
        \centering
        \caption{ Soccer games - excluding cross leagues}
        \includegraphics[width=\linewidth]{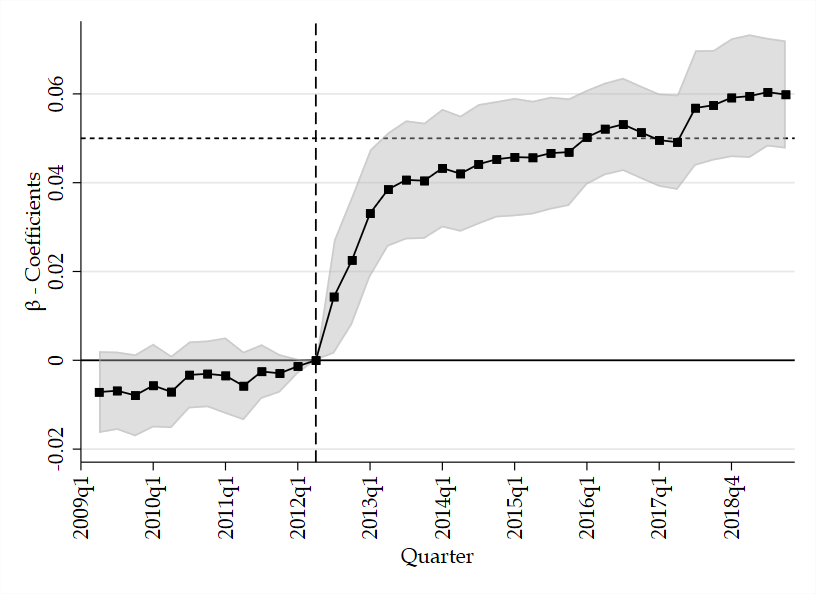} 
         \label{fig:timing2}
    \end{subfigure}
     		\begin{minipage}{15.8cm}  \footnotesize  \emph{Notes:} 
This figures shows the average quarterly and weekly betting prices for soccer events (Panel a and b) and corresponding estimated treatment effects (Panel c and d) of the German sports-betting tax on consumers’ betting prices over time. Only agencies that target the German market (Treatment) and agencies that have a foreign domain (Control) are considered. Unique agencies: N=32 (Treatment=10/Control=22). The vertical lines illustrate the introduction of the tax reform. For Panel a and b, the circles signify average weekly betting prices in each group, pooled and aligned by the respective quarter. The solid line illustrates the average quarterly betting prices. The odds considered for the graphs in Panel a) and b) are trimmed by quarter at the quarterly 1- and 99-percentiles, considering all events and all agencies. For Panel c) and d), the shaded area shows the 95\% confidence intervals of the estimated effects and the horizontal dotted line signifies a full-pass through of the tax on consumers.
		\end{minipage} 
		\label{fig:avgprices_all_robust}
  
\end{figure}

%% file: Figures_tex/TWFE_robust.tex
\begin{figure}[!ht]

    \centering
      \caption{Effect of tax on consumer betting prices over time - for all sports/all agencies and for agencies where data is available for all year}
          \begin{subfigure}[t]{0.48\textwidth}
        \centering
        \caption{ All sports }
        \includegraphics[width=\linewidth]{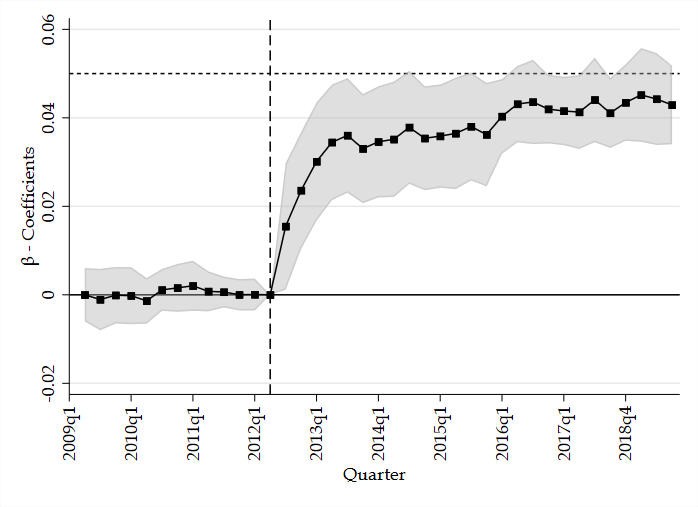} 
         \label{fig:timing1}
    \end{subfigure}
    \hfill
    \begin{subfigure}[t]{0.48\textwidth}
        \centering
         \caption{ All sports - excluding cross leagues }
        \includegraphics[width=\linewidth]{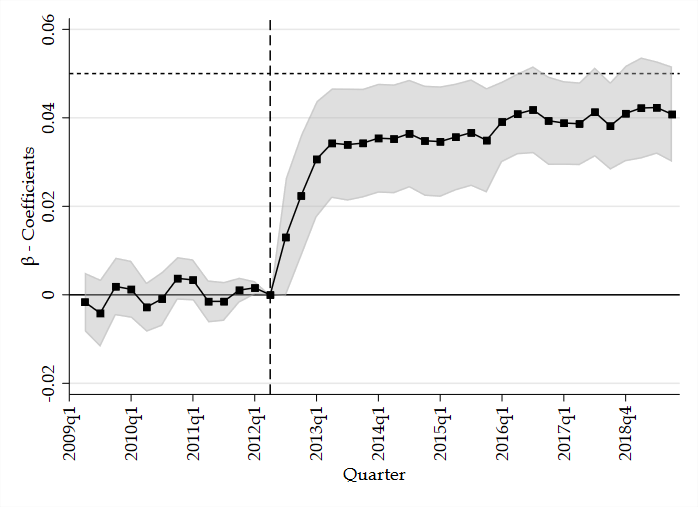} 
        \label{fig:timing2}
    \end{subfigure}
          \vspace{0.5cm}

     \begin{subfigure}[t]{0.48\textwidth}
        \centering
        \caption{Soccer games}
        \includegraphics[width=\linewidth]{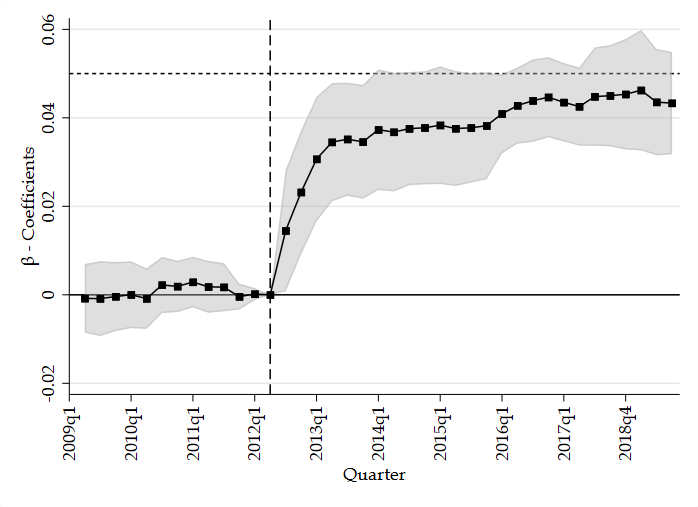} 
        \label{fig:timing1}
    \end{subfigure}
    \hfill
    \begin{subfigure}[t]{0.48\textwidth}
        \centering
        \caption{Soccer games - excluding cross leagues}
        \includegraphics[width=\linewidth]{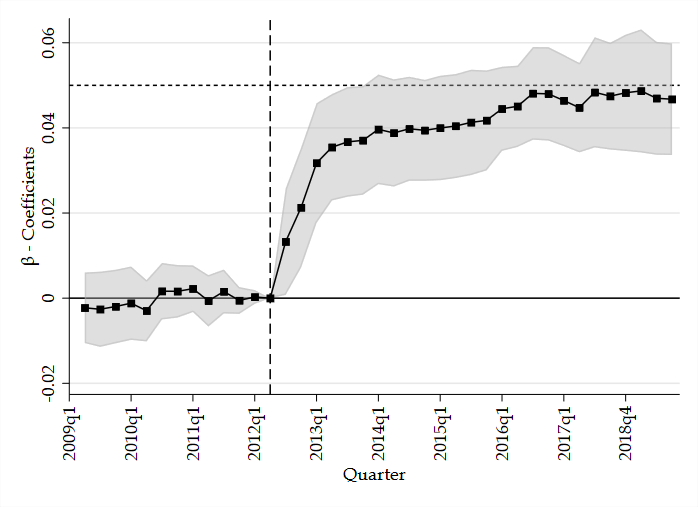} 
         \label{fig:timing2}
    \end{subfigure}
      \vspace{0.5cm}
  
    \begin{subfigure}[t]{0.48\textwidth}
        \centering
        \caption{ All sports }
        \includegraphics[width=\linewidth]{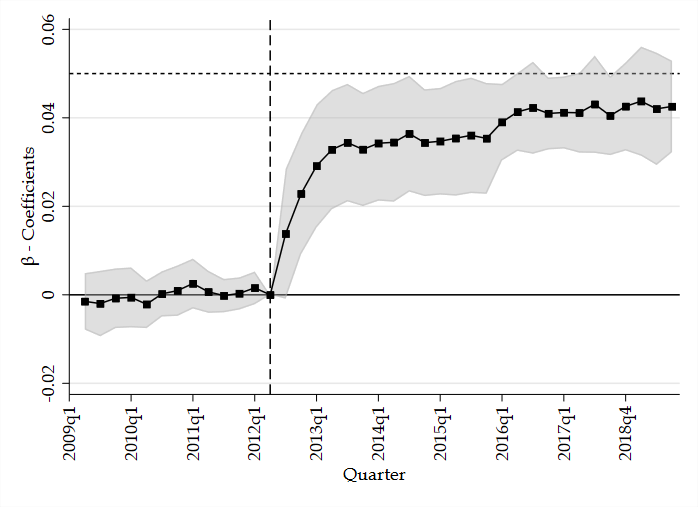} 
         \label{fig:timing1}
    \end{subfigure}
    \hfill
    \begin{subfigure}[t]{0.48\textwidth}
        \centering
         \caption{ All sports - excluding cross leagues }
        \includegraphics[width=\linewidth]{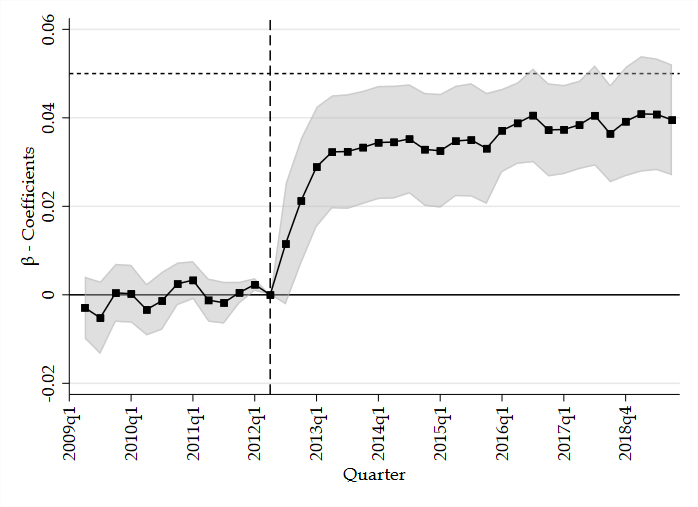} 
        \label{fig:timing2}
    \end{subfigure}

     		\begin{minipage}{15.8cm}  \footnotesize  \emph{Notes:} 
This figures shows the estimated effects and respective 95-\% confidence intervals of the German sports-betting tax on consumers' betting prices over time. Only agencies for which odds data in all years is available are considered, N=30 (Treatment=9/Control=21). Panel (a) and (b) consider all sports included in the data set, while Panel (c) and (d) only consider soccer events. Panel (b) and (d) excludes German leagues for the control group and Non-German leagues for the treatment group. The vertical dashed line illustrates the introduction of the tax reform. The horizontal dotted line signifies a full-pass through of the tax on consumers.
		\end{minipage} 
		\label{fig:coef robust}
  
\end{figure} 

%% file: Figures_tex/TWFE_all_tr.tex
\begin{figure}[!ht]

    \centering
      \caption{Pre-trends and tax effects on consumer betting prices over time - all agencies - trimmed sample}
        \begin{subfigure}[t]{0.48\textwidth}
        \centering
        \caption{ Soccer games}
        \includegraphics[width=\linewidth]{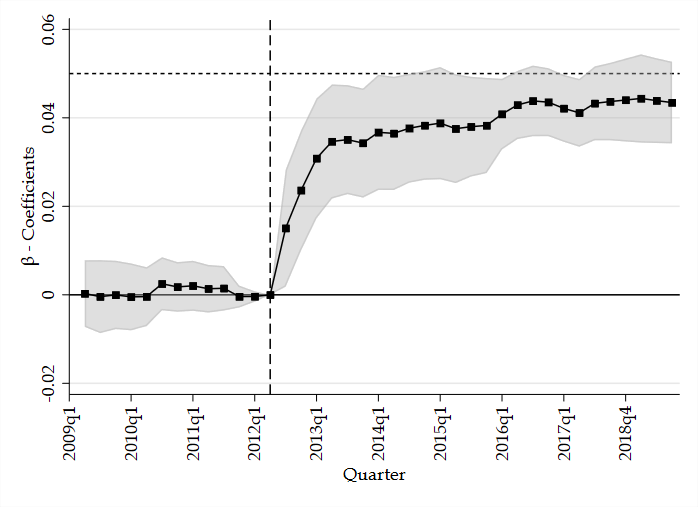} 
        \label{fig:timing1}
    \end{subfigure}
    \hfill
    \begin{subfigure}[t]{0.48\textwidth}
        \centering
        \caption{ Soccer games - excluding cross leagues}
        \includegraphics[width=\linewidth]{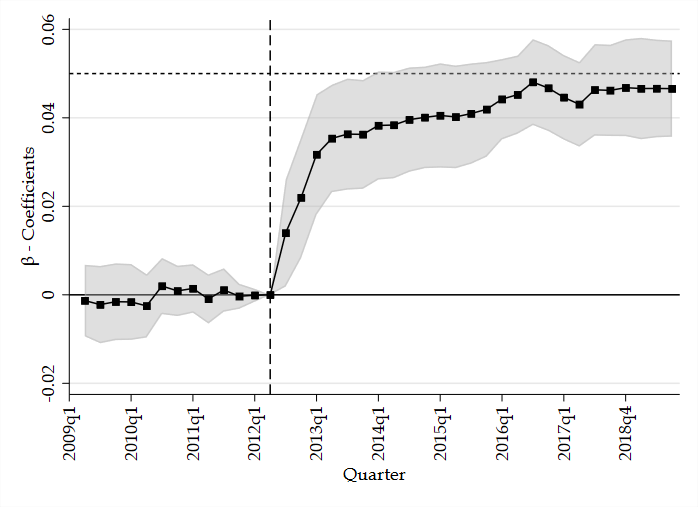} 
         \label{fig:timing2}
    \end{subfigure}

    \vspace{0.5cm}
\begin{subfigure}[t]{0.48\textwidth}
        \centering
        \caption{ All sports }
        \includegraphics[width=\linewidth]{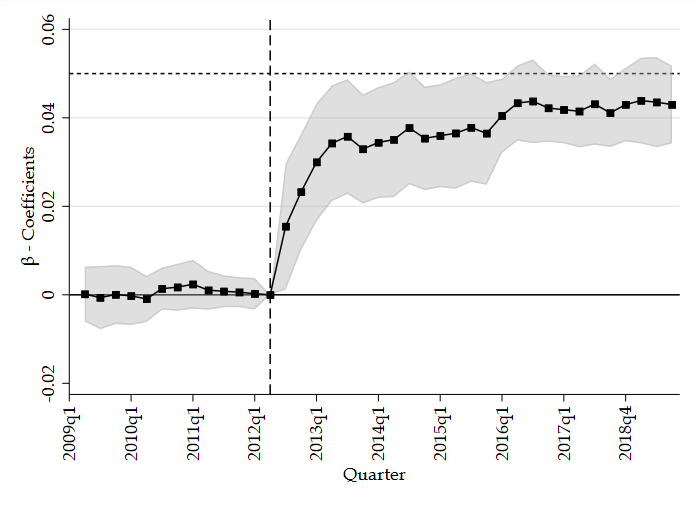} 
         \label{fig:timing1}
    \end{subfigure}
    \hfill
    \begin{subfigure}[t]{0.48\textwidth}
        \centering
         \caption{ All sports - excluding cross leagues }
        \includegraphics[width=\linewidth]{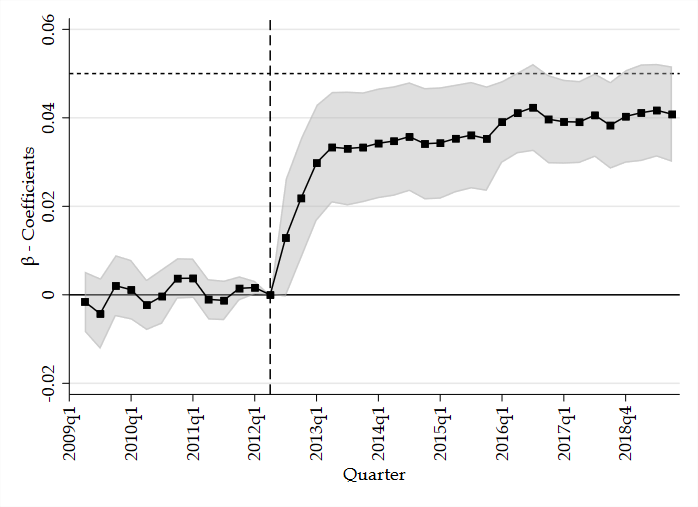} 
        \label{fig:timing2}
    \end{subfigure}
     		\begin{minipage}{13.8cm}  \footnotesize  \emph{Notes:} 
This figures shows the estimated effects and respective 95-\% confidence intervals of the German sports-betting tax on consumers' betting prices over time. All agencies are considered, N=65 (Treatment=10/Control=55). The odds considered for the graphs are trimmed by quarter at the quarterly 1- and 99-percentiles, considering all events and all agencies. Panel (a) and (b) consider soccer matches, while Panel (c) and (d) consider all sports events. Panel (b) and (d) excludes German leagues for the control group and Non-German leagues for the treatment group. The vertical dashed line illustrates the introduction of the tax reform. The horizontal dotted line signifies a full-pass through of the tax on consumers.
		\end{minipage} 
		\label{fig:coef_tr}
  
\end{figure} 

%% file: Figures_tex/TWFE_het_rob.tex
\begin{figure}[!ht]

    \centering
      \caption{Pre-trends and tax effects on consumer betting prices over time - agencies where data is available for all year}
          \begin{subfigure}[t]{0.48\textwidth}
        \centering
        \caption{ All sports }
        \includegraphics[width=\linewidth]{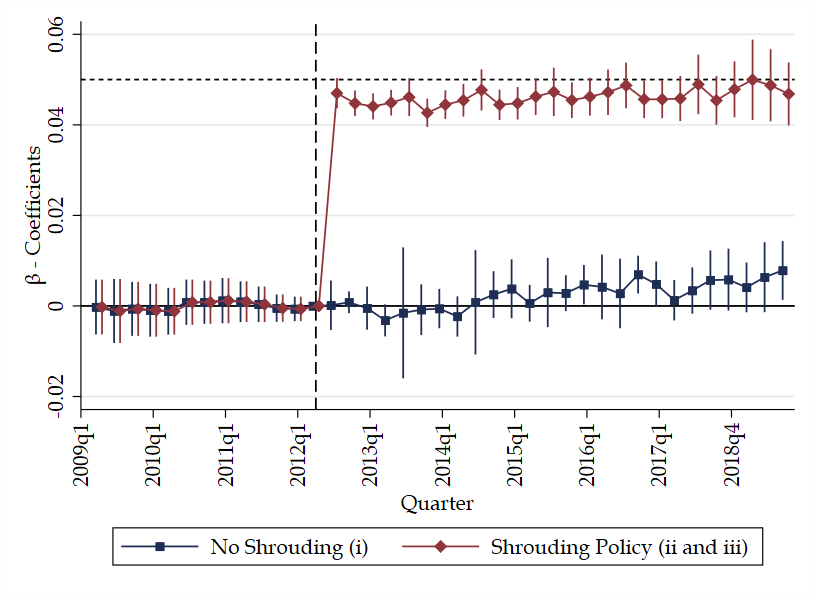} 
         \label{fig:timing1}
    \end{subfigure}
    \hfill
    \begin{subfigure}[t]{0.48\textwidth}
        \centering
         \caption{ All sports - excluding cross leagues }
        \includegraphics[width=\linewidth]{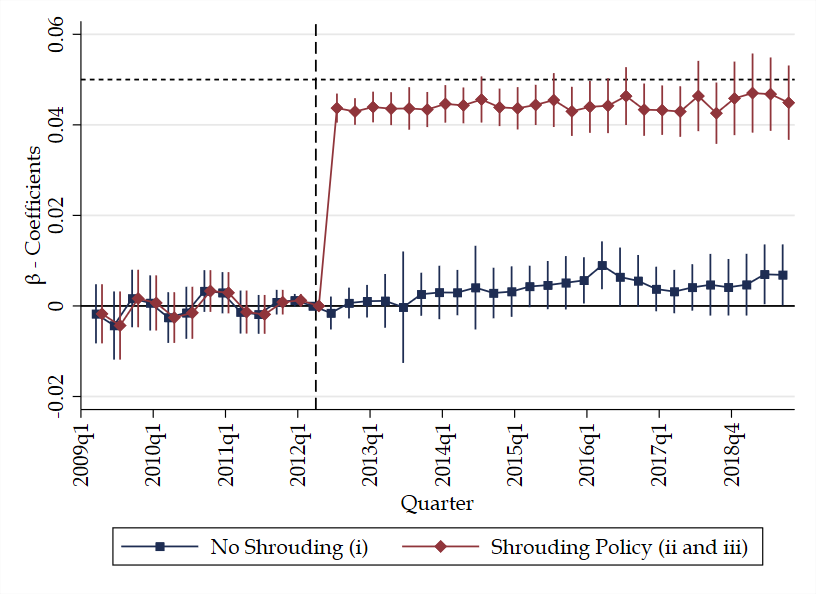} 
        \label{fig:timing2}
    \end{subfigure}
          \vspace{0.5cm}

     \begin{subfigure}[t]{0.48\textwidth}
        \centering
        \caption{ Soccer games}
        \includegraphics[width=\linewidth]{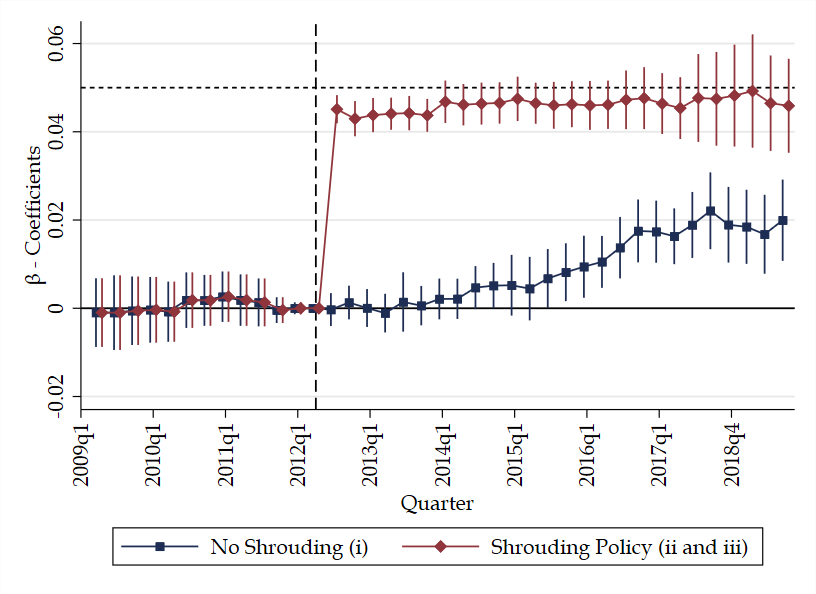} 
        \label{fig:timing1}
    \end{subfigure}
    \hfill
    \begin{subfigure}[t]{0.48\textwidth}
        \centering
        \caption{ Soccer games - excluding cross leagues}
        \includegraphics[width=\linewidth]{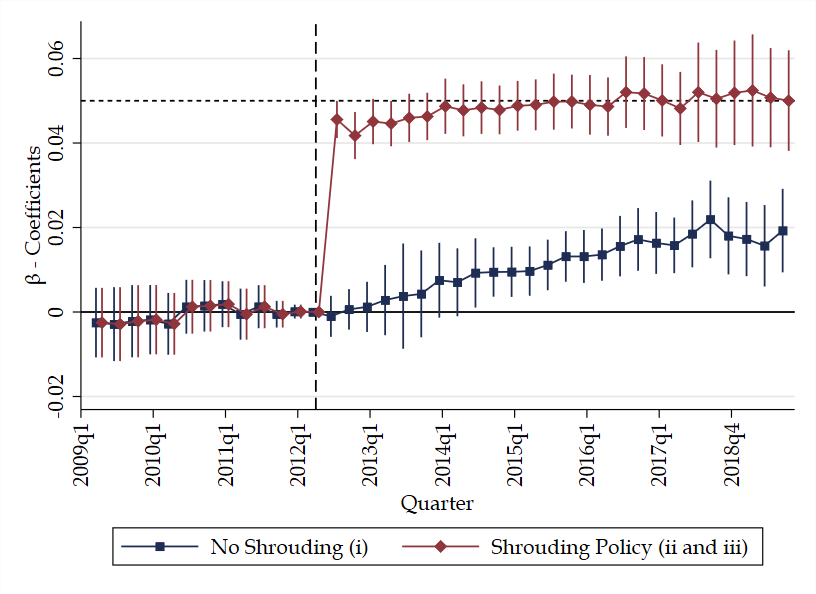} 
         \label{fig:timing2}
    \end{subfigure}
      \vspace{0.5cm}
  
    \begin{subfigure}[t]{0.48\textwidth}
        \centering
        \caption{ All sports }
        \includegraphics[width=\linewidth]{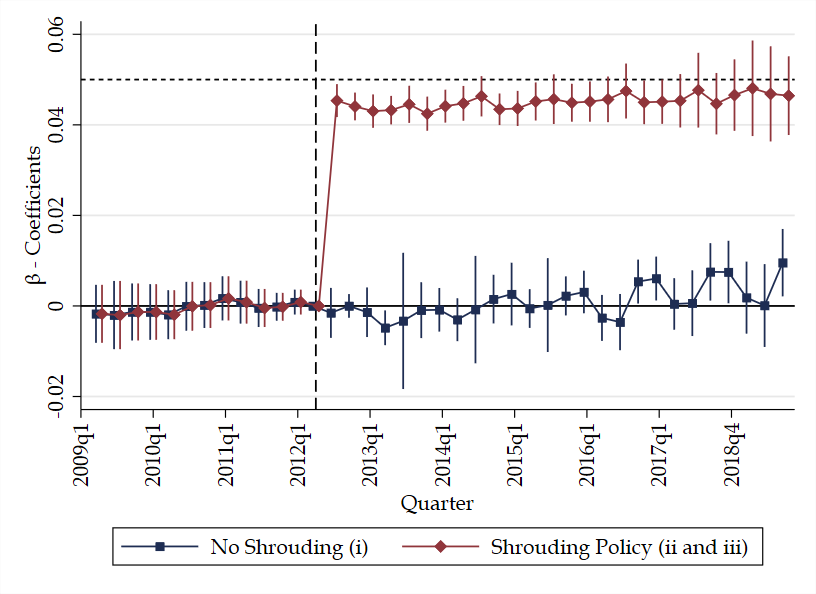} 
         \label{fig:timing1}
    \end{subfigure}
    \hfill
    \begin{subfigure}[t]{0.48\textwidth}
        \centering
         \caption{ All sports - excluding cross leagues }
        \includegraphics[width=\linewidth]{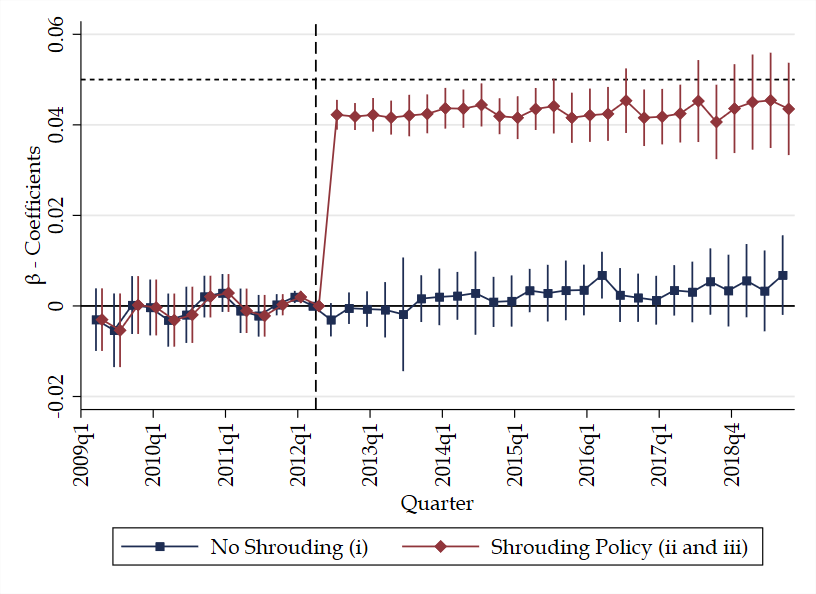} 
        \label{fig:timing2}
    \end{subfigure}

     		\begin{minipage}{15.8cm}  \footnotesize  \emph{Notes:} 
This figures shows the estimated effects and respective 95-\% confidence intervals of the German sports-betting tax on consumers' betting prices over time. Only agencies for which odds data in all years is available are considered, N=30 (Treatment=9/Control=21). Panel (a) and (b) consider all sports included in the data set, while Panel (c) and (d) only consider soccer events. Panel (b) and (d) excludes German leagues for the control group and Non-German leagues for the treatment group. The vertical dashed line illustrates the introduction of the tax reform. The horizontal dotted line signifies a full-pass through of the tax on consumers.
		\end{minipage} 
		\label{fig:coef_het_robust}
  
\end{figure}

%% file: Tables/AppendixTables.tex
\input{Tables/treated_policies_overview}

\input{Tables/DiffinDiff_allsports}

\input{Tables/DiffinDiff_robust}

\input{Tables/DiffinDiff_het_compl}

\input{Tables/DiffinDiff_het_all}

%% file: Tables/treated_policies_overview.tex
\begin{table}[!htbp]\centering
\caption{List and of included competitions and number of unique matches}
\label{sum_policy}
\begin{threeparttable}

\begin{tabular}{ l|c|c } 

 Betting agency & Type of "shrouding" policy & Implementation month  \\  \hline
 Bet365 & Deduction of winnings & November 2012  \\ 
 Bet-at-home & Deduction of wager & August 2012  \\
 Betfair Sportsbook & No shrouding & Not applicable \\ 
 Betsson & Deduction of wager & February 2016 \\ 
 Betway & Deduction of winnings & March 2013 \\ 
  Bwin & Deduction of winnings & August 2012  \\ 
 Interwetten & Deduction of winnings & July 2012  \\ 
 Sportingbet & Deduction of winnings & October 2012   \\ 
 Tipico & No shrouding & Not applicable \\ 
 Unibet & Deduction of winnings & December 2012 \\ 

   \hline

\end{tabular}
\begin{tablenotes}[flushleft]
\item \emph{Notes:} The table lists all treated agencies and shows the respective shrouding policy and its implementation date. 
\end{tablenotes}
\end{threeparttable}
\end{table}

%% file: Tables/DiffinDiff_allsports.tex
\begin{table}[ !ht]\centering
\def\sym#1{\ifmmode^{#1}\else\(^{#1}\)\fi}
\caption{Avg. effect of tax on consumer betting prices - all sports}
\label{diffindiff_allsports}
\scalebox{0.7}{
\begin{threeparttable}
\begin{tabular}{l*{7}{c}}
\hline\hline

                 &\multicolumn{3}{c}{All Leagues}     & &\multicolumn{3}{c}{Excl. "cross" leagues}     \\    \cmidrule{2-4} \cmidrule{6-8}
                    &\multicolumn{1}{c}{(1)}      &\multicolumn{1}{c}{(2)}      &\multicolumn{1}{c}{(3)}   &   &\multicolumn{1}{c}{(4)}      &\multicolumn{1}{c}{(5)}      &\multicolumn{1}{c}{(6)}      \\
\midrule
\textbf{Panel A: All agencies} &&&&&&& \\
[1em]
Tax effect on prices             &               0.037\sym{***}&               0.037\sym{***}&               0.038\sym{***}&  &             0.036\sym{***}&               0.035\sym{***}&               0.037\sym{***}\\
                    &             (0.005)         &             (0.005)         &             (0.005)         &   &          (0.005)         &             (0.005)         &             (0.005)         \\

\hline
Observations        &             3,057,547         &             3,057,547         &             3,057,547         &  &           2,027,511         &             2,027,511         &             2,027,511         \\
\(R^{2}\)           &               0.617         &               0.730         &               0.818         &    &           0.649         &               0.735         &               0.831         \\
\midrule

\textbf{Panel B: Compl. agencies} &&&&&&&\\
[1em]
Tax effect on prices            &               0.037\sym{***}&               0.036\sym{***}&               0.038\sym{***}&   &            0.036\sym{***}&               0.035\sym{***}&               0.036\sym{***}\\
                    &             (0.005)         &             (0.005)         &             (0.005)      &   &             (0.005)         &             (0.005)         &             (0.005)         \\

\hline
Observations        &             2,027,861         &             2,027,861         &             2,027,861         &  &           1,223,098         &             1,223,098         &             1,223,098         \\
\(R^{2}\)           &               0.468         &               0.656         &               0.760         &   &            0.445         &               0.614         &               0.742         \\

\midrule

Constant            &                 Yes         &                 Yes         &                 Yes       &  &                 Yes         &                 Yes         &                 Yes         \\

Time FE             &                 Yes         &                 Yes         &                 Yes    &     &                 Yes         &                 Yes         &                 Yes         \\

Agency FE           &                 Yes         &                 Yes         &                  No       &  &                 Yes         &                 Yes         &                  No         \\

League FE           &                  No         &                 Yes         &                  No     &    &                  No         &                 Yes         &                  No         \\

League-agency FE    &                  No         &                  No         &                 Yes      &   &                  No         &                  No         &                 Yes         \\
\hline \hline
\end{tabular}
\begin{tablenotes}[flushleft]
\item \emph{Notes:} This table reports the estimated average sports betting tax effects on consumer betting prices of all sports events in the sample according to Eq. \ref{eq:tax_avg_effect}, comparing changes in prices before and after the tax reform between German (treatment group) and Non-German agencies (control group). Each agency-event combination (unique observation) is equally weighted. Columns 1-3 consider all leagues, while Columns 4-6 exclude events in German leagues for the control group and Non-German leagues for the treatment group. All estimations include time-fixed effects, and the columns differ in the included agency and league fixed effects. Estimations in Panel A consider observations from all agencies, and Panel B considers agencies with observations from all years only. Robust standard errors clustered at the agency level are reported in the brackets. * denotes significance at the 10-\%, ** at the 5-\% and *** at the 1-\% level.
\end{tablenotes}
\end{threeparttable}
}
\end{table}

%% file: Tables/DiffinDiff_robust.tex
\begin{table}[ !ht]\centering
\def\sym#1{\ifmmode^{#1}\else\(^{#1}\)\fi}
\caption{Avg. effect of tax on consumer betting prices - control group with foreign country-specific domain}
\label{diffindiff_robust}
\scalebox{0.75}{
\begin{threeparttable}
\begin{tabular}{l*{7}{c}}
\hline\hline

                 &\multicolumn{3}{c}{All Leagues}     & &\multicolumn{3}{c}{Excl. "cross" leagues}     \\    \cmidrule{2-4} \cmidrule{6-8}
                    &\multicolumn{1}{c}{(1)}      &\multicolumn{1}{c}{(2)}      &\multicolumn{1}{c}{(3)}   &   &\multicolumn{1}{c}{(4)}      &\multicolumn{1}{c}{(5)}      &\multicolumn{1}{c}{(6)}      \\

\midrule

\textbf{Panel A: All agencies} &&&&&&& \\
[1em]
Tax effect on prices            &               0.047\sym{***}&               0.047\sym{***}&               0.047\sym{***}& &              0.051\sym{***}&               0.051\sym{***}&               0.051\sym{***}\\
                    &             (0.005)         &             (0.005)         &             (0.005)      &   &             (0.005)         &             (0.005)         &             (0.005)         \\

\hline
Observations        &              983,713         &              983,713         &              983,713     &    &              550,380         &              550380         &              550380         \\
\(R^{2}\)           &               0.579         &               0.728         &               0.796      &   &               0.589         &               0.741         &               0.810         \\
\midrule
\textbf{Panel B: Compl. agencies} &&&&&&& \\
[1em]
Tax effect on prices            &               0.049\sym{***}&               0.050\sym{***}&               0.050\sym{***}&  &             0.054\sym{***}&               0.054\sym{***}&               0.054\sym{***}\\
                    &             (0.004)         &             (0.004)         &             (0.004)       &  &             (0.005)         &             (0.005)         &             (0.005)         \\

\hline
Observations        &              644,326         &              644,326         &              644,326       &  &              288,854         &              288,854         &              288,854         \\
\(R^{2}\)           &               0.554         &               0.711         &               0.746         &   &            0.582         &               0.743         &               0.761         \\

\midrule

Constant            &                 Yes         &                 Yes         &                 Yes       &  &                 Yes         &                 Yes         &                 Yes         \\

Time FE             &                 Yes         &                 Yes         &                 Yes    &     &                 Yes         &                 Yes         &                 Yes         \\

Agency FE           &                 Yes         &                 Yes         &                  No       &  &                 Yes         &                 Yes         &                  No         \\

League FE           &                  No         &                 Yes         &                  No     &    &                  No         &                 Yes         &                  No         \\

League-agency FE    &                  No         &                  No         &                 Yes      &   &                  No         &                  No         &                 Yes         \\
\hline \hline

\end{tabular}
\begin{tablenotes}[flushleft]
\item \emph{Notes:} This table reports the estimated average sports betting tax effects on consumer betting prices of soccer events according to Eq. \ref{eq:tax_avg_effect}, comparing changes in prices before and after the tax reform between German agencies (treatment group) and agencies with a foreign country-specific domain (control group). Each agency-event combination (unique observation) is equally weighted. Columns 1-3 consider all leagues, while Columns 4-6 exclude events in German leagues for the control group and Non-German leagues for the treatment group. All estimations include time-fixed effects, and the columns differ in the included agency and league fixed effects. Estimations in Panel A consider observations from all agencies, and Panel B considers agencies with observations from all years only. Robust standard errors clustered at the agency level are reported in the brackets. * denotes significance at the 10-\%, ** at the 5-\% and *** at the 1-\% level.
\end{tablenotes}
\end{threeparttable}
}
\end{table}

%% file: Tables/DiffinDiff_het_compl.tex
\begin{table}[ !ht]\centering
\def\sym#1{\ifmmode^{#1}\else\(^{#1}\)\fi}
\caption{Avg. effect of tax on consumer betting prices - soccer - compl. agencies}
    \label{diffindiff_het_compl}
    \scalebox{0.75}{
\begin{threeparttable}
\begin{tabular}{l*{7}{c}}
\hline\hline

                 &\multicolumn{3}{c}{All Leagues}     & &\multicolumn{3}{c}{Excl. "cross" leagues}     \\    \cmidrule{2-4} \cmidrule{6-8}
                                  &\multicolumn{2}{c}{Subsamples}     & Interact. & &\multicolumn{2}{c}{Subsamples} &  Interact.  \\    \cmidrule{2-3} \cmidrule{6-7}
                    &\multicolumn{1}{c}{No shroud}      &\multicolumn{1}{c}{Shrouding}      &   &   &\multicolumn{1}{c}{No shroud}      &\multicolumn{1}{c}{Shrouding}      &   \\
                    &\multicolumn{1}{c}{(1)}      &\multicolumn{1}{c}{(2)}      &\multicolumn{1}{c}{(3)}   &   &\multicolumn{1}{c}{(4)}      &\multicolumn{1}{c}{(5)}      &\multicolumn{1}{c}{(6)}      \\
\midrule

Tax effect on prices           &               0.008\sym{**} &               0.041\sym{***}&                             &       &        0.008\sym{*}  &               0.045\sym{***}&                             \\
                    &             (0.003)         &             (0.003)         &                      &       &             (0.003)         &             (0.003)         &                             \\
[1em]
$T_{i,m}$           &                             &                             &               0.046\sym{***}& &                            &                             &               0.050\sym{***}\\
                    &                             &                             &             (0.003)         &  &                           &                             &             (0.004)         \\
[1em]
$T_{i,m}$ x $no Shroud_{i,m}$          &                             &                             &              -0.041\sym{***}&    &                         &                             &              -0.043\sym{***}\\
                    &                             &                             &             (0.002)         &    &                         &                             &             (0.002)         \\
[1em]
Constant            &               0.085\sym{***}&               0.084\sym{***}&               0.084\sym{***}&    &           0.092\sym{***}&               0.092\sym{***}&               0.092\sym{***}\\
                    &             (0.004)         &             (0.003)         &             (0.003)         &  &           (0.003)         &             (0.003)         &             (0.003)         \\
\midrule
Time FE             &                 Yes         &                 Yes         &                 Yes         &  &               Yes         &                 Yes         &                 Yes         \\
Agency FE           &                 Yes         &                 Yes         &                 Yes         &  &               Yes         &                 Yes         &                 Yes         \\
League FE           &                 Yes         &                 Yes         &                 Yes         &   &              Yes         &                 Yes         &                 Yes         \\
\hline
Observations        &              927,168         &             1,236,292         &             1,276,400         &  &            698,912         &              768,145         &              778,207         \\
\(R^{2}\)           &               0.577         &               0.713         &               0.741         &   &            0.579         &               0.668         &               0.681         \\
\hline\hline
\end{tabular}
\begin{tablenotes}[flushleft]
\item \emph{Notes:} This table reports the estimated average sports betting tax effects on consumer betting prices for soccer events according to Eq. \ref{eq:tax_avg_effect}, comparing changes in prices before and after the tax reform between German and Non-German agencies. Columns 3 and 6 only consider German agencies that do not shroud taxes. Likewise, columns 1 and 4 (2 and 5) only consider agencies that deduct a 5\% tax surcharge from the advertised winnings (from the advertised wager). Each agency-event combination (unique observation) is equally weighted. Columns 1-3 consider all leagues, while Columns 4-6 exclude events in German leagues for the control group and Non-German leagues for the treatment group. Only agencies for which odds data in all years is available are considered. All estimations include time-, agency- and league-fixed effects. Estimations in Panel A consider observations from soccer matches, and Panel B considers observations from all sports. Robust standard errors clustered at the agency level are reported in the brackets. * denotes significance at the 10-\%, ** at the 5-\%, and *** at the 1-\% level.
\end{tablenotes}
\end{threeparttable}
}
\end{table}

%% file: Tables/DiffinDiff_het_all.tex
\begin{table}[ !ht]\centering
\def\sym#1{\ifmmode^{#1}\else\(^{#1}\)\fi}
\caption{Avg. effect of tax on consumer betting prices - all sports}
    \label{diffindiff_het_allsport}
    \scalebox{0.75}{
    \begin{threeparttable}
\begin{tabular}{l*{7}{c}}
\hline\hline

                 &\multicolumn{3}{c}{All Leagues}     & &\multicolumn{3}{c}{Excl. "cross" leagues}     \\    \cmidrule{2-4} \cmidrule{6-8}
                                  &\multicolumn{2}{c}{Subsamples}     & Interact. & &\multicolumn{2}{c}{Subsamples} &  Interact.  \\    \cmidrule{2-3} \cmidrule{6-7}
                    &\multicolumn{1}{c}{No shroud}      &\multicolumn{1}{c}{Shrouding}      &   &   &\multicolumn{1}{c}{No shroud}      &\multicolumn{1}{c}{Shrouding}      &   \\
                    &\multicolumn{1}{c}{(1)}      &\multicolumn{1}{c}{(2)}      &\multicolumn{1}{c}{(3)}   &   &\multicolumn{1}{c}{(4)}      &\multicolumn{1}{c}{(5)}      &\multicolumn{1}{c}{(6)}      \\
\midrule
Tax effect on prices            &               0.002         &               0.041\sym{***}&                             &  &             0.001         &               0.040\sym{***}&                             \\
                    &             (0.002)         &             (0.003)         &                             &   &          (0.002)         &             (0.003)         &                             \\
[1em]
$T_{i,m}$           &                             &                             &               0.046\sym{***}&    &                         &                             &               0.044\sym{***}\\
                    &                             &                             &             (0.002)         & &                            &                             &             (0.002)         \\
[1em]
$T_{i,m}$ x $no Shroud_{i,m}$          &                             &                             &              -0.045\sym{***}&    &                         &                             &              -0.042\sym{***}\\
                    &                             &                             &             (0.001)         &  &                           &                             &             (0.001)         \\
[1em]
Constant            &               0.084\sym{***}&               0.079\sym{***}&               0.081\sym{***}&    &           0.088\sym{***}&               0.088\sym{***}&               0.088\sym{***}\\
                    &             (0.003)         &             (0.003)         &             (0.003)         &   &          (0.003)         &             (0.003)         &             (0.003)         \\
\midrule
Time FE             &                 Yes         &                 Yes         &                 Yes         &  &               Yes         &                 Yes         &                 Yes         \\
Agency FE           &                 Yes         &                 Yes         &                 Yes         &  &               Yes         &                 Yes         &                 Yes         \\
League FE           &                 Yes         &                 Yes         &                 Yes         &   &              Yes         &                 Yes         &                 Yes         \\
\hline
Observations        &             2,466,796         &             2,959,638         &             3,057,547         &   &          1,901,791         &             2,006,335         &             2,027,511         \\
\(R^{2}\)           &               0.703         &               0.744         &               0.755         &   &            0.713         &               0.740         &               0.742         \\
\hline\hline
\end{tabular}
\begin{tablenotes}[flushleft]
\item \emph{Notes:} This table reports the estimated average sports betting tax effects on consumer betting prices according to Eq. \ref{eq:tax_avg_effect}, comparing changes in prices before and after the tax reform between German and Non-German agencies. Columns 3 and 6 only consider German agencies that do not shroud taxes. Likewise, columns 1 and 4 (2 and 5) only consider agencies that deduct a 5\% tax surcharge from the advertised winnings (from the advertised wager). Each agency-event combination (unique observation) is equally weighted. Columns 1-3 consider all leagues, while Columns 4-6 exclude events in German leagues for the control group and Non-German leagues for the treatment group. All estimations include time-, agency- and league-fixed effects. Estimations in Panel A consider observations from soccer matches, and Panel B considers observations from all sports. Robust standard errors clustered at the agency level are reported in the brackets. * denotes significance at the 10-\%, ** at the 5-\%, and *** at the 1-\% level.
\end{tablenotes}
\end{threeparttable}
}
\end{table}

%% file: Chapters/Standard_tax_incidence.tex
\section{Standard theory of tax pass-through and tax salience}\label{sec:app_standard}

This section presents a partial equilibrium model of taxation based on the textbook model by \citet{kotlikoff1987tax} and the tax salience workhorse model by \citet{chetty2009salience}. It derives the pass-through rates of an excise tax in this model and discusses the implications when market competition becomes imperfect following \citet{weyl2013pass} and \citet{kroft2020salience}.


\subsubsection*{Textbook model of taxation}

Let $q=p-t$ denote the suppliers' price net of taxes. $p$ is the effective tax-inclusive price faced by consumers, and $t$ is the per unit tax. The pass-through rate is defined as $\rho=\frac{dp}{dt}$, i.e., the rate at which consumer prices increase when the tax rises. In the standard model of taxation, based on a simple partial equilibrium setting under perfect competition, the equilibrium condition is given by: $D(p)=S(p-t)$.\footnote{Note that the partial equilibrium analysis presupposes that the market is small relative to the entire economy \citep{kotlikoff1987tax}, which applies to the market of sports betting.} In the textbook model \citep{kotlikoff1987tax}, the economic incidence of a tax is independent of who pays the tax (statutory incidence). Fully differentiating the equilibrium condition with respect to $t$ and solving for the pass-through rate gives us:

\begin{equation*}
\rho=\dfrac{\dfrac{\delta S}{\delta p}}{\left(\dfrac{\delta S}{\delta p} -\dfrac{\delta D}{\delta p}\right)}=\dfrac{1}{\left( 1+\dfrac{\eta_D}{\eta_S}\right)},
\end{equation*}

where $\eta_D=-\frac{\delta D}{\delta p} \frac{p}{D(p)}$ is the \emph{elasticity of demand} and $\eta_S=\frac{\delta S}{\delta p} \frac{p}{S(p-t)}$ is the \emph{elasticity of supply}. Accordingly, the pass-through rate $\rho$ only depends on the supply and demand elasticities. The pass-through is higher when demand is less elastic relative to supply. A full-pass through  ($\rho=1$), i.e., the demand side of the market bears the entire burden of the tax, implies that either supply is perfectly elastic or demand is perfectly inelastic. The total burden of the tax is shared between consumers and suppliers ($0\leq \rho\leq1$), i.e., there is no "over-shifting"  of the tax onto consumers. 

If competition is imperfect, "over-shifting" of a tax is possible in equilibrium as the convexity of the demand and supply curves becomes important for the pass-through of a tax. Suppose a monopolist who chooses quantity $w$ to maximize following profit function: $\Pi=(p(w)-t)w-c(w)$, where $p(w)$ is the inverse demand and  $c(w)$ the monopolist's cost function. In such a setting, the pass-through of a tax is given by \citep{weyl2013pass}: 

\begin{equation*}
\rho=\dfrac{1}{\left( 1+\dfrac{\eta_D-1}{\eta_S}+\dfrac{1}{\eta_{ms}}\right)},
\end{equation*}

where $ms=-\frac{\delta p(w)}{\delta w} w$ is the negative of the marginal consumer surplus and $\eta_{ms}= \frac{ms}{\frac{\delta ms}{\delta w} w}$ is the elasticity of the inverse marginal surplus function. Consequently, the smaller $\eta_{ms}$, which can be understood as the curvature (of the logarithm) of demand \citep{weyl2013pass}, the higher the pass-through onto consumers ceteris paribus. Pass-through is larger than one if and only if $\eta_{ms}$ is negative \citep{bulow1983note}, i.e. demand  is very convex \citep{pless2019pass}. This result in monopolies carries over to oligopolistic competition models, with the difference that the degree of market competition and market power becomes an additional determining factor for pass-through rates \citep[see][for a thorough discussion]{weyl2013pass}. In general, the pass-through of taxes is higher when markets are less competitive and vice versa. This connection to market competition and market power has important implications not only for questions in public finance but also for various topics in industrial organization \citep{ pless2019pass, miravete2020one}. For instance, over-shifting of taxes can be used as a test for market power in different markets \citep{pless2019pass}.


\subsubsection*{Tax salience}

The standard theory of taxation assumes that consumers are perfectly informed about taxes and that they optimize fully with respect to tax-inclusive effective prices. The tax (surcharge) elasticity is assumed to be equivalent to the elasticity with respect to the (posted) price. Thus, the tax incidence in the standard model, both under perfect and imperfect competition, is independent of whether consumer price changes are due to changes in the base price or the tax surcharge. Numerous studies suggest that this assumption is often violated and that the incidence also depends on the salience of a tax or the way taxes are presented \citep[e.g.][]{chetty2009salience}. Suppose consumers are subject to limited attention and the tax surcharge is not fully salient. In that case, consumers may react differently to changes in displayed prices compared to changes in the tax surcharge. Consequently, the tax incidence and theoretical welfare implication of the tax may change.
 
In relation to the "shrouding" policies by betting agencies presented in section \ref{sec:tax_shrouding}, one can decompose effective consumer prices in a tax surcharge set by suppliers, $\tau$, and advertised consumer prices, $\Tilde{p}$: $p=\Tilde{p}+\tau$. Accordingly, supplier prices can be rewritten to: $q=\Tilde{p}-t+\tau$. For the sake of simplicity, assume for now that the tax surcharge is always equal to the tax ($\tau=t$), implying $q=\Tilde{p}$ and $p=\Tilde{p}+t=q+t$.
 
Following the tax salience literature \citep{chetty2009salience}, I assume that the consumers only perceive a fraction $0\leq \psi$ of the actual tax surcharge ($\tau=t$). In contrast to the theoretical model in Section \ref{sec:theory_sin}, I follow the tax salience literature here and assume that the salience parameter is exogenously given and not subject to strategic "shrouding" considerations by firms. The perceived price from the consumers' perspective is thus given by: $p_{\psi}=q+\psi t$. $\psi$ defines the extent of underreaction to the tax, which can be interpreted as the degree of tax salience or inattention. As in the model of \citet{chetty2009salience}, $\psi$ is determined the ratio of the tax ($\eta_{D, p_{\psi}|t}$) and price elasticity ($\eta_{D, p_{\psi}|q}$). The equilibrium condition is given by $D(q + \psi t)=S(q)$. By the implicit function theorem, we can derive the following incidence on producers: 

\begin{equation*}
\dfrac{dq}{dt}=\dfrac{ \dfrac{ \delta D}{\delta p_{\psi}|t} }{\left( \dfrac{\delta S}{\delta q} - \dfrac{\delta D}{\delta p_{\psi}|q}  \right)}=\dfrac{\psi \dfrac{ \delta D}{\delta p_{\psi}|q} }{\left( \dfrac{\delta S}{\delta q} - \dfrac{\delta D}{\delta p_{\psi}|q}  \right)}=\dfrac{\psi   \eta_{D, p_{\psi}|q}}{\left( \dfrac{p}{q} \eta_{S, q} + \eta_{D, p_{\psi}|q}  \right)}
\end{equation*}

Accordingly, the incidence on consumers is equal to:

\begin{equation*}
\dfrac{dp}{dt}=1+\dfrac{dq}{dt}=\dfrac{\left( \dfrac{p}{q} \eta_{S, q} + (1-\psi)\eta_{D, p_{\psi}|q}  \right)}{\left( \dfrac{p}{q} \eta_{S, q} + \eta_{D, p_{\psi}|q}  \right)}
\end{equation*}

Consequently, the pass-through of a tax is determined not only by the supply and demand elasticities but also by the salience parameter $\psi$. The higher the salience of a tax, the lower the incidence on consumers. If the consumers are entirely inattentive to the tax ($\psi=0$), the incidence falls completely on consumers. We are back to the standard model if the tax is fully salient ($\psi=1$). 

\citet{bradley2020hidden} derives the tax incidence under tax salience in a monopoly and \citet{kroft2020salience} generalizes the tax salience model to imperfect competition models along the lines of \citet{weyl2013pass}. Under imperfect competition, tax salience's effect on pass-through rates depends on the market structure and the curvature of demand. In special cases, higher tax salience can even increase the incidence on consumers.

%% file: Chapters/proofs.tex
\subsection{Perfect competition}\label{sec:app_proof_perfcomp}

\begin{lemma}\label{lem:app_proof_perfcomp}

Suppose $\gamma<1$ and the homogeneous consumers optimally chooses $\Tilde{x}^*(t)$ according to her decision utility Eq. \ref{eq:dec_util}. Then $t^*=\dfrac{(1-\gamma)c_x(x^*)}{\theta}$ financing a uniform lump sum transfer $W^*$ back to the consumers would implement the first-best solution $(x^*, z^*)$. $\theta=1$ represents the standard case with no underreaction.
\end{lemma}

\begin{proof}

$\Tilde{x}^*(t)$ satisfies $v_x(\Tilde{x}^*(t))-\gamma c_x(\Tilde{x}^*(t))=m+\theta t$. $t^*=\frac{(1-\gamma)c_x(x^*)}{\theta}$ implies, according to the first-order condition, that $v_x(\Tilde{x}^*(t))-\gamma c_x(\Tilde{x}^*(t))=m+\frac{\theta(1-\gamma)c_x(x^*)}{\theta}$, which can be rewritten to $v_x(\Tilde{x}^*(t))-c_x(x^*)-m=\gamma(c_x(\Tilde{x}^*(t))c_x(x^*))$ and is satisfied if and only if $\Tilde{x}^*(t)=x^*$ and consequently $\Tilde{z}^*(t)=z^*$

\end{proof}

\subsection{Imperfect competition}\label{sec:app_proof_imperfcomp}

\begin{lemma}\label{lem:nopure} In the case with both types of consumers ($0<\lambda<1$) and $\theta<1$, there is no pure strategy equilibrium where no firm shrouds. If $s>0$, there is no symmetric pure strategy equilibrium.
\end{lemma}

\begin{proof}
If all firms do not shroud taxes ($\tau_k=0$), Bertrand competition implies $p^*_k=p^s_k=m+t$. In this case, a firm can profitably deviate by shrouding taxes $\tau_k=t$ and setting a slightly lower posted price, implying that the firm serves all inattentive consumers for a profit because of underreaction to the tax ($\theta<1$). Suppose all firms shroud taxes $\tau_k=t$ and $s>0$. Bertrand competition implies that posted prices will equal $m$ and firms' profits are again equal to zero. In that case, a firm can make positive profits by unshrouding taxes and serving all attentive consumers' at higher than marginal cost prices if $s>0$. 
\end{proof}

\begin{lemma}\label{lem:segment}
If $N\geq4$, there is an asymmetric equilibrium in pure strategies with two market segments: i) at least two firms shroud taxes, choose $\{p^s_k=m, \tau_k=t\}$ and serve all inattentive consumer; and ii) at least two firms do not shroud taxes, choose $\{p^s_k=m+t, \tau_k=0\}$ and serve all attentive consumers.
\end{lemma}

\begin{proof}
The equilibrium outcomes and prices in each market segment are equivalent to the boundary cases discussed above. The shrouding segment will only be populated by inattentive consumers and the non-shrouding segment will only be populated by attentive consumers. According to Bertrand competition, we have marginal cost pricing in both segments and profits are equal to zero in both segments, implying that no firm can profitably deviate.
\end{proof}